\documentclass[11pt,reqno]{amsart}

\usepackage{lettrine}

\usepackage{amsmath}
\usepackage{amsthm}
\usepackage{amsfonts}
\usepackage{amssymb}

\usepackage{dsfont}
\usepackage{stmaryrd}

\usepackage{geometry}
\usepackage{fixmath}

\usepackage{upgreek}
\usepackage{bm}

\usepackage{graphicx}
\usepackage{fullpage}

\usepackage{color} 
\usepackage{float}

\usepackage{enumitem}

\usepackage{algorithm}
\usepackage{algpseudocode}

\usepackage{thmtools}
\usepackage{thm-restate}

\usepackage{hyperref}

\usepackage{cleveref}

\crefname{lem}{Lemma}{Lemmas}
\crefname{section}{Section}{Sections}
\crefname{lemma}{Lemma}{Lemmas}
\crefname{thm}{Theorem}{Theorems}
\crefname{corollary}{Corollary}{Corollaries}
\crefname{theorem}{Theorem}{Theorems}
\crefname{defn}{Definition}{Definitions}
\crefname{definition}{Definition}{Definitions}
\crefname{fact}{Fact}{Facts}
\crefname{figure}{Fig.}{Figures}
\crefname{clm}{Claim}{Claims}
\crefname{claim}{Claim}{Claims}
\crefname{prop}{Proposition}{Propositions}
\crefname{proposition}{Proposition}{Propositions}
\crefname{algocf}{Algorithm}{Algorithms}

\newtheorem{theorem}{Theorem}[section]
\newtheorem{lemma}[theorem]{Lemma}
\newtheorem{proposition}[theorem]{Proposition}
\newtheorem{claim}[theorem]{Claim}
\newtheorem{corollary}[theorem]{Corollary}
\newtheorem{definition}[theorem]{Definition}

\theoremstyle{definition}
\newtheorem{condition}[definition]{Condition}

\newcommand{\cA}{{\mathcal{A}}}
\newcommand{\cB}{{\mathcal{B}}}
\newcommand{\cC}{{\mathcal{C}}}
\newcommand{\cD}{{\mathcal{D}}}
\newcommand{\cE}{{\mathcal{E}}}
\newcommand{\cF}{{\mathcal{F}}}
\newcommand{\cH}{{\mathcal{H}}}

\newcommand{\cK}{{\mathcal{K}}}
\newcommand{\cL}{{\mathcal{L}}}
\newcommand{\cN}{{\mathcal{N}}}
\newcommand{\cM}{{\mathcal{M}}}
\newcommand{\cP}{{\mathcal{P}}}

\newcommand{\cR}{{\mathcal{R}}}
\newcommand{\cS}{{\mathcal{S}}}
\newcommand{\cT}{{\mathcal{T}}}

\newcommand{\G}{\mathbold{G}}

\newcommand{\bI}{{\mathbold{I}}}
\newcommand{\bJ}{{\mathbold{J}}}

\newcommand{\bX}{{\mathbold{X}}}
\newcommand{\bY}{{{\mathbold Y}}}

\newcommand{\degr}{\mathrm{deg}}

\newcommand{\Exp}{\mathbb{E}}

\newcommand{\Ind}{{\mathds{1}}}

\newcommand{\Inf}{ {\rm \Upgamma}}
\newcommand{\bInf}{\bm{\Upgamma}}

\newcommand{\compWeight}{\Upupsilon}
\newcommand{\WA}{{\UpA}}
\newcommand{\WB}{{\tt W}}

\newcommand{\pure}{{locally simple}}

\newcommand{\cappedWA}{\UpM}
\newcommand{\cappedWB}{{\WB}}

\newcommand{\vecJ}{\mathcal{J}}
\newcommand{\bvecJ}{\bm{\mathcal{J}}}

\DeclareSymbolFont{EuclidLetter}{U}{eur}{m}{n}
\DeclareMathSymbol{\UpA}{\mathord}{EuclidLetter}{65}
\DeclareMathSymbol{\UpB}{\mathord}{EuclidLetter}{66}
\DeclareMathSymbol{\UpJ}{\mathord}{EuclidLetter}{74}
\DeclareMathSymbol{\UpM}{\mathord}{EuclidLetter}{77}
\DeclareMathSymbol{\UpS}{\mathord}{EuclidLetter}{83}
\DeclareMathSymbol{\UpT}{\mathord}{EuclidLetter}{84}

\DeclareSymbolFont{EuclidLetter}{U}{eur}{m}{n}
\DeclareMathSymbol{\UpB}{\mathord}{EuclidLetter}{66}
\DeclareMathSymbol{\UpJ}{\mathord}{EuclidLetter}{74}
\DeclareMathSymbol{\UpT}{\mathord}{EuclidLetter}{84}

\newcounter{proptyno}

\newcounter{kcomcount}
\newcommand{\IntrstGraphFam}{\mathcal{G}}
\newcommand{\Tmix}{{\rm T}_{\rm Mix}}

\newcommand{\distance}{{\tt dist}}
\newcommand{\EdgeBlockWeight}{\distance}

\newcommand{\impV}{\mathrm{W}}

\newcommand{\One}{\cA}
\newcommand{\Two}{\cL}

\newcommand{\In}{\mathrm{In}}
\newcommand{\Out}{\mathrm{Out}}

\newcommand{\BTree}{\widehat{T}}
\newcommand{\BBTree}{\overline{T}}
\newcommand{\BRest}{\widehat{R}}
\newcommand{\BCDist}{q}

\newcommand{\hatC}{\widehat{C}}

\newcommand{\GWT}{\pmb{\UpT}}
\newcommand{\ST}{\pmb{T}}

\title{On sampling diluted Spin Glasses \\using  Glauber dynamics} 

\author{Charilaos Efthymiou and Kostas Zampetakis}

\date{\today}

\begin{document}

\maketitle

\begin{abstract}%
{\em Spin-glasses}  are natural  Gibbs distributions that have  been studied in 
theoretical computer science for  many decades. Recently,   they have been gaining renewed 
attention from the community  as they emerge naturally in 
{\em neural computation} and {\em learning},  {\em network inference}, 
{\em optimisation}  and many other areas.

Here we consider the {\em Edwards-Anderson} spin-glass  distribution at inverse temperature $\beta$ when  the underlying  graph 
is   an instance of $\G(n,d/n)$.  This is a random graph on $n$ vertices such that  each edge appears independently with probability $d/n$, 
where the expected degree $d=\Theta(1)$. 
We study the problem of efficiently  sampling from the aforementioned distribution using the well-known Markov chain called 
{\em Glauber dynamics}.

For a certain range of $\beta$, that depends only on the expected degree $d$ of the graph,
and   for  typical instances of the Edwards-Anderson model on $\G(n,d/n)$,
we show that the corresponding (single-site) Glauber dynamics exhibits mixing time 
$O(n^{2+\frac{3}{\log^2 d}})$. 

The range of   $\beta$  for which we obtain  our rapid-mixing result corresponds to  the expected influence  being smaller than $1/d$. This bound is very natural, and we conjecture that it is 
the best  possible for rapid mixing.

As opposed to the {\em mean-field} spin-glasses, where the Glauber dynamics  has been studied before,  less is known 
for the {\em diluted} cases like the one we consider here.  The latter problems  are more  challenging to work with because
the corresponding instances involve {\em two levels} of randomness, i.e., the random graph instance and the random Gibbs 
distribution.

We establish our results by utilising the well-known {\em path-coupling} technique. 
In the standard setting of Glauber dynamics on $\G(n,d/n)$ one has to deal 
with the so-called effect of high degree vertices. 
Here, with the spin-glasses, rather than considering  vertex-degrees, it is more natural 
to use a different measure on the vertices of the graph, that  we call 
{\em aggregate influence}.

We  build  on the block-construction approach proposed by  
[Dyer, Flaxman, Frieze and Vigoda: 2006]   to circumvent the  problem 
with the high degrees in the  path-coupling analysis.
Specifically, to obtain our results, we first  establish rapid mixing for
an  appropriately defined block-dynamics. 
We design this dynamics such that vertices of large aggregate
influence are placed deep inside their blocks. Then, we  obtain  rapid mixing   
for the (single-site) Glauber dynamics  by utilising  a comparison argument.

\end{abstract}

\section{Introduction}

{\em Spin-glasses}  are natural, high dimensional Gibbs distributions that have  been studied in theoretical computer science for 
many decades. Recently, they have been gaining renewed attention from the community  as they emerge naturally in  {\em neural computation} and {\em learning}, e.g., in the Hopfield model,  as models of  {\em network inference}, e.g., in the stochastic block model, in 
{\em optimisation}, {\em counting-sampling} and many other areas, e.g.,  see \cite{SteinNewmanSpinGlassBook, GamarnikJWFOCS20,CoEfJKKCMI,OptElAlaoui, KoehLRColt22,eldan2022spectral,EfthICALP22}.  

Furthermore, spin-glasses are widely considered to be canonical models of  {\em extremely} disordered systems \cite{mezard1990spin,SteinNewmanSpinGlassBook} and, as such,  they  have beed studied extensively  
in mathematics \& mathematical physics, e.g., see 
\cite{TalagrandAnnals,panchenko2013parisi,franz2001exact,guerra2004high},   and also in  statistical physics \cite{mezard1990spin,SteinNewmanSpinGlassBook,RSBParisi}.  In particular, as far as physics is concerned,  spin-glasses have been studied intensively since the early '80s, while the seminal, groundbreaking work of  Giorgio Parisi  on  spin-glasses got him  the  Nobel Prize in Physics in 2021.

In this work, we consider the natural problem of sampling from these distributions. 
To this end, our endeavour is to employ  the powerful Markov Chain Monte Carlo  (MCMC) method.
Typically, the MCMC sampling algorithms are  very simple to describe and implement in practice, however,
analysing their performance can  be extremely  challenging.

We focus on one of the best-known cases of spin-glasses, called   the {\em Edwards-Anderson model} (EA model) 
\cite{EAIsingIntroWork}. Given a fixed graph $G=(V,E)$,  and the vector $\bm{\vecJ}=\{\bJ_e: e\in E\}$ of independent, 
identically distributed (i.i.d.) standard Gaussians, the EA model with {\em inverse temperature} $\beta>0$,  corresponds 
to  the {\em random} Gibbs  distribution  $\mu=\mu_{G,\bm{\vecJ}, \beta}$  on the configuration space $\{\pm 1\}^V$  such that 
\begin{align*}
{\mu}(\sigma) &\propto  \exp\left( \beta\sum\nolimits_{\{w,u\}\in E} {\Ind}\{\sigma(u)=\sigma(w)\} \cdot \bJ_{\{u,w\}}   \right) \enspace,
\end{align*}
where $\propto$ stands for ``proportional to".

In the literature,  we also encounter this distribution   as the {\em Viana-Bray} model (e.g. see \cite{guerra2004high}), or 
as the {\em 2-spin model} (e.g., see \cite{PanchenkoTalagrand,GamarnikJWFOCS20}).  It is worth mentioning at this 
point that the Edwards-Anderson model on the complete graph corresponds to  the well-known {\em Sherrington-Kirkpatrick} 
model (SK model) \cite{SKModel}.

On a first account, the EA model may look innocent, i.e., it looks similar to the standard Ising model with just  the addition 
of the Gaussian couplings.  It turns out, though, that it is a {\em fascinating} distribution with a lot of intricacies, while the 
configuration space has  extremely rich   structure, e.g.  in various settings, it is conjectured to exhibit the ``{\em infinite} 
Replica Symmetry Breaking" \cite{mezard1990spin}.

We use  the  {\em Glauber dynamics} to sample from the  Edwards-Anderson 
spin-glass.  We assume that the underlying geometry is an  instance of the {\em sparse}  random graph $\G(n,d/n)$.  
This is a random graph  on $n$ vertices, such that each edge appears, independently, with probability $p=d/n$,  for
some constant $d>0$.   Note that we obtain an instance of the problem by  first drawing the underlying graph form the distribution  $\G(n,d/n)$ 
and then, given the  graph,  we generate the random Gibbs distribution.

Sampling from Gibbs distributions induced by
instances of $\G(n, d/n)$, or, more generally, instances of so-called random constraint satisfaction
problems, is at the heart of recent endeavors to investigate connections between 
{\em phase transitions} and the efficiency of algorithms, e.g. \cite{OptasOghlan08,alaoui2020algorithmic,COghlanEfth11,GalStefVigJACM15,GamSudan17,SlySun12}. 
The MCMC sampling problem on random graphs has garnered a lot of attention, e.g., see 
\cite{dyer2006randomly,BezakovaGGS22,EftFeng23,mossel2010gibbs,EfthymiouHSV18,DyerFriez10,ChManMo23}, as it is considered to be an intriguing case to study. 
So far the focus has been on sampling standard Gibbs distributions, which, already, is a very challenging problem. 
The study of spin-glasses takes us a step further.
Working with the EA model, we introduce an
{\em extra level of disorder} which is due to the random couplings at the edges of the graph.
Hence, in our analysis, we need to deal with the disorder of both $\G(n,d/n)$ and   the random couplings.

For typical instances of {\em both} the random graph $\G(n,d/n)$ and the EA model,   we show that the Glauber dynamics exhibits mixing time ${\textstyle O(n^{2+\frac{3}{\log^2 d}})}$ 
for  any inverse temperature  $\beta>0$ such that  
\begin{align}\label{eq:IntroUniqueness}
\Exp\left[\tanh\left({\textstyle \frac{\beta}{2}} \left| \bJ \right |\right) \right]<1/d\enspace,
\end{align}
where the expectation is with respect to the standard Gaussian random variable $\bJ$.
It is elementary to show that for large $d$, the  above condition corresponds to having $0<\beta<\frac{\sqrt{2\pi}}{d}$. In turn, this implies that the {\em expected influence} is smaller than $1/d$,  a quite natural requirement to have. 
We  {\em conjecture} that the  region for  $\beta$ in which we establish the 
 rapid-mixing of Glauber dynamics  is  the {\em best possible}.

We also believe that the bound we obtain on the mixing  time is very close to the optimal one. 
Our conjecture is that the   mixing time for Glauber dynamics, in the aforementioned range of $\beta$,  
is $n \exp(\Theta(\sqrt{\log n}))$.
This  is  because we 
typically have  isolated stars with rather large couplings at their edges. Note that we expect  to see couplings of magnitude as large as $\Theta(\sqrt{\log n})$.

To the best of our knowledge, not much is known about the  mixing time of Glauber dynamics for the 
EA model on $\G(n,d/n)$. 
There is only a weaker, non-MCMC, sampler for this distribution,  proposed in~\cite{EfthICALP22}.
 The  error at the output for this algorithm is a vanishing function of $n$.
Here, with the Glauber dynamics,  we obtain the standard approximation guarantees one gets from the  MCMC samplers,
which are considered to be the best possible.  

As opposed to the MCMC samplers for the EA model on $\G(n,d/n )$, more is known about  the {\em mean-field}  SK model,  with the most recent works being \cite{AlaouiMSFOCS22,bauerschmidt2019very,KoehLRColt22,eldan2022spectral}. 
For the cases we consider here, these results imply a much weaker bound on $\beta$ compared to what  we get with our approach, i.e., 
they require $\beta$ to be a vanishing function of $n$\footnote{More concretely, it is elementary to verify that
they require $\beta=O\left(\sqrt{\frac{\log\log n}{\log n}}\right)$.}.
This comes without surprise as the SK model lives on the complete graph, hence there is {\em no geometry} in the problem. 
On the other hand, the  EA model on $\G(n,d/n)$  has a very rich, and as it turns out, intricate geometry. 
Overall, the results for the SK-model fare poorly because they cannot address the geometric challenges that emerge in the EA model on $\G(n,d/n)$.

Our analysis does not rely on the newly introduced {\em Spectral Independence} (SI) method 
\cite{OptMCMCIS,ChenLVStoc21,Chen0YZFOCS22},   even though there  are rapid mixing results 
for Gibbs distributions on $\G(n,d/n)$  that utilise SI, e.g., see \cite{BezakovaGGS22,EftFeng23}. 
When it comes to spin-glasses, the natural quantities that arise with SI turn out to be  too complicated to work with.
Our approach exploits the classic {\em path-coupling}  
technique  \cite{bubley1997path}. 
Specifically, we build on the machinery developed in the sequence of results in \cite{dyer2006randomly,EfthymiouDA14,efthymiou2018sampling}
that establish  fast mixing  for Glauber dynamics on (standard) Gibbs distributions on $\G(n,d/n)$. 
Basically,  our approach builds on these ideas so that we can  accommodate the extra disorder that the problem exhibits. 
Note that these previous works
are about  distributions such as the colouring model, the hard-core model, etc., and not for spin-glasses.

For typical instances of $\G(n,d/n)$  all but an $\exp(-\Theta(d))$ fraction
of  the vertices are of degree close to $d$, while we expect to have degrees as huge as
$\Theta\left(\frac{\log n}{\log\log n}\right)$. In that respect, it is natural to have the
parameters of the problem expressed in terms of  the expected degree $d$, 
rather than, e.g.,  the maximum degree. Hence, a major challenge in the analysis is how to 
circumvent the so-called ``effect of high degrees". 

Roughly speaking, the approach underlying \cite{efthymiou2018sampling} 
is as follows: 
Rather than considering (single site) Glauber dynamics, we  consider 
{\em block dynamics}.  That is,    there is an appropriately constructed  block-partition of the set of  
vertices  such that, at each step, 
we update the  configuration of a randomly  chosen block.  
As it was already observed in \cite{dyer2006randomly},  the typical instances of $\G(n,d/n)$ admit 
a  block-partition, such that the high-degree vertices are  hidden  deep inside the 
blocks,  in a way that  makes their effect vanish.  
This allows one to circumvent the problem that the high-degree vertices pose in the 
path-coupling analysis and   show  fast  mixing  for the block dynamics. Subsequently,  
one obtains  the bounds on the mixing time for the single-site  dynamics 
by  using  {\em comparison}.

A primary challenge of the above approach is to construct the desirable block-partition.
 The contribution coming from the aforementioned works amounts to introducing  a {\em weighting-scheme}  
 (a set of potentials) for the paths in the graph, which is further leveraged for the block construction. Typically, this approach entails heavy probabilistic analysis.

It turns out that the use of  potentials for the paths is  quite  natural in our setting, too.
One of our main contributions  here is to  introduce  {\em new} weights (new potentials) for the paths
of the graph, which  accommodate the richer structure of the problem.  
 Unlike the previous works that focus only vertex-degrees, we utilise concepts from Spectral Independence and introduce a new measure for the vertices of the graph, which we call {\em aggregate influence}

Overall, getting a handle on the behaviour of the weights of the paths is one of the most technically demanding parts of this work. 
It is worth mentioning that the set of blocks we obtain here is quite different than the one appears in \cite{efthymiou2018sampling}.   
Note that in this work, the  block partition is such that vertices with large degree are hidden deep inside the blocks.   
Here, it is typical to have single-vertex blocks consisting of  a high-degree vertex,  or having multi-vertex blocks whose vertices are all  low-degree ones.  
This comes without surprise, as   the notions of the aggregate influence and the degree of a vertex are  different from each other.

\subsection{Results}\label{sec:MainResults}

We let $\G=\G(n, d/n)$ be the Erd\H{o}s–R\'enyi graph on a set $V_n$ of $n$ vertices, with edge probability $d/n$, where $d>0$
is a fixed number.
The Edwards-Anderson model on $\G$ at \emph{inverse temperature} $\beta >0$,   is defined as follows: 
for  $\bm{\vecJ} = \{ \bJ_e :e \in E(\G)\}$ a family of independent,  {\em standard Gaussians}, and for 
$\sigma \in \{\pm1\}^{V_n}$, we let 
\begin{align}\label{eq:kSpin1}
	\mu_{{\G}, \bm{\vecJ},\beta}(\sigma)&
    =\frac{1}{Z_\beta(\bm{G}, \bm{\vecJ})}
    \cdot \exp\left( \beta \sum_{x\sim y} {\Ind}\{\sigma(y)=\sigma(x)\} \cdot {\bJ}_{\{x,y\}}  \right) \enspace,
\end{align}
where
	\begin{align*}
	Z_\beta(\G, \bm{\vecJ}) &
    =\sum_{\tau\in\{\pm1\}^{V_n}} \exp\left( \beta \sum_{x\sim y}{\Ind}\{\tau(y)=\tau(x)\}\cdot \bJ_{\{x,y\}}  \right)\enspace.
	\end{align*}
Typically, we study this distribution as $n\to \infty$.

We use the discrete time,   (single site)   {\em Glauber dynamics}   $(X_t)_{t \geq 0}$ to approximately sample from  
the aforementioned distributions $\mu=\mu_{\G, \bm{\vecJ},\beta}$. 
  Glauber dynamics is a Markov chain with  state space the  support of the distribution $\mu$.
We assume that the chain   starts from an arbitrary configuration $X_0\in \{\pm 1\}^{V_n}$. For 
$t\geq 0$, the transition from the state $X_t$ to $X_{t+1}$ is according to the  following steps: 
\begin{enumerate}
\item choose uniformly at random a vertex $v$,
\item for every vertex $w$ different than $v$, set $X_{t+1}(w)=X_t(w)$,
\item set $X_{t+1}(v)$ according to the marginal of $\mu$ at $v$, conditional on 
the neighbours  of $v$ having the configuration  specified by $X_{t+1}$.
\end{enumerate}

It is standard to show that when a Markov chain satisfies a set of technical conditions called 
{\em ergodicity}, then it converges to a unique stationary distribution. For the cases 
we consider here, the Glauber dynamics  is trivially ergodic.

Let  $P$ be the transition matrix of an  ergodic Markov chain $(X_t)$  with a finite state space 
$\Omega$ and equilibrium distribution $\mu$. For  $t\geq 0$, and $\sigma\in \Omega$, let 
$P^{t}(\sigma, \cdot)$ denote the distribution of $X_t$, when the initial state of the chain 
satisfies $X_0=\sigma$.  The  {\em mixing time} of the Markov chain  $(X_t)_{t \geq 0}$ is 
defined by
\begin{align*}
\Tmix &=  \max_{\sigma\in \Omega}\min { \left\{t > 0 : \Vert P^{t}(\sigma, \cdot) - \mu \Vert_{\rm TV} \leq{1}/{\mathrm{e}} \right \}}\enspace .
\end{align*}
Our focus is  on the mixing time of the aforementioned Markov chain. 

Finally, for $k>1$ we let 
\begin{align}\label{def:UniqRegion}
 \beta_{c}(k)=\frac{\sqrt{2\pi}}{k}\enspace.
\end{align}

\begin{theorem}\label{thrm:MainResult}
For any $\varepsilon\in (0,1)$, there exists $d_0=d_0(\varepsilon) \ge 1$, such that for  $d\ge d_0$, for
$\beta \leq (1-\varepsilon)\beta_c(d)$, there is a constant $C>0$ such that the following is true: 

Let $\G=\G(n,d/n)$, while let $\mu$ be the Edwards-Anderson model on $\G$, at
inverse temperature $\beta$. 
Then, with probability $1- n^{-1/4}$ over the instances of $\G$ and $\mu$, the Glauber 
dynamics on $\mu$ exhibits mixing time
\begin{align*}
 \Tmix\leq C \cdot n^{2+\frac{3}{\log^2 d}} \enspace.
\end{align*}
\end{theorem}

A couple of remarks are in order.
First, in Theorem \ref{thrm:MainResult}, 
we write  $\beta<\beta_c(d)$ to specify  the region  that we have rapid mixing. 
This condition is  equivalent to  what we have in \eqref{eq:IntroUniqueness}, since we assume sufficiently large $d$.
Using  $\beta<\beta_c(d)$ instead of \eqref{eq:IntroUniqueness}, leads to  cleaner  derivations. 

Second, we note that Theorem \ref{thrm:MainResult} continues to hold if we replace the Gaussian couplings with couplings drawn from an arbitrary {\em sub-Gaussian} distribution. To see this, notice that if $J^\prime$
is sub-Gaussian, and $J$ is a zero-mean Gaussian, there exist a constant $c\ge 0$, such that for all $t>0$ we have that $\Pr[|J^\prime |\ge t] \le c\cdot \Pr[|J|\ge t]$. Therefore, for any (using coupling) any sub-Gaussian model to the EA-model we consider here.
\section{Approach}\label{sec:Approach}\label{sec:OverviewBlockConstr}

In this section, we give a high-level overview of the construction we use to prove Theorem~\ref{thrm:MainResult}.  Note that the construction relies on introducing a good number of potentials.

\subsection{Aggregate Influence} Consider the  graph $\G=\G(n,d/n)$, and the vector of real numbers $\bvecJ = \{\bJ_e : e \in E(\G)\}$,
where $\bJ_e$'s are i.i.d. standard Gaussians, i.e., $\cN(0,1)$. 
For  $\beta>0$, consider also the Edwards-Anderson model  $\mu=\mu_{\G,\bvecJ, \beta}$ on $\{\pm 1\}^{V_n}$, that is, 
 \begin{align}
\mu(\sigma)&\propto \exp\left( \beta \sum\nolimits_{x\sim y}{\Ind}\{\sigma(y)=\sigma(x)\}\cdot \bJ_{\{x,y\}}  \right), & \forall \sigma\in \{\pm 1\}^{V_n}\enspace.
 \end{align}
We usually refer to each $\bJ_e$ as the \emph{coupling} at the edge $e$ in the graph.

For each edge $e\in E(\G)$ consider the
{\em influence} $\Inf_e$ defined by
\begin{align}
\Inf_e&= \frac{\left|1-\exp\left( \beta \bJ_e  \right) \right|}{1+\exp\left( \beta \bJ_{e}\right) }\enspace.
\end{align}
A natural way of viewing the influence   $\Inf_e$ is as a measure of  correlation decay over the edge $e$.  
Note that the influence is a very natural quantity to consider  and emerges in many different contexts such as 
 \cite{ChenLVStoc21,EfthICALP22,EfthZamp2302}, just to mention a few.  

To get some intuition, note that  $\Inf_e\in [0,1]$ is an increasing function of  $|J_e|$.  
For a ``heavy edge",   i.e., an edge $e$ that $|J_e|$  is ``big", 
  the corresponding influence is  close to 1. 
On the other hand, for a ``light edge", i.e., when $|J_e|$ is ``small",  the corresponding influence is  close to 0.

Given the set of  influences $\{\Inf_e\}_{e\in E(\G)}$, for every vertex $w\in V(\G)$, we define 
the {\em aggregate influence}  such that 
\begin{align}
\WA(w) &= \sum\nolimits_{{z \sim w }} \Inf_{\{w,z\}}\enspace.
\end{align}
The quantities $\WA(w)$'s play a key role in the analysis, as they are used in  our construction
of blocks.

The choice of parameters in Theorem \ref{thrm:MainResult}, 
implies that 
for each vertex $w$, we have that $\Exp[\WA(w)]=(1-\varepsilon)$,  
where the expectation is with respect to both the degree of $w$ and 
the couplings at its incident edges.

In Theorem \ref{theorem:TailBound4WA}, in the appendix, we show that  $\WA(w)$ is well-concentrated, i.e.,  
for any fixed  $\delta>0$, we have that 
\begin{align}
\Pr[\WA(w)> \Exp[\WA(w)]+\delta ] \leq \exp(-\Omega(d))\enspace.
\end{align}

In our analysis, we would have liked that each vertex in $\G$ has aggregate influence $<1$. 
From the above tail bound we expect to have many, but relatively rare, heavy vertices in  $\G$, i.e., vertices
with  aggregate influence $>1$. Note that, we typically have ${\rm poly}(n)$  many vertices each of aggregate influence as huge as $\Theta(\log^{4/3} n)$.

\subsection{Path Weights}

We introduce a weighting-scheme for the vertices of the graph $\G$
that uses $\WA(w)$'s.
There are parameters $d,\delta>0$ in the scheme, where  $d>0$ is a large number, 
while $\delta\in (0,1)$. Each vertex $w$ is assigned  weight $\cappedWA(w)$ defined by
\begin{align}\label{eq:expLambdaW}
\cappedWA(w) = \begin{cases}
                1-\delta/2  & \text{ if } \WA(w) \le 1-\delta,  \\
                d \cdot \WA(w) & \text{ otherwise. }
                \end{cases}
                \enspace 
\end{align}

Given the above  weights for the vertices, we define a weight for each path in $\G$. Specifically, the 
path  $P=(v_0, \ldots, v_\ell)$ is assigned weight $\cappedWA(P)$, defined by
\begin{equation*}
\cappedWA(P) = \prod_{i = 0}^{\ell} \cappedWA(v_{i})\enspace.
\end{equation*}
%
Subsequently, we introduce the notion of  {\em block-vertices} in $\G$. A vertex $w$ is called block-vertex if
every path  $P$ that emanates from $w$ is ``light", i.e., it has weight $\cappedWA(P) < 1$. 
Intuitively, $w$ being a block-vertex implies that $\WA(w)<1$, while  
every heavy vertex $v$, i.e., having $\WA(v)>1$, needs to be far from $w$.

For the range of the parameters we consider here, it turns out that there is a plethora
of block-vertices in $\G$. Specifically, we show the following result (Theorem~\ref{thm:AllPathsGood}, in the appendix).
Let $\cP$ be the set of paths   $P$ in $\G$ of  length $|P| =\frac{\log n}{\sqrt{d}}$, 
such that  there is no block-vertex in $P$. Then, we have that 
\begin{align}\label{eq:ExpositionShortPathsWithBreaks}
\Pr[\cP \ \textrm{is empty}]= 1- o(1)\enspace. 
\end{align}
Establishing the above, is one of the main technical challenges in this paper, as  the quantities we consider are inherently 
quite involved.

Below, we show how we use  \eqref{eq:ExpositionShortPathsWithBreaks} for the block construction.

\subsection{Block construction}\label{sec:highlevelBocks}

The aim is to obtain a block partition  $\cB=\{B_1, B_2, \ldots, B_N\}$ 
such that each $B\in \cB$  is small,  simply structured, while $\partial_{\rm out}B$, 
the outer boundary of $B$, consists  exclusively of block vertices. Note that $\partial_{\rm out}B$
is the set of vertices outside $B$ which have a neighbour inside the block. 

Let us give a high-level description of how $\cB$ looks like.
%
Recall that, typically, $\G$ is locally tree-like,  however, there are some relatively rare  short cycles, i.e., cycles of length less than 
$4\frac{\log n}{\log^4 d}$, 
which are far apart from each other.  

Each block $B\in \cB$ can be one of the following 
\begin{enumerate}
\item single vertex, 
\item a tree,
\item a unicyclic graph. 
\end{enumerate}
If $B$ consists of a single vertex, then this vertex must be a block-vertex. If $B$ is  multi-vertex, then $\partial_{\rm out}B$ consists of block-vertices.
This, somehow, guarantees that 
the heavy vertices are hidden  deep inside the blocks.   
%
The unicyclic blocks contain only  short cycles.

Intuitively,   \eqref{eq:ExpositionShortPathsWithBreaks}  guarantees 
that the blocks are not extensive structures.  Note that,  
for a heavy vertex $w$, we can reach the boundary of its block, by following any
path of length as small as  $\frac{\log n}{\sqrt{d}}$ that emanates from $w$.

Compared to the blocks we  have in \cite{EfthymiouDA14,efthymiou2018sampling} 
the ones we get  here are quite different. Note that the actual structure of the blocks depends on both
 graph $\G$ and the EA model $\mu_{\G, \bm{\vecJ}, \beta}$ on this graph. Furthermore, 
here, it is typical to see single vertex blocks consisting of a high-degree vertex, i.e., degree $\gg d$,  or having multi-vertex blocks whose  vertices are all low-degree ones, i.e., degree $\approx d$.  Of course, this has to do with the fact  that the aggregate influence is a different quantity  than the degree of a vertex.

\subsection{Rapid Mixing of Block Dynamics}\label{sec:MixingBlockDynamics}

Suppose that  we have a  block partition $\cB=\{B_1, \ldots, B_N\}$  as the one we describe in Section \ref{sec:highlevelBocks}. 

We  consider the {\em block dynamics} $(X_t)_{t\geq 0}$  with respect to the set of
blocks $\cB$.
The transition from  $X_t$ to  $X_{t+1}$  is  according to the  following steps: 
\begin{enumerate}
\item Choose uniformly at random a block  $B\in \cB$.
\item For every vertex $w$ outside  $B$, set $X_{t+1}(w)=X_t(w)$.
\item Draw $X_{t+1}(B)$, the configuration at $B$,  according to the marginal of $\mu$ at $B$, 
conditional on  the vertices in $\partial_{\rm out} B$ having the configuration  specified by $X_{t+1}$.
\end{enumerate}

At this stage, the goal is to show that $(X_t)_{t\geq 0}$ 
exhibits mixing time $\Tmix$ such that
\begin{align}\label{eq:ExpTarget4BlockDyn}
\Tmix &=O(N\log N)\enspace, 
\end{align}
where $N$ is the number of blocks in $\cB$.

We  use path coupling \cite{bubley1997path} to show~\eqref{eq:ExpTarget4BlockDyn}.
That is, we consider $(X_t)_{t\geq 0}$ and $(Y_t)_{t\geq 0}$, two copies of the block dynamics.
Assume that  at some time $t\geq 0$,  the configurations $X_t$ and $Y_t$ differ at a single vertex, $u^*$.
It suffices to show that  we have  one-step contraction i.e., we
can couple the two copies of the block dynamics such that 
the expected distance between $X_{t+1}$ and $Y_{t+1}$ is smaller than that of
$X_t$ and $Y_t$.

Typically, we establish contraction with respect to the {\em Hamming  metric} between  two configurations. It turns out that, for block dynamics, this metric is suboptimal. 
In contrast to the single-site dynamics, 
when we update a block $B$ that is adjacent to the disagreeing vertex $u^*$, the number of disagreements  grows by the size of $B$. 
For this reason, we follow an analysis for  path coupling  which adapts to the setting of block dynamics.

For each vertex $z$, we write $B_z$ for the block that $z$ belongs to.  Furthermore, let 
\begin{align}\label{eq:def:OfBOut}
\WA_{\rm out}(z)=\sum_{\substack{z\sim w \\ w\notin B_z}} \Inf_{\{z,w\}}\enspace. 
\end{align}
That is, $\WA_{\rm out}(z)$ is the sum of influences over the edges that 
connect  $z$ with its neighbours  outside the block $B_z$. 
Typically, if vertex $z$ has at least one neighrbour outside $B_z$, then
$\WA_{\rm out}(z)>n^{-7/3}$. 

We let $\impV \subseteq V$  be the set that of vertices   $z$ such that   
$\WA_{\rm out}(z)>0$. We call $\impV$ the set of {\em external} vertices. 
On the other hand, we call {\em internal} all the vertices in  
$V\setminus \impV$. Note that the internal vertices have no neighbours outside their block and hence $\WA_{\rm out}(z)=0$.

For the path coupling, we introduce the following distance metric
for any  two $\sigma,\tau\in \{\pm 1\}^V$
\begin{align*}
\EdgeBlockWeight(\sigma,\tau) &= \sum_{z\in V\setminus \impV}\Ind\{z\in \sigma \oplus \tau\}+
n^4{}\cdot \sum_{z\in \impV} \WA_{\rm out}(z)  \cdot \Ind\{z\in \sigma \oplus \tau\}\enspace, 
\end{align*}
where $\sigma\oplus \tau$ is the set of vertices $w\in V$ that the two configurations disagree, i.e., 
it consists of the vertices $w$ such that $\sigma(w)\neq \tau(w)$. 

The above metric assigns completely different weights to the internal and external vertices, respectively.   If some vertex $z$ is internal, 
then its disagreement  gets  (tiny) {weight} $1$.  On the other hand, if $z$ is external, its disagreement
gets  {weight} which is equal to  $n^4\times \WA_{\rm out}(z)\gg 1$. 
Particularly, we have that $n^4\times \WA_{\rm out}(z)=\Omega(n^{4/3})$, for all external vertices $z$.

The above metric essentially captures  that  the disagreements that do matter in the path coupling analysis, are those 
which involve vertices  at the boundary of blocks, i.e., external vertices. 
In particular,  the ``potential" for an external vertex  to spread 
disagreements to adjacent blocks increases with  $\WA_{\rm out}(z)$.
Let us remark that this observation was  first introduced and exploited in \cite{efthymiou2018sampling} in an analogous  setting.

In the path coupling analysis, we also   exploit  properties of the block partition $\cB$. Particularly, we use that the block vertices  are  far from the heavy ones, i.e., the vertices that have aggregate  influence larger than $1$. 
A heavy vertex in  block  $B$, once it becomes disagreeing,   tends to create higher than typical number of new disagreements. This is highly undesirable.
Having a large distance  between the heavy vertices inside $B$ and the boundary
 $\partial_{\rm out} B$
implies that the probability of a heavy vertex becoming disagreeing is very small. As a consequence,  the overall expected contribution
of the heavy vertices  becomes negligible.

In light  of all the above, we conclude that
there is a constant $C>0$ such that 
for any vertex $u^*\in V$, for any  pair of configurations $X_{t}, Y_{t}$ that differ on  $u^*$  
there is a coupling  such that 
\begin{align}
\Exp \left[ \left. 
\EdgeBlockWeight(X_{t+1}, Y_{t+1}) \  \right | \ 
 X_t, Y_t
\right ]
\leq   (1-C/N) \cdot \EdgeBlockWeight (X_{t}, Y_{t})\enspace. 
\end{align}
Using path coupling, and arguing  that $N=\Omega(\sqrt{n})$, it is standard   
to obtain \eqref{eq:ExpTarget4BlockDyn} from  the above inequality.
For further details, see Theorem \ref{thrm:RapidMixingBlockDyn}, in the appendix.

\subsection{Rapid Mixing for Glauber Dynamics - Comparison}

This part is a bit technical. Recall that our aim is to obtain a bound on the mixing time for the 
single-site Glauber dynamics, whereas  so far we only have a bound for the mixing time of the block-dynamics. 
To this end, we utilise  a well-known comparison argument form \cite{martinelli1999lectures}, 
which  relates the relaxation times of the Glauber dynamics and the block dynamics. 
Recall that the relaxation time of a Markov chain with transition matrix $P$  is equal to 
$\frac{1}{1-\lambda^*}$, where $\lambda^*$ is the second  largest eigenvalue in magnitude of $P$.

Letting  $\uptau_{\rm rel}$ and  $\uptau_{\rm block}$  be the relaxation times of the Glauber dynamics and the block
dynamics, respectively, from  \cite{martinelli1999lectures}, we obtain that 
\begin{align}\label{eq:ComparisonHighLevel}
\uptau_{\rm rel} &\leq \textstyle \uptau_{\rm block} \cdot \left ( \max_{B\in \cB}\{ \uptau_B \} \right )  \enspace,
\end{align}
where   $\uptau_B$ is  the relaxation time of the (single site) Glauber dynamics on each block
$B\in \cB$, under worst-case condition $\sigma$ at the boundary $\partial_{\rm out}B$.

From the rapid mixing result of block dynamics it is standard to obtain a bound on $\uptau_{\rm block}$,
hence, it remains to get a bound on $\uptau_B$, for all $B\in \cB$.
Our  endeavour  to bound $\uptau_B$ gives rise to a new weight over the paths in $\G=\G(n,d/n)$. 
Specifically, for a path $P$ in $\G$, we define the weight
\begin{align}
\compWeight(P)& ={\textstyle \beta \sum_{e}}  |J_{e}| + {\textstyle \sum_{v}} \log\degr(v)\enspace.
\end{align}
In the first sum, which involves the couplings,  the variable $e$ varies over all edges having at least one endpoint in $P$. 
The second sum varies over all the vertices in $P$.

 Building on a recursive argument from  \cite{efthymiou2018sampling,mossel2010gibbs},  
 for every tree-like block $B$ rooted at vertex $v$, we show 
\begin{align}
\uptau_B &\leq \textstyle \exp\left( \max_{P}\{ {\compWeight}(P)\}\right) 
\enspace, 
\end{align}
where the maximum is  over all the paths $P$ in $B$  from the root $v$ 
to $\partial_{\rm out}B$.
Note that such a path $P$ is at most $\frac{\log n}{\sqrt{d}}$ long. 
We also  get a similar bound for the relaxation time when $B$ is  unicyclic.

We show that with probability $1-o(1)$ over the instances of $\G$ and $\mu$,    every path $P$ in $\G$ of length at most $\frac{\log n}{\log^4 d}$, 
satisfies $\compWeight(P) \leq \frac{\log n}{\log^2 d}$. 
This implies that for every block $B$ we have that
\begin{align}
\uptau_B &\leq n^{\frac{3}{\log^2 d}} 
\enspace. 
\end{align}
As mentioned earlier, it is standard to obtain an estimate for $\uptau_{\rm block}$ from our rapid mixing results for 
the block-dynamics. Hence,  plugging $\uptau_B$'s and $\uptau_{\rm block}$ into  \eqref{eq:ComparisonHighLevel}, gives the desired bound  on  $\uptau_{\rm rel}$, the relaxation time for Glauber dynamics.  From this point on, it is standard  to get   Theorem \ref{thrm:MainResult}. 
For further details, see Section \ref{sec:Comparison}, in the appendix.

\section{Basic Notions}\label{sec:BasicNotions}

\subsubsection*{Gibbs distribution \& Influences:}

Even though most of the notions we describe below have been already introduced earlier, we present them here in full detail and formality.

We use the triplet $(G,\vecJ, \beta)$, where $G=(V, E)$ is a graph, $\vecJ=\{J_e \in \mathbb{R}:e\in E\}$, and $\beta\in \mathbb{R}_{> 0}$, to represent the ``glassy" Gibbs distribution $\mu=\mu_{G,\vecJ, \beta}$ on $\{\pm 1\}^{V}$ defined by 
\begin{align}\label{eq:DefOfGlassyMu}
\mu(\sigma) 
&\propto 
\exp\left( \beta \cdot \sum_{x\sim y} {\Ind}\{\sigma(y)=\sigma(x)\} \cdot J_{\{x,y\}} \right), & \forall \sigma \in \{\pm 1\}^{V}\enspace.
\end{align}
Clearly, the above gives rise to the EA model on $G$ once we take the couplings on the edges to be i.i.d. standard Gaussians.

The above triplet specifies  a set of {\em influences} over the edges of $G$. That is, for every $e\in E$ we define the {\em edge-influence} $\Inf_e$ such that
\begin{align}\label{eq:DefOFInf}
\Inf_e&= \frac{\left|1-\exp\left( \beta J_e  \right) \right|}{1+\exp\left( \beta J_{e}\right) }\enspace.
\end{align}
This is the standard influence we encounter in the context of Spectral Independence, e.g., see \cite{ChenLVStoc21}.

The following argument concerning influences is standard. Consider two adjacent vertices $u$ and $w$ in $G$, and let $\sigma, \tau$ be any two configurations at the neighbours of $w$ disagreeing only on at $u$. Writing $e=\{u,w\}$ for the edge between $w$ and $u$,  we have that
\begin{align}\label{eq:OneStepDisVsInfluence}
|| \mu^{\sigma}_w-\mu^{\tau}_w ||_{\rm TV}\leq \Inf_e\enspace,
\end{align}
where $\mu^{\sigma}_w$ is the marginal of $\mu$ at $w$ conditional on the configuration $\sigma$, and similarly for $\mu^{\tau}_w$.

\subsubsection*{Aggregate Influences:} For each vertex $u\in V$ we define the aggregate influence $\WA(u)$ such that
\begin{align}\label{eq:DefOfVPotential}
\WA(u)= \sum_{{w\sim u }} \Inf_{\{u,w\}}\enspace.
\end{align}
Consider now the random graph $\G(n,d/n)$ and the EA model on this graph with inverse temperature $\beta$. For a vertex $u$ in this graph, the corresponding  quantity, $\WA(u)$, turns out to be  an  involved random variable. 
The number of edges incident to $u$, hence the number of summads in $\WA(u)$, is Binomially distributed, i.e., ${\tt Bin}(n,d/n)$, while for each edge $e$ incident to $u$ the influence is also a random variable, i.e., for each $e$, we have $\bInf_e= \frac{\left|1-\exp\left( \beta \bJ_e  \right) \right|}{1+\exp\left( \beta \bJ_{e}\right)}$, where $\bJ_e$ is a  standard Gaussian.

There is a natural connection between the aggregate influence and the condition for the inverse temperature $\beta \le (1-\varepsilon) \beta_c(d)$ we have in Theorem \ref{thrm:MainResult}. Specifically, for any such $\beta$ we have that
\begin{align*}
\Exp[\WA(u)] \leq 1-\varepsilon\enspace. 
\end{align*}
In the analysis we establish the following tail bound for $\WA(u)$ establishing this random variable is well-concentrated.

\begin{theorem}\label{theorem:TailBound4WA}
For $\varepsilon>0$, there exists $d_0=d_0(\varepsilon) \ge 1$, such that for all $d\ge d_0$ and 
$0<\beta \le (1-\varepsilon)\beta_c(d)$, the following is true:

Consider $\G=\G(n,d/n)$ and the EA model on this graph at inverse temperature $\beta$. 
For a  fixed vertex $u$  in $\G$, we have that
\begin{align}\label{eq:WeightAVertexTailBound}
\Pr\left[\WA(u) \geq {\Exp[\WA(u)]+\frac{\varepsilon}{2}} \right]
\le \exp\left(-\frac{\varepsilon^4 }{8\pi}\cdot d\right)
\enspace.
\end{align}
\end{theorem}

A crucial corollary from the above result is that,  for the range of the parameters we consider in Theorem \ref{thrm:MainResult},  the probability of having $\WA(u)>1$ is exponentially small in $d$.

\subsubsection*{Weighting-Scheme:}
As before, consider the triplet $(G,\vecJ, \beta)$, for a fixed graph $G=(V,E)$, couplings  $\vecJ=\{J_e \in \mathbb{R}\ : \ e\in E\}$, and $\beta > 0$. Also, consider the distribution $\mu=\mu_{G,\vecJ, \beta}$,  i.e., the Gibbs distribution defined with respect to $(G,\vecJ, \beta)$.

We use the $\WA(u)$'s induced by $\vecJ$ to introduce a weight scheme for both the vertices and the paths of $G$. Let us start with the weights for the vertices. 

\begin{definition}[$(d,\varepsilon)$ Vertex-Weight]\label{def:VertexWeights}
For parameters $\varepsilon>0$ and $d>0$, and for any $u\in V$ define  the \emph{vertex weight}
\begin{align}
\cappedWA(u) &= \begin{cases}
                {1-\frac{\varepsilon}{4}}  & \text{ if } \WA(u) \le {1-\frac{\varepsilon}{2},}  \\
                d\cdot\WA(u) & \text{ otherwise. }
                \end{cases}
                \enspace 
\end{align}
\end{definition}

Then, having defined the weights $\cappedWA(u)$, we define the weight of paths in~$G$. 

\begin{definition}[$(d,\varepsilon)$ {Path}-Weight]\label{def:WightOfPath}
For parameters $\varepsilon>0$ and $d>0$, and for any path  $P=(w_0, \ldots, w_\ell)$ in $G$,  we define the path weight
\begin{align*}
\cappedWA(P) &=  \prod\nolimits_{i = 0}^{\ell} \cappedWA(w_{i})\enspace.
\end{align*}
\end{definition}

That is, the weight of a path, is obtained by multiplying the weights of its vertices. 
Finally, we have the notion of block vertex, which plays a key role in the construction of the block-partition in Section~\ref{sec:StructuralProporties}. 

\begin{definition}[$(d,\varepsilon)$-Block Vertex]\label{def:BlockVertex}
For  $\varepsilon, d>0$,  a vertex $u$ in $G$ is called  $(d,\varepsilon)$-{\em block vertex}, if for every path $P$ of length at most $\log n$ emanating  from $u$,  the $(d, \varepsilon)$-weight of $P$ is less than 1, i.e.,  $\cappedWA(P) < 1$. 
\end{definition}

\cref{def:WightOfPath,def:BlockVertex} make it apparent that every path $P$ connecting a heavy vertex $w$ and a block vertex $u$, must contain a large ``buffer" of light vertices. We exploit this property of the block vertices in the block construction that follows.

\section{Block Partition \& Other Structural properties}\label{sec:StructuralProporties}

In order to show  \cref{thrm:MainResult},  rather than considering (single site) Glauber dynamics, 
firstly, we  consider  {\em block dynamics}.  Specifically, we show that  there is an appropriately 
constructed  block-partition  $\cB$ of the set of   vertices  such that the corresponding 
block dynamics mixes fast. 
Subsequently, we utilise a comparison argument to show that the above result implies  rapid mixing 
for the Glauber dynamics, as well. 

The block partition $\cB$ is specified  with respect to the weights we defined in the previous section. 
Hence, $\cB$ depends on $\G=\G(n,d/n)$ and the specification of the Edwards-Anderson model on this graph. 
In what follows, we describe in full detail how the  set of blocks $\cB$  looks like. Furthermore, 
we show that the typical instances of   $\G$ and $\mu$, the Edwards-Anderson model on $\G$ at inverse temperature $\beta$
as specified in Theorem \ref{thrm:MainResult},  admit such a partition.

\begin{definition}[$(d,\varepsilon)$-Block Partition]\label{def:BlockDecomp}
For $\varepsilon, d>0$,    
 the vertex partition $\mathcal{B}=\{B_1, \ldots, B_{N}\}$ is called 
\emph{$(d,\varepsilon)$-block partition} if 
 for every $B\in \cB$  the following is true:
\begin{enumerate}
\item block $B$ is a tree with at most one extra edge, \label{itm:BPunicyclic}
\item  if $B$ is a multi-vertex block, then we have the following:
\begin{enumerate}
\item every $w\in \partial_{\rm out}B$ is a $(d,\varepsilon)$-block vertex, \label{itm:BPblockBoundary}
\item  every $w\in \partial_{\rm out}B$ has exactly one neighbour in $B$, \label{itm:BPSingleNeighBoundary}
\item  if $B$ contains a cycle, this is a short one, i.e.,  its length is at most $4\frac{\log n}{\log^4 d}$,   \label{itm:BPShortCycle}
\item the distance of the short cycle from the boundary of its block is at least $\log^5d$.
\label{itm:BuffCond}
\end{enumerate}
\item if $B$ is a single-vertex block, then this vertex is $(d,\varepsilon)$-block.
\end{enumerate}
\end{definition}

We would like to stress here that the blocks in $\cB$ are really simply structured, i.e., these are trees or unicyclic graphs.

Apart from the weight on paths introduced in Definition \ref{def:WightOfPath}, we consider yet another weight for the paths. In particular, for a path $P$, we define 
\begin{align}\label{def:Weitgh4Comparison}
\compWeight(P)& = \beta \sum_{e } |J_{e}| + \sum_{v} \log\degr(v)\enspace,
\end{align}
where the first sum varies over all edges $e$ in $G$ having least one endpoint in $P$, while the second sum varies over all vertices of $P$. 
As we discuss later,  this weight arises in our comparison  argument.

For any $d,\varepsilon>0$, for $0<\beta \le (1-\varepsilon)\beta_c(d)$,  
we define $\IntrstGraphFam (d, \varepsilon)$ to be the set of  triplets $(G, {\vecJ} , \beta)$ possessing the following  properties:
\begin{enumerate}
\item \label{itm:propty1} $G$ admits a $(d, \varepsilon)$-block partition,
\item \label{itm:propty2}every path $P$ in $G$ of length $|P| \le \frac{\log n}{\log^4 d}$, satisfies $
\compWeight(P) \leq {\frac{\log n}{\log^2 d}}
$,
\item \label{itm:propty3} every edge $e$ in $G$ satisfies $n^{-7/3}\leq |J_e| \le 10\sqrt{\log n}$.
\end{enumerate}
In light of all the above,   we prove the following result.

\begin{theorem}\label{thrm:BlockParitionGnp}
For $\varepsilon>0$, there exists $d_{0}=d_0(\varepsilon) \ge 1$ such that for any
$d\ge d_0$, and for any $0<\beta \le (1-\varepsilon)\beta_c(d)$,
 the following is true:

For $\G=\G(n,d/n)$, and for $\bm{\vecJ} = \{\bJ_e :e \in E(\G)\}$ a family of i.i.d. {\em standard Gaussians}
we have 
\begin{align*}
\Pr\left[(\G,\bm{\vecJ}, \beta)\in \IntrstGraphFam(d,\varepsilon) \ \right] \ge 1-n^{-{1/4}}
\enspace.
\end{align*}
\end{theorem}
The proof of Theorem \ref{thrm:BlockParitionGnp} appears in Section \ref{sec:thrm:BlockParitionGnp}.

\section{Fast Mixing of Block Dynamics}\label{sec:BlockDynamicsMixingT}


For $\varepsilon>0$ and sufficiently large $d>0$, we consider a triplet $(G,\vecJ, \beta)\in \IntrstGraphFam(d,\varepsilon)$, and a $(d,\varepsilon)$-block 
partition $\cB$ obtained from this triplet. We also consider $(X_t)_{t\geq 0}$ the block-dynamics defined with respect to the Gibbs distribution 
$\mu=\mu_{G,\vecJ,G}$ and the set of blocks $\cB$. We show that $(X_t)_{t\geq 0}$  is optimally  mixing, i.e., the mixing time is $O(N\log N)$, where $N$ is the number of blocks in $\cB$.

Recall that $(X_t)_{t\geq 0}$ is a discrete-time Markov chain. 
We get $X_{t+1}$ from $X_t$ by choosing uniformly at random a block  $B_t\in \cB$ and updating the
configuration of $B_t$ according to the marginal of $\mu$ at this block, conditional on that the configuration 
outside $B_t$ being as specified by $X_{t}$.

\begin{theorem}\label{thrm:RapidMixingBlockDyn}
For $\varepsilon>0$, there exists $d_0=d_0(\varepsilon) \ge 1$, such that for all $d\ge d_0$, for
$0<\beta \le (1-\varepsilon)\beta_c(d)$ and any $(G, \vecJ, \beta) \in \IntrstGraphFam (d, \varepsilon)$, there exists $C>0$ such that
 the following is true:

Let $\cB$ a $(d,\varepsilon)$-block partition of $G$ and consider the Gibbs distribution $\mu=\mu_{G,\vecJ, \beta}$.
The block dynamics defined with respect to  $\mu$  and the set of blocks $\cB$, exhibits
mixing time 
\begin{align*}
 \Tmix\leq C\cdot N \log N \enspace, 
\end{align*}
where $N$ is the number of blocks in $\cB$.   
\end{theorem}

The proof of Theorem \ref{thrm:RapidMixingBlockDyn} appears in Section \ref{sec:thrm:RapidMixingBlockDyn}. 
Some remarks are in order. 
First, note that in the above theorem we do not need to have an instance of $\G(n,d/n)$,
or Gaussian couplings at the edges, i.e., it holds for every triplet in $\IntrstGraphFam$.

Second, even though we use block dynamics only as a tool for the analysis, it is worth mentioning that one can actually implement this dynamics efficiently. 
This is due to the fact that the blocks in $\cB$ are trees with at most one extra edge. Hence,  each step $t$ of
the dynamics  can be implemented in  $O(|B_t|)$ using {\em dynamic programming}, where $|B_t|$ is the number of vertices in the block $B_t$.

\section{Fast mixing for single site Glauber dynamics -  Proof of Theorem \ref{thrm:MainResult}}\label{sec:Comparison}

We use the  rapid mixing result from Section \ref{sec:BlockDynamicsMixingT} to upper bound the mixing time of the single site Glauber dynamics  by means of the following well-known comparison result from  \cite{martinelli1999lectures}. 
\begin{proposition}\label{prop:Comparison}
For the graph $G=(V,E)$,  let $(X_t)_{t\geq 0}$ be the   block dynamics, with set of blocks $\cR$, such that  each vertex $u \in V$ belongs to $M_u$ different blocks. 
Furthermore, let $(Y_t)_{t\geq 0}$ be the  {\em single site} dynamics on $G$.

Let $\uptau_{\rm block}$ and $\uptau_{\rm rel}$ be the relaxation times of $(X_t)$ and $(Y_t)$, respectively. 
Furthermore, for each block $B\in \cR$, let $\uptau_B$ be the relaxation time of the 
 single site dynamics on $B$, given any arbitrary condition at $\partial_{\rm out} B$.  
Then,
\begin{align*}
\uptau_{\rm rel} &\leq \textstyle \uptau_{\rm block} \cdot \left ( \max_{B\in \cR}\{ \uptau_B \} \right ) \cdot  \left (\max_{u\in V} \{ M_u \}  \right) \enspace.
\end{align*}
\end{proposition}

For $\varepsilon>0$ and sufficiently large $d>0$,  consider $(G,\vecJ,\beta)\in \IntrstGraphFam(d,\varepsilon)$ and the corresponding distribution   $\mu=\mu_{G,\vecJ,\beta}$, specified as in \eqref{eq:DefOfGlassyMu}.  Also, let $\cB$ be a $(d,\varepsilon)$-block partition for $G$.

Let  $(X)_{t\geq 0}$ be  the   block dynamics defined with respect to  $\mu$ and the set of  blocks $\cB$. 
Also, let $(Y)_{t\geq 0}$ let  be the  (single site) Glauber dynamics on $\mu$.  

Suppose that $\uptau_{\rm block}$ is the relaxation time for $(X_t)$, while let $\uptau_{\rm rel}$ be the relaxation time for $(Y_t)$. 
Finally, for each $B\in \cB$, let $\uptau_{B}$ be the relaxation time of the  single site dynamics on $B$, for
arbitrary condition $\sigma$ at $\partial_{\rm out} B$.  
We have the following result.

\begin{theorem}\label{thrm:RelaxBounds}
It holds that $\uptau_{\rm block}=O(n\log n)$, while for any $B\in\cB$ we have that  $\uptau_{B} \leq {n^{\frac{3}{\log^2 d}}}$.
\end{theorem}
The proof of Theorem \ref{thrm:RelaxBounds} appears in Section \ref{sec:thrm:RelaxBounds}.
The following corollary is immediate from Theorem \ref{thrm:RelaxBounds}  and Proposition \ref{prop:Comparison}.
\begin{corollary}\label{cor:FastMixGinFContTime}
For any $\varepsilon>0$ there exist $d_0=d_0(\varepsilon)\ge 1$, such that for every $d\ge d_0$ and $0<\beta \le (1-\varepsilon)\beta_c(d)$  the following is true:

For any graph $G=(V,E)$, and $\vecJ=\{J_e: e\in E\}$ such that $(G,\vecJ,\beta)\in \IntrstGraphFam(d,\varepsilon)$,
consider the Gibbs distribution $\mu=\mu_{G,\vecJ, \beta}$.  Then, the single site Glauber dynamics on $\mu$
exhibits  relaxation time 
\begin{align*} 
\uptau_{\rm rel} &\leq {n^{1+\frac{3}{\log^2 d}}} \enspace. 
\end{align*}
\end{corollary}

\begin{proof}[Proof of Theorem \ref{thrm:MainResult}]
According to Theorem \ref{thrm:BlockParitionGnp} we have the following:

For any $\varepsilon>0$, there exists $d_0=d_0(\varepsilon)$ such that for any $d \ge d_0$ and for $0<\beta \le (1-\varepsilon)\beta_c(d)$ the 
following is true:
For $\G=\G(n,d/n)$, for $\bm{\vecJ} = \{\bJ_e :e \in E(\G)\}$ a family of i.i.d.  standard Gaussians,
we have that
\begin{align*}
\Pr\left[(\G,\bm{\vecJ}, \beta)\in \IntrstGraphFam(d,\varepsilon) \ \right] \ge 1-n^{-1/4}
\enspace.
\end{align*}
Hence, due to Corollary \ref{cor:FastMixGinFContTime}, with probability at least $1-n^{-1/4}$ over the instances of $(\G,\bm{\vecJ})$ 
the single site, Glauber dynamics exhibits relaxation time $\uptau_{\rm rel}$ such that
\begin{align*}
\uptau_{\rm rel}&\leq {n^{1+\frac{3}{\log^2 d}}}\enspace. 
\end{align*}
It is now standard that 
\begin{align*}
\Tmix &= O\left(n^{2+\frac{3}{\log^2 d}} \right)\enspace,  
\end{align*}
concluding the proof of Theorem \ref{thrm:MainResult}.
\end{proof}

\section{Proof of Theorem \ref{thrm:RapidMixingBlockDyn}}\label{sec:thrm:RapidMixingBlockDyn}

We use path coupling \cite{bubley1997path} for the proof of Theorem \ref{thrm:RapidMixingBlockDyn}.
We establish contraction  by introducing a metric on the configuration space $\{\pm 1\}^{V}$ (also discussed in Section \ref{sec:MixingBlockDynamics}). For each vertex $z\in V$, we let $B_z$ denote the block that $z$ belongs to. Furthermore, let 
\begin{align}\label{eq:def:OfBOutA}
\WA_{\rm out}(z)=\sum_{\substack{z\sim w \\ w\notin B_z}} \Inf_{\{z,w\}}\enspace. 
\end{align}
That is, $\WA_{\rm out}$ is the aggregate influence of vertex $z$ on its neighbours {\em outside} block $B_z$. Let also $\impV\subseteq V$  consist of all vertices $z\in V$ such that  $\WA_{\rm out}(z)>0$.
Perhaps it is useful to remind the reader that all couplings $J_e$'s are assumed to be non-zero. Specifically, for every edge $e$, we have that $|J_e|>n^{-7/3}$.
The above  implies that,  if $\WA_{\rm out}(z)=0$, i.e., $z \in V\setminus \impV$, then $z$ has no neighbours outside $B_z$. We call such a vertex {\em internal}.

In light of the above, we introduce the following distance metric on the configuration space~$\{\pm 1\}^V$:  
For any  two $\sigma,\tau\in \{\pm 1\}^V$, we let
\begin{align}\label{eq:DefOfHammingWeights}
\EdgeBlockWeight(\sigma,\tau) &= \sum_{z\in V\setminus \impV}\Ind\{z\in \sigma \oplus \tau\}+
n^4\cdot \sum_{z\in \impV} \WA_{\rm out}(z)  \cdot \Ind\{z\in \sigma \oplus \tau\}\enspace,
\end{align}
where $\sigma\oplus \tau$ is the set of vertices $w\in V$ that the two configurations disagree, i.e., 
$\sigma(w)\neq \tau(w)$. Let us note that a similar distance metric first appeared in \cite{EfthymiouHSV18}.

We proceed now with the path coupling argument. Consider two copies of the block dynamics $(X_t)_{t\geq 0}$ and $(Y_t)_{t\geq 0}$. 
Assume that  at time $t\geq 0$,  the configurations $X_t$ and $Y_t$ differ at a single vertex $u^*$. It suffices to show that we have contraction at $t+1$, i.e.,  the expected distance between $X_{t+1}$ and $Y_{t+1}$ is smaller than that between $X_t$ and $Y_t$.
To this end, we prove the following theorem.

\begin{theorem} \label{thrm:BDContraction}
Suppose that $u^*\in \impV$.  For any $B\in\cB$ such that $u^*\in \partial_{\rm out} B$,  there exists a coupling between $X_{t+1}$ and $Y_{t+1}$, such that 
\begin{align*}
\Exp \left[   
\EdgeBlockWeight(X_{t+1}, Y_{t+1})
- \EdgeBlockWeight (X_{t}, Y_{t}) \  \mid \ 
 X_t, Y_t, \ B\  \textrm{updated at time $t+1$}
\right ]
& \leq  \Inf_{\{u^*,z\}} \cdot n^4\cdot (1-\varepsilon/6) \enspace,
\end{align*}
where $z\in B$ is the unique neighbour of $u^*$ in $B$.   
\end{theorem}

The proof of Theorem \ref{thrm:BDContraction} appears in Section \ref{sec:thrm:BDContraction}. 

In light of Theorem \ref{thrm:BDContraction}, we proceed with the path coupling analysis.
Consider now two cases for vertex $u^*$. In the first case, assume that $u^*\in \impV$, whereas in the second one, assume that $u^*\notin \impV$.  We establish contraction for both cases. 

We start by assuming that $u^*$ is external, i.e., $u^*\in \impV$. Suppose that at time $t+1$ the block $B_{t+1}\in \cB$ is updated. 
If we have  $B_{u^*}=B_{t+1}$, then we  use identity coupling and get that  $X_{t+1}=Y_{t+1}$. Hence, we have that
\begin{align}\label{eq:ExpDistanceBuStarUpdateExt}
\Exp \left[    
\EdgeBlockWeight(X_{t+1}, Y_{t+1})
- \EdgeBlockWeight (X_{t}, Y_{t})
 \  \mid \ 
 X_t, Y_t, \ B_{t+1}=B_{u^*}
\right ] 
{=}
-n^4\cdot \WA_{\rm out}(u^*) \enspace.
\end{align}
If $B_{t+1}$ is such that $u^*\notin \partial_{\rm out} B_{t+1}$, we have that $\EdgeBlockWeight(X_{t+1}, Y_{t+1})=\EdgeBlockWeight(X_{t}, Y_{t})$ by using identity coupling. Hence, we have that
\begin{align}\label{eq:ExpDistanceBDistantUpdateExt}
\Exp \left[    
\EdgeBlockWeight(X_{t+1}, Y_{t+1})
- \EdgeBlockWeight (X_{t}, Y_{t})
 \  \mid \ 
 X_t, Y_t, \ u^*\notin \partial_{\rm out} B_{t+1}
\right ] 
=0\enspace. 
\end{align} 
Finally, if  $B_{t+1}$ is such that $u^*\in \partial_{\rm out}B_{t+1}$, then  Theorem \ref{thrm:BDContraction} implies that
\begin{align}
\nonumber
\Exp \left[  
\EdgeBlockWeight(X_{t+1}, Y_{t+1})
- \EdgeBlockWeight (X_{t}, Y_{t})
 \  \mid \ 
 X_t, Y_t,\  u^*\in \partial_{\rm out}B_{t+1}  
\right ]
&\leq \sum_{\substack{z\sim u^* \\ z \notin B_{u^*}}} \Inf_{\{u^*,z\}} \cdot n^4\cdot (1-\varepsilon/6) 
\\
&= n^4\cdot (1-\varepsilon/6)\cdot \WA_{\rm out}(u^*)  \enspace,  \label{eq:ExpDistanceBCloseUpdateExt}
\end{align}
where the last equality follows from the definition of $\WA_{\rm out}(u^*)$. 
Combining now \eqref{eq:ExpDistanceBuStarUpdateExt}, \eqref{eq:ExpDistanceBDistantUpdateExt} and \eqref{eq:ExpDistanceBCloseUpdateExt}, we get that
\begin{align}\label{eq:ImportantU}
\Exp \left[  
\EdgeBlockWeight(X_{t+1}, Y_{t+1})
- \EdgeBlockWeight (X_{t}, Y_{t})
 \  \mid \ 
 X_t, Y_t, \ u^*\in \partial_{\rm out}B_{t+1}  
\right ]
&\leq -\varepsilon \cdot \frac{n^4}{6N} \cdot \WA_{\rm out}(u^*)\enspace. 
\end{align}
On the other hand, if the disagreeing vertex, $u^*$  is internal, i.e., $u^*\in V\setminus \impV$, then the disagreement cannot spread.
Hence,  when the block $B_{u^*}$ updates, then the disagreement disappears, and thus, if $u^*\notin \impV$, then we have that
\begin{align}\label{eq:NImportantU}
\Exp \left[ 
\EdgeBlockWeight(X_{t+1}, Y_{t+1})
- \EdgeBlockWeight (X_{t}, Y_{t})
 \  \mid \ 
 X_t, Y_t  
\right ]
\leq  -\frac{1}{N}\enspace.
\end{align}
From \eqref{eq:NImportantU} and \eqref{eq:ImportantU} it is standard to get Theorem \ref{thrm:RapidMixingBlockDyn}.

\subsection{Proof of Theorem \ref{thrm:BDContraction}}\label{sec:thrm:BDContraction}
For any set of vertices $\Lambda$ such that $\Lambda\subseteq B$, we define the quantity 
\begin{align*}
\cD(\Lambda, X_t, Y_t) 
&=\sum_{w\in \Lambda}   n^4 \cdot 
                        \WA_{\rm out}(w) \cdot 
                        \Ind \{w\in X_t\oplus Y_t\} 
        \enspace.
\end{align*}
%
%
Furthermore, for any  $w\in B\cup \partial_{\rm out} B$ we let
\begin{align*}
S_w(\Lambda) &=  
\Exp\left[ 
\left .  {\cD}(\Lambda, X_{t+1}, Y_{t+1})  \ \right |
\    X_{t+1} (w)\neq Y_{t+1}(w), \  X_{t},Y_{t}, \   B\textrm{ updated at $t+1$} 
\right] \enspace.
\end{align*}
Recall that $z\in B$ is the unique neighbour of $u^*$ in $B$. We have that
\begin{align}\label{eq:Target4:thrm:BlockUpdtCovergent}
\Exp
\left [  \left . 
\EdgeBlockWeight(X_{t+1}, Y_{t+1}) - \EdgeBlockWeight (X_{t}, Y_{t}) 
\ \right |
 X_t, Y_t, \ B\  \textrm{updated at $t+1$}
\right ]
\leq   S_{u^*}(B)+\Inf_{\{u^*,z\}} \cdot n \enspace.
\end{align}
To see the above, notice that the term $S_{u^*}(B)$ equals the contribution of the external vertices in the block $B$, i.e., those vertices $w$ such that $\WA_{\rm out}(w)>0$. The term $\Inf_{\{u^*,z\}}\cdot  n$ is an overestimation for the contribution of the internal vertices of $B$. Specifically, once the disagreement of $u^*$ propagates at $z$, which happens with probability at most $\Inf_{\{u^*,z\}}$, then the contribution of internal vertices in $B$ is at most $n$.

The theorem follows by showing that  
\begin{align}\label{eq:Target4thrm:BDContraction}
S_{u^*}(B) \leq \Inf_{\{ u^*,z\}}\cdot n^4 \cdot (1-\varepsilon/5)\enspace.
\end{align}
The following standard, technical, result is useful in our analysis.

\begin{lemma}\label{lemma:Coupling4TreeMeasures}
Let  $H=(V_H, E_H)$ be the graph induced by  $B\cup \partial_{\rm out}B$, for a multi-vertex block $B\in \cB$.

For any  $\Lambda\subset V_H$ that includes $\partial_{\rm out}B$, for any  vertex 
$w\in V_H\setminus \Lambda$, where   $\sigma\in \{\pm 1\}^{\Lambda}$,  let  
$\mu^+=\mu(\cdot\ |\ \{ \Lambda, \sigma\}, \{w, +\} )$ and $\mu^-=\mu(\cdot\ |\ \{ \Lambda, \sigma\}, \{w, -\} )$. 
There exist a coupling $\nu$ of the measures $\mu^+$ and $\mu^-$ such that   the following is true:

Let   $M \subseteq V_H\setminus  \Lambda$, be a subset of neighbours of  $w$ in $B$ 
which do not  belong to the cycle inside $B$,  if such a cycle exists. 
Then, for  $(\bX,\bY)$ distributed as in $\nu$,  we have that 
\begin{align}\nonumber
\Exp[ |(\bX  \oplus \bY)\cap M|]&\leq {\sum\nolimits_{u\in M}} \Inf_{\{w,u\}}\enspace,
\end{align}
where $\Inf_{\{w,u\}}$ is the influence of the edge $\{w,u\}$.
\end{lemma}

\begin{proof}
Choose $\nu$ to be the maximal coupling for each one of neighbours of $w$ in $M$. Since none of the vertices in $M$ belongs to the same cycle, it is elementary to show that we have
\begin{align*}
\Exp[ |(\bX  \oplus \bY)\cap M|] 
&=
\textstyle {\sum_{u\in M}} ||\mu^+_{u}-\mu^-_u ||_{\rm TV}\enspace,
\end{align*}
where $\mu^+_{u}, \mu^-_u$ are the marginals of $\mu^+, \mu^-$, respectively, at vertex $u\in M$.  It is standard to show that $||\mu^+_{u}-\mu^-_u ||_{\rm TV}\leq \Inf_{\{w,u\}}$. For example, see discussion in \eqref{eq:OneStepDisVsInfluence}. 

The above concludes the proof of the lemma. 
\end{proof}

To proceed with the proof of \eqref{eq:Target4thrm:BDContraction}, we make some useful observation following from the definition of the  block-vertex (Definition \ref{def:BlockVertex}) and that of the block-partition $\cB$ (Definition \ref{def:BlockDecomp}).

\begin{corollary}\label{cor:HeavyWABound}
For any multi-vertex block $B\in\cB$, for any $u\in \partial_{\rm out}B$ and any vertex $w\in B$ the following is true:

Let $P$ be any path of length $\ell$ connecting $u$ and $w$. Let  also $M$ be the set of heavy vertices in $P$, i.e., for all $x\in M$ we have $\WA(x)>1-\varepsilon/2$. Then, we have that
\begin{align*}
\prod_{x\in M}\WA(x
) \leq d^{-|M|}(1-\varepsilon/4)^{|M|-\ell}\enspace.
\end{align*}
\end{corollary}
 
The above corollary follows by noticing that any path $P$ in $B$ connecting $w$ to a vertex in $\partial_{\rm out}B$, needs to satisfy $\cappedWA(P)<1$. The inequality then follows from the definition of the weight $\cappedWA(P)$.
\begin{corollary}\label{cor:BufferWeightsSmall}
For every multi-vertex block $B$ and every vertex $w\in B$ such that $\WA(w)>1-\varepsilon/2$, the closest block-vertex to $w$ is at distance  $>\log d$. 
\end{corollary}

Corollary~\ref{cor:BufferWeightsSmall} follows from Corollary~\ref{cor:HeavyWABound} by noticing that if there is a path $P$ from $w$ to $\partial_{\rm out}B$ with $|P| \le \log(d)$, then $\cappedWA(P)>1$. Clearly, this cannot be true, since  $\partial_{\rm out}B$ consists only of block-vertices.

Let us introduce a few useful concepts. 
We  let $\BTree$ be the  subset of  $B$ that contains all  vertices reachable from $u^*$ through a path within $B$ of  length at most $\log d$. It is easy to see that $\BTree$ always induces a tree, since property \eqref{itm:BuffCond} of the block partition $\cB$ implies that the distance of the cycle in $B$, if any, is at least $\log^5 d$ from $z$. Furthermore, all vertices in $w\in \BTree$ satisfy that $\WA(w)<1-\varepsilon/2$. This is due to Corollary \ref{cor:BufferWeightsSmall}.

Consider the root of $\BTree$ to be the vertex $z$, this is the unique vertex in $B$ adjacent to $u^*$. Also, let $\BRest=B\setminus \BTree$. Then, the linearity of expectation yields
\begin{eqnarray} \label{eq:LinearityOfQvBVsQvT+Rest}
S_{u^*}(B) = S_{u^*}\left(\BTree\right)+S_{u^*}\left(\BRest\right)\enspace. 
\end{eqnarray}
We estimate the quantities  $S_{u^*}(\BTree)$ and $S_{u^*}(\BRest)$, separately. As far as $S_{u^*}(\BTree)$ is concerned, we have the following lemma. 

\begin{lemma} \label{lemma:Bound4QTPlustUniformity}
We have that
$S_{u^*}\left(\BTree\right)\leq \Inf_{\{u^*,z\}}\cdot n^4 \cdot (1-\varepsilon/2)$.
\end{lemma}

As far as $S_{u^*}(\BRest)$ is concerned, we work as follows:
Let $S_A$ be the  contribution to $S_{u^*}(\BRest)$  coming from vertices reachable from $z$ with paths in $B$ that do not include vertices of the cycle in $B$ (if there is any). Also, let $S_B$ be  contribution to $S_{u^*}(\BRest)$  coming from vertices that are reachable from $z$ via a path that includes vertices from the cycle (if there is any). 
The linearity of expectation implies that 
\begin{align}\label{eq:SUVsSASB}
S_{u^*}\left(\BRest\right)&=S_A+S_B\enspace. 
\end{align}

We prove the following bounds for $S_A$ and $S_B$ in sections \ref{sec:lemma:ContributionsSA} and \ref{sec:prop:ContributionsSB}, respectively.

\begin{lemma}\label{lemma:ContributionsSA}
We have that $S_A\leq   \Inf_{\{u^*, z\}}\cdot n^4 \cdot d^{-{\varepsilon}/{5}}$.
\end{lemma}

\begin{proposition}\label{prop:ContributionsSB}
We have that $S_{B} \leq \Inf_{\{u^*, z\}}\cdot n^4  \cdot d^{-12}$.
\end{proposition}

Plugging the bounds from  Lemma \ref{lemma:ContributionsSA} and Proposition \ref{prop:ContributionsSB} into \eqref{eq:SUVsSASB}, we get
\begin{align*}
S_{u^*}\left(\BRest\right)&\leq 2 \Inf_{\{u^*, z\}}\cdot n^4  \cdot d^{-\varepsilon/5}  \enspace. 
\end{align*}
Plugging the above bound  for $S_{u^*}(\BRest)$, and the bound from Lemma \ref{lemma:Bound4QTPlustUniformity} for $S_{u^*}(\BTree)$ into \eqref{eq:LinearityOfQvBVsQvT+Rest}, gives the desired bound for  $S_{u^*}(B)$, and Theorem \ref{thrm:BDContraction} follows.
\hfill $\Box$

\subsection{Proof of Lemma \ref{lemma:Bound4QTPlustUniformity}}\label{sec:lemma:Bound4QTPlustUniformity}

Since $\BTree\cup\{u^*\}$ is a tree, while there is a single disagreement at $u^*$, we can use the coupling introduced in Lemma \ref{lemma:Coupling4TreeMeasures}. 

We prove the lemma using induction on the height of $\BTree$. The base case is when $\BTree$ is a single 
vertex tree, i.e., the height is $0$. Let $\BTree=\{z\}$.  In this case, recall that due to
Corollary \ref{cor:BufferWeightsSmall}, we must have
$\WA_{\rm out}(z)\le \WA(z) \le 1-\varepsilon/2$.  
Then, by Lemma \ref{lemma:Coupling4TreeMeasures}, we have that
\begin{align*}
S_{u^*}\left(\BTree\right) & \leq  \Inf_{\{u^*, z\}} \cdot n^4 \cdot \WA_{\rm out}(z)  \leq  \Inf_{\{u^*, z\}} \cdot n^4 \cdot (1-\varepsilon/2) \enspace.
\end{align*}
We proceed with the inductive step.  Recall that  $z$ is  the root of $\BTree$. For every vertex $y$, child of vertex $z$ in $\BTree$, we let  $\BTree_{y}$ be  the subtree of $\BTree$ rooted at $y$, and containing all its decadents.
The induction hypothesis is that $S_{z}(\BTree_{y})\leq \Inf_{\{z, y\}}  n^4 \cdot (1-\varepsilon/2)$, for every $y$ child of $z$. Then, we have that 
\begin{align*}
S_{u^*} \left(\BTree \right )  
&\leq 
\Inf_{\{u^*, z\}} \cdot \left(n^4\WA_{\rm out}(z) + \sum_{y\in N(z) \cap \BTree} S_z\left( \BTree_{y}\right) \right) \\
&\le  \Inf_{\{u^*, z\}} \cdot \left(n^4\WA_{\rm out}(z) + \sum_{y\in N(z) \cap \BTree} \Inf_{\{z, y\}}  \cdot n^4 \cdot (1-\varepsilon/2) \right) 
&\mbox{[induction hypothesis]} 
\\
&\le  \Inf_{\{u^*, z\}} \cdot n^4 \cdot\left(\WA_{\rm out}(z) + \left[\WA(z)- \WA_{\rm out}(z) \right]    \right) \\   
&\le \Inf_{\{u^*, z\}} \cdot  n^4\cdot (1-\varepsilon/2) \enspace.
\end{align*}
The above proves the inductive step, and thus, the lemma follows.  
\hfill $\Box$

\subsection{Proof of Lemma \ref{lemma:ContributionsSA}}\label{sec:lemma:ContributionsSA}
 First note that  the vertices in $B$ that are reachable from $z$ using a path that does not include vertices 
 in the cycle, induce a tree.  Let us denote this tree with $\BBTree$.

Working as in Lemma \ref{lemma:Bound4QTPlustUniformity} we have that
\begin{align*}
S_{u^*}\left(\BBTree\right) \leq 
\Inf_{\{u^*, z\}} \cdot \left(n^4\cdot \WA_{\rm out}(z) + { \sum\nolimits_{x \in N(z)\cap B}} \;S_z\left(\BBTree_x\right) \right) \enspace,
\end{align*}
where $\BBTree_x$ is the subtree of $\BBTree$ including $x$ and all its decedents. 
Repeating the above step, we get 
\begin{align}
S_{u^*}\left(\BBTree\right)
&\leq
\Inf_{\{u^*, z\}}  \left(n^4 \WA_{\rm out}(z) 
	\; + \;
\sum_{x \in N(z)\cap B} 
\Inf_{\{ z, x\}} \left(n^4 \WA_{\rm out}(x) 
	+ 
\sum_{y \in (N(x)\cap B)\setminus \{z\} }  S_{x}\left( \BBTree_{y}\right) \right) \right)
\nonumber\\
&\le   \Inf_{\{u^*, z\}}   \left( n^4 \WA_{\rm out}(z)  + 
\left[\WA(z)-\WA_{\rm out}(z) \right] 
\max_{x \in N(z)\cap B } \left\{ n^4 \WA_{\rm out}(x) + { \sum_{y \in (N(x)\cap B)\setminus \{z\} }} S_x\left(\BBTree_y\right) \right\} 
\right) 
\nonumber 
\enspace.
\end{align}
Repeating the above steps deeper on $\BBTree$, and using induction, we find a worst-case path emanating from $z$, $P^* = (w_0 = z,  \ldots, , w_\ell)$, such that
\begin{align}
S_{u^*}\left(\BBTree\right)
&\leq  \Inf_{\{u^*, z\}}\cdot  n^4
\cdot
\left(
\sum_{j=0}^\ell  \WA_{\rm out}(w_j)
\prod_{i=0}^{j-1} 
\left [ \WA(w_i)-\WA_{\rm out}(w_i)\right]
\right)\enspace.  \label{eq:QvTGenBound}
\end{align}
For $r=\log d$, it is direct that
\begin{equation}\label{eq:SvBB} 
S_A \leq   \Inf_{\{u^*, z\}}\cdot n^4 \cdot\sum_{j\geq r+1}   \WA_{\rm out}(w_j)
\prod_{i=0}^{j-1}   
\left[ \WA(w_i)-\WA_{\rm out}(w_i) \right] \enspace.
\end{equation}
Since for every vertex $w\in B$ we have $0\leq \WA_{\rm out}(w) \  \leq 1-\varepsilon/2 <1$, we get that
\begin{align}\label{eq:QBDef}
S_A & \leq \Inf_{\{u^*, z\}}\cdot n^4\cdot  
\sum_{j\geq r+1}  \prod_{i=0}^{j-1}  \left[\WA(w_i)-\WA_{\rm out}(w_i) \right] \leq  n^4\cdot  
\sum_{j\geq r+1}  \prod_{i=0}^{j-1}  \WA(w_i) \enspace.
\end{align}
The lemma follows by bounding appropriately the magnitude of each summand in \eqref{eq:QBDef}, separately. 

For $j\geq 1$, let $M_j$ be the set of heavy vertices in $\{w_0,\ldots, w_{r+j-1}\}$, i.e, those vertices $w$ with $\WA(w) > 1 - \varepsilon/2$. Letting also $m_j=|M_j|$, we define
\begin{align}
\Upphi(j) &=    \prod_{i=0}^{r+j-1}   \WA(w_i) \leq  \left( 1-\varepsilon/2 \right)^{r+j-m_j} \prod_{w\in M_j}\WA(w)\enspace.   
\label{eq:FirstBound4Rj}
\end{align}
To bound $\Upphi(j)$, we need to argue about  $\prod_{w\in M_j}\WA(w)$. 
Using Corollary \ref{cor:HeavyWABound} we get that
\begin{eqnarray}\label{eq:ProdDegreeBound}
\textstyle \prod_{w\in M_j}\WA(w) \leq d^{-m_j}\left( 1-\varepsilon/4\right)^{-r-j+m_j} \enspace.
\end{eqnarray}	
Plugging (\ref{eq:ProdDegreeBound}) into (\ref{eq:FirstBound4Rj}) we get that
\begin{align}
\Upphi(j)  &\leq  \left(\frac{1-\varepsilon/2}{1-\varepsilon/4}\right)^{r+j-m_j} d^{-m_j}  \leq 
\left(1-\frac{\varepsilon}{4}\right)^{r+j-m_j} d^{-m_j} 
\leq d^{-\frac{\varepsilon}{4}} \left(1-\frac{\varepsilon}{4}\right)^{j}. 
\label{eq:RJFinallBound} 
\end{align}
In the last inequality we use that $(1-\varepsilon/4)^{r}< d^{-\frac{\varepsilon}{4}} $ and $\left(d/(1-\varepsilon/4)\right)^{-m_j}\leq 1$, since $r=\log d$.
Plugging (\ref{eq:RJFinallBound})  into (\ref{eq:QBDef}), and changing variable, we get
\begin{align}
S_A & \leq  \Inf_{\{u^*, z\}}n^4 \cdot\sum_{j\geq 1} \Upphi(j) \ 
\leq  \Inf_{\{u^*, z\}}  n^4  \cdot d^{-\frac{\varepsilon}{4}}  \sum_{j\geq 0}\left( 1-\varepsilon/4 \right)^{j} 
=   \Inf_{\{u^*, z\}}  n^4 \cdot d^{-\frac{\varepsilon}{4}} \cdot 4\varepsilon^{-1}  
\leq   \Inf_{\{u^*, z\}}  n^4   \cdot d^{-\frac{\varepsilon}{5}}  , \nonumber 
\end{align}
where the last inequality holds for sufficiently large $d$. 
The lemma follows.
\hfill $\Box$

\subsection{Proof of Proposition \ref{prop:ContributionsSB}}\label{sec:prop:ContributionsSB}
If there is no cycle in $B$, then $S_B=0$, hence the proposition is trivially true. 
For what follows, assume that block $B$ contains the short cycle $C$. 

Recall the definition of $\BBTree$  from the proof of Lemma \ref{lemma:ContributionsSA}. That is, $\BBTree$ is the sub-block of $B$ which  contains all vertices that are reachable from $u^*$ through a path inside $B$ that does not pass through any vertex of $C$. 
Let $w$ be the vertex in the cycle $C$ that is closest to $z$. Note that $w$ is not in $\BBTree$.  
Let  the set $\Lambda$ consists of all vertices in $B$ apart from $w$, that do not belong to $\BBTree$. 

\begin{lemma}\label{lemma:DisWeightCycleStuff}
 Let $\mathbb{P}$ be the set of paths from $w$ to the external vertices in $\Lambda$, i.e., that  have at least one neighbour in  $\partial_{\rm out} B$.
We have that 
\begin{align*}
S_{w}(\Lambda)& 
\leq  2 \Inf_{\{u^*,z\}} \cdot n^4 \cdot \max_{{\rm P}\in \mathbb{P}} \left\{
\sum^{|P|}_{j=1} \WA_{\rm out}(u_j)\prod^{j-1}_{i=1}[\WA(u_i)-\WA_{\rm out}(u_i)]\right\}\enspace. 
\end{align*}
\end{lemma}

Repeating the same line of argument to that we used in \eqref{eq:QvTGenBound} (proof of Lemma \ref{lemma:ContributionsSA}) we get that
\begin{align}
S_B & \leq 2 \Inf_{\{u^*, z\}}\cdot n^4\cdot  
\sum_{j\geq \BCDist+1}  \prod_{i=0}^{j-1}  \left[\WA(u_i)-\WA_{\rm out}(u_i) \right] \leq 2\Inf_{\{u^*, z\}}\cdot  n^4\cdot  
\sum_{j\geq \BCDist+1}  \prod_{i=0}^{j-1}  \WA(u_i) \enspace,
\end{align}
where $\BCDist\geq \log^5 d$ is the distance from $u^*$ to the cycle in $B$. 
Then, using arguments which are almost identical to those in the proof of Lemma \ref{lemma:ContributionsSA} we conclude that
\begin{align}
S_B &\leq  \Inf_{\{u^*, z\}}\cdot n^4 \cdot d^{-12} \enspace. \nonumber 
\end{align}
All the above conclude the proof of Proposition \ref{prop:ContributionsSB}. \hfill $\Box$

\subsubsection{Proof of Lemma \ref{lemma:DisWeightCycleStuff}}\label{sec:lemma:DisWeightCycleStuff}

Let $x_0, x_1, \ldots, x_{k}$ be the vertices in the cycle $C$, for some integer $k\geq 3$.  Also, we identify $w$ with the vertex $x_0$, i.e., $w=x_0$.

Let $M$ be the set of vertices in $B$ that are outside $\Lambda$  but 
have at least one neighrbour inside  $\Lambda$.
Note that  $w\in M$. Furthermore,  our assumptions imply that the graph induced by $\Lambda$ 
a tree as it does not contain the whole cycle $C$, recall that  $w$ is not included in $\Lambda$.


We need to show that there exists a coupling $\upkappa$ of the marginals
$\mu_{\Lambda}(\cdot |\ M, \sigma^+)$ and  
$\mu_{\Lambda}(\cdot\  |\ M, \sigma^-)$,
where $\sigma^+, \sigma^-$ are configurations at $M$  that disagree only at $w$
such that 
\begin{align}\label{eq:target4lemma:DisWeightCycleStuff}
\lefteqn{
S_{w}(\Lambda)=\Exp\left[ 
\left .  {\cD}(\Lambda, X_{t+1}, Y_{t+1})  \ \right |
\    X_{t+1} (M)=\sigma,\  
Y_{t+1}(M)=\tau,  B\textrm{ updated at $t+1$} 
\right] } \hspace{6cm} \nonumber \\
 & \leq 2 n^4 \cdot \max_{{\rm P}\in \mathbb{P}} \left\{ \sum^{|P|}_{j=1} \WA_{\rm out}(u_j)\prod^{j-1}_{i=1}[\WA(u_i)-\WA_{\rm out}(u_i)] \right \}\enspace. 
\end{align}
For  $\sigma^+$ we imply that $\sigma^+(w)=+1$ and, similarly, for  $\sigma^-$ we imply
that $\sigma^-(w)=-1$.  

For the sake of the analysis, rather than considering set $M$, we consider
a slightly different set for the boundary of $\Lambda$. We call this new set $\widehat{M}$. 
Specifically, we introduce a new vertex $z$ in the cycle $C$ between
$w$ and $x_1$. That is, instead of having $w$ connected to $x_1$, we have that $w$ is connected to 
$z$ and, in turn,  $z$ is connected to $x_1$.   

Note that for any configuration $\sigma$ at $M$ we have that 
the measure  $\mu_{\Lambda}(\cdot |\ M, \sigma)$ is identical to the measure
$\mu_{\Lambda}(\cdot |\ \widehat{M}, \widehat{\sigma})$, where $\widehat{\sigma}$
is the configuration at $\widehat{M}$ such that for every $u\in M\cap \widehat{M}$
we have $\sigma(u)=\widehat{\sigma}(u)$, and $\widehat{\sigma}(z)=\widehat{\sigma}(w)$. 

We consider below a  coupling $\upkappa_1$ between the distributions  
$\mu_{\Lambda}(\cdot |\ \widehat{M}, \widehat{\sigma}^{++})$ and  
$\mu_{\Lambda}(\cdot\  |\ \widehat{M}, \widehat{\sigma}^{+-})$, 
where $\widehat{\sigma}^{++}, \widehat{\sigma}^{+-}$ are configurations at $\widehat{M}$ such that for every $u\in M\cap \widehat{M}$ we have $\widehat{\sigma}^{+-}(u)=\widehat{\sigma}^{++}(u)=\sigma^{+}(u)$, while  $\widehat{\sigma}^{++}(z)=+1$ and $\widehat{\sigma}^{+-}(z)=-1$. 

Furthermore, we also consider below a coupling $\upkappa_2$ between the distributions  
$\mu_{\Lambda}(\cdot |\ \widehat{M}, \widehat{\sigma}^{--})$ and  
$\mu_{\Lambda}(\cdot\  |\ \widehat{M}, \widehat{\sigma}^{+-})$, 
where $\widehat{\sigma}^{+-}$ is as defined above, while 
$\widehat{\sigma}^{--}$ is a configuration at $\widehat{M}$ such that
for every $u\in M\cap \widehat{M}$ we have $\widehat{\sigma}^{--}(u)=\sigma^{-}(u)$
and  $\widehat{\sigma}^{--}(z)=-1$. 

Note that $ \widehat{\sigma}^{+-}$ and $ \widehat{\sigma}^{--}$ differ only on the configuration 
at $w$. 

It is immediate that we can compose $\upkappa_1$ and $\upkappa_2$ to obtain the coupling $\upkappa$. That is,
we generate a configuration $\bY_A$ from $\mu_{\Lambda}(\cdot |\ M, \sigma^+)$ 
(which is identical distribution to $\mu_{\Lambda}(\cdot |\ \widehat{M}, \widehat{\sigma}^{++})$). 
Then,  we use $\upkappa_1$ to obtain a configuration $\bY_B$ from 
$\mu_{\Lambda}(\cdot\  |\ \widehat{M}, \widehat{\sigma}^{+-})$ conditional on $\bY_A$. Similarly, we use 
 $\upkappa_2$ to obtain a configuration $\bY_C$ from 
$\mu_{\Lambda}(\cdot\  |\ \widehat{M}, \widehat{\sigma}^{--})$ (which is identical distribution
to $\mu_{\Lambda}(\cdot\  |\ {M},{\sigma}^{-})$ conditional on $\bY_B$. 
The pair $(\bY_A, \bY_C)$ induce the coupling $\upkappa$. 

With a slight abuse of notation, we let
$$
S^{(1)}_w(\Lambda)=\Exp\left[ 
\left .  {\cD}(\Lambda, X_{t+1}, Y_{t+1})  \ \right |
\    X_{t+1} (\widehat{M})=\sigma,\  
Y_{t+1}(\widehat{M})=\tau,  B\textrm{ updated at $t+1$} 
\right] 
$$
where we use coupling $\upkappa_1$ for the configuration of $X_{t+1}(\Lambda)$
and $Y_{t+1}(\Lambda)$.
Similarly, we define the quantity $S^{(2)}_w(\Lambda)$ with respect to the coupling $\upkappa_2$. 

The triangle inequality yields that 
\begin{align}
S_{w}(\Lambda) \leq S^{(1)}_w(\Lambda)+S^{(2)}_w(\Lambda)\enspace. 
\end{align}
We use couplings $\upkappa_1$ and $\upkappa_2$ similar to what we have in 
the proofs of \cref{lemma:Coupling4TreeMeasures,lemma:ContributionsSA}.

Then, \eqref{eq:target4lemma:DisWeightCycleStuff} standard arguments we have seen in the proof of \cref{lemma:ContributionsSA,lemma:Coupling4TreeMeasures}.  
 The lemma follows. 
\hfill $\Box$

\section{Proof of Theorem \ref{thrm:BlockParitionGnp}}\label{sec:thrm:BlockParitionGnp}

We show that each of the three properties defining $\IntrstGraphFam (d, \varepsilon)$ holds with high probability. 

\subsection{Property \ref{itm:propty1}: Existence of Block Partition}
For $\varepsilon, d$ and $\beta$ as specified in Theorem \ref{thrm:BlockParitionGnp},
let $\G=\G(n,d/n)$, while let the set of influences $\{\bInf_e\}$ 
be induced by the Edwards-Anderson model on $\G$, with inverse temperature~$\beta$. 

Recall that a cycle $C$ of $\G$ is short if its length is at most {$4\frac{\log n}{\log^4 d}$}. 
Recall also that  a vertex $u$ of $\G$ is a $(d,\varepsilon)$-block vertex, if for every path $P$ of length at most $\log n$ that emanates from $u$, we have that $\cappedWA(P) < 1$, where $\cappedWA$ is the $(d,\varepsilon)$ path-weight, (see Definition \ref{def:WightOfPath}).

In order to prove Property  \ref{itm:propty1}, we provide an algorithm which, for typical instances of $(\G,\{\bInf_e\})$, outputs an $(d,\varepsilon)$-block partition $\mathcal{B} = \{B_1, \ldots, B_N\}$. 

Specifically, the algorithm takes as input two parameters $\varepsilon, d>0$, and a graph-influence pair $(G, \{\Inf_{e}\})$. Then, it proceeds as follows. 

As a preprocessing step, we determine the set of $(d,\varepsilon)$-block vertices in $(G, \{\Inf_{e}\})$, and the set $\cC$ of short cycles in $G$. Then, the algorithm checks whether the following condition is true:
\begin{condition}\label{eq:CHeck0Alg}
The distance of any two cycles in $\cC$ is at least $\textstyle 2\frac{\log n}{\log^2 d}$.
\end{condition}
If Condition \ref{eq:CHeck0Alg} is false, then the algorithm terminates and reports failure. 

At the next step the algorithm constructs the blocks containing the short cycles in $\cC$. 
The aim is for each cycle $C \in \cC$ to have its own block $B_C$. Let us denote with $\hatC$ the set of all vertices within distance $\log^5 d$ from $C$. To proceed with the construction of unicyclic blocks, we need the following condition to be true:

\begin{condition}\label{eq:CHeck1Alg}
For every $C \in \cC$, for every vertex $w$ at distance $\geq 2 \frac{\log n}{\sqrt{d}}$ from $\hatC$, every path that connects $\hatC$ to  $w$ contains at least one $(d, \varepsilon)$-block vertex .
\end{condition}

If Condition \ref{eq:CHeck1Alg} is false, then the algorithm terminates and reports failure. 
The block $B_C$ contains $\hatC$, and all vertices reachable from $\hatC$ through a path $(v_0,\ldots, v_{\ell})$ such that none of the $v_i$'s, for $i=1,\ldots, \ell$, is a $(d, \varepsilon)$-block vertex.  Note that due to Condition \ref{eq:CHeck1Alg}, we must have that
\begin{align}\label{eq:CHeck1AlgCor}
B_{C} \subseteq  N\left(\hatC, \textstyle{2\frac{\log n}{\sqrt{d}}}\right) , \; \text{ for all } C \in \cC \enspace,
\end{align}
where $N(S,r)$ is the set of vertices in $G$ that are reachable from $S$ via a path of length at most $r$.

In the next stage, the algorithm constructs the tree-structured blocks. Let $U$ be the set of vertices $u$ that are not contained in any block constructed in the previous step, and are such that $\WA(u) > 1-\frac{\varepsilon}{2}$, i.e., $u$ is a ``heavy" vertex. For each $u \in U$, we check whether the following condition is true:

\begin{condition}\label{eq:CHeck2Alg}
For every $u \in U$, and for every vertex $w$ at distance $\geq 4 \frac{\log n}{\sqrt{d}}$ from $u$, every path that connects $u$ to  $w$ contains at least one $(d, \varepsilon)$-block vertex .
\end{condition}

If Condition \ref{eq:CHeck2Alg} is false, then the algorithm terminates and reports failure. 

We create the tree blocks iteratively, as follows: For  each vertex $u\in U$ 
whose block has not been specified  yet, the algorithm creates the block $B_u$ that consists of $u$, and all vertices $w$ that are reachable from $u$ through a path $(v_0,\ldots, v_{\ell})$ such that none of the $v_i$'s, for $i=1,\ldots, \ell$, is a $(d, \varepsilon)$-block vertex. Note that due to Condition \ref{eq:CHeck2Alg}, we must have 
\begin{align}\label{eq:CHeck2AlgCor}
B_{u} \subseteq  N\left(u, \textstyle{4\frac{\log n}{\sqrt{d}}}\right) , \; \text{ for all } u \in U \enspace.
\end{align}

Finally, for every vertex $w$ of $G$, which is not included in any block constructed so far, we define $B_w=\{w\}$.  The algorithm concludes by returning the set of blocks $\cB$ comprised by all the blocks created in each of the three steps. 

Theorem \ref{thm:algGivesPart} below, which we prove in Section \ref{sec:thm:algGivesPart},  establishes that the algorithm successfully returns a block-partition $\cB$ with high probability over the instances $(\G,\{\bInf_e\})$.

\begin{theorem}\label{thm:algGivesPart}
For any $\varepsilon>0$, there exists $d_{0}=d_0(\varepsilon)\ge 1$ such that for any
$d\ge d_0$, and for any
$
0<\beta \le (1-\varepsilon)\beta_c(d), 
$
the following is true:

Let $\G=\G(n,d/n)$, and let the set of influences $\{\bInf_e\}$ 
be induced by the Edwards-Anderson model on~$\G$, with inverse temperature $\beta$. Then, 
\begin{align}
\Pr\left[ \text{ the above algorithm successfully returns a block-partition }\right] \ge 1-n^{-{2/3}}
\end{align}
\end{theorem}
In light of Theorem \ref{thm:algGivesPart}, all it remains to show is that when the algorithm returns a partition $\cB$, then, $\cB$ is indeed a $(d,\varepsilon)$-block partition, i.e., $\cB$  satisfies all properties of Definition~\ref{def:BlockDecomp}. 

First, note that for every singleton block $B_w=\{w\}$, vertex $w$ must be a $(d,\varepsilon)$-block, since otherwise $w$ would have been considered during the two first stages of the algorithm. 

If $B$ is a multi-vertex block, i.e., was created in the first or second stage of the algorithm, then, by construction, every vertex of $\partial_{\rm out}B$ must be a $(d,\varepsilon)$-block vertex, establishing property~\eqref{itm:BPblockBoundary}.

We observe that no two blocks intersect. Unicyclic blocks do not intersect due to Condition~\ref{eq:CHeck0Alg} and \eqref{eq:CHeck1AlgCor} implying that \eqref{itm:BuffCond} is true.
The heavy vertices in $U$ do not belong
to any of the unicyclic blocks.  Any two block  $B_z$ and $B_x$ generated at the second stage
cannot intersect with each other as they are separated by using block vertices at their boundaries. 
Hence, we conclude that the set of blocks that is created by the algorithm is a partition of the
vertex set.

We finally notice that $B \cup \partial_{\rm out} B$ should contain the same number of cycles as $B$, yielding that every vertex in $\partial_{\rm out}B$ has exactly one neighbour in $B$, which  implies
\eqref{itm:BPSingleNeighBoundary}. To see why this true, for the sake of contradiction assume the opposite.
That is, there exists a multi-vertex block $B$ and $z\in \partial_{\rm out}B$ such that $z$ has two 
neighbours in $B$. Then, since our blocks are low diameter, our assumption implies that there is a short 
cycle in the vertices induced by $B\cup \partial_{\rm out}B$ that contains $z$. 
This is a contradiction. If $B$ is 
unicyclic, then the existence of this additional short cycle violates Condition~\ref{eq:CHeck0Alg}. On the other hand, if $B$ is a tree, then 
this short cycle must have been considered at the first stage of the algorithm, so it cannot emerge at the second stage of the algorithm.

\subsection{Property \ref{itm:propty2}: Lower Bound on \texorpdfstring{$\compWeight$}{Lg}} To establish Property 2, it suffices to use the  following lemma.
For $\G(n,d/n)$ and the EA model with temperature $\beta$ at this graph, recall  from
\eqref{def:Weitgh4Comparison},   the  following weight over the paths $P$ in $\G$:
\begin{align}
\compWeight(P)& =\beta \sum\nolimits_{e }\beta |\bJ_{e}| + \sum\nolimits_{v} \log\degr(v)\enspace,
\end{align}
where $\bJ_e$ are the couplings in the EA model. 

\begin{lemma}\label{lem:SmallCompWeight}
For any $\varepsilon>0$, there exists $d_{0}=d_0(\varepsilon) \ge 1$ such that for any
$d\ge d_0$, and for any $0<\beta \le (1-\varepsilon)\beta_c(d)$,
 the following is true:

Let $\G=\G(n,d/n)$, and let $\mu$ be the EA model on $\G$. Then,
\begin{align}
\Pr\left[\text{every path $P$ in  $\G$ with $\textstyle{|P| \le \frac{\log n}{\log^4 d}}$ has $\textstyle{\compWeight(P) \le {\frac{\log n}{\log^2 d}}}$ } \right] \ge 1 - n^{-d^{8/10}} \enspace.
\end{align}
\end{lemma}

\begin{proof}[Proof of Lemma \ref{lem:SmallCompWeight}]
Let $\cK$ be the set of paths $P$ in $\G$ with $|P| \le {\frac{\log n}{\log^4 d}}$, while $\compWeight(P)> {\frac{\log n}{\log^2 d}}$. For every path $P$, let  $\One(P)$, and $\Two(P)$ be the set of edges in $\G$ that have exactly one, and exactly two endpoints in $P$, respectively. Let also $\ell = {\frac{\log n}{\log^4 d}}$.

Clearly, each vertex of $P$ contributes at most $n$ edges in $\One(P)$, and each of them with probability~$d/n$. Hence, $|\One(P)|$ is dominated by a ${\tt Binom}(n\ell, d/n)$, and the Chernoff bound gives
\begin{align}
\Pr \left[|\One(P)| \ge 3 d\cdot \ell \right] \le e^{- d\cdot \ell} 
\le n^{-d^{9/10}} \enspace,
\end{align}
where the last inequality holds for large enough $d$. Similarly, we see that $|\Two(P)|$ is dominated by a  ${\tt Binom}(\ell^2, d/n) + \ell$, as there are already $\ell$ edges in $P$, and there can be at most $\ell^2$ edges in total with both endpoints in $P$. Applying the Chernoff bound with $\xi = \frac{n}{\ell} -2$, we see that 
\begin{align}
\Pr \left[|\Two(P)| \ge (1+\xi) \cdot\frac{d\cdot \ell^2}{n} + \ell\right] 
\le \exp\left({- \frac{\xi^2}{2+\xi}}\cdot\frac{d\cdot \ell^2}{n}  \right)
\le n^{-d^{9/10}} \enspace,
\end{align}
where the last inequality holds for large enough $d$. Noticing that
$
(1+\xi) \cdot\frac{d\cdot \ell^2}{n} + \ell \le 3d\cdot\ell
$, from  the union bound, we have that
\begin{align}\label{eq:BoundOnEdgeNumm}
\Pr \left[|\One(P)|+|\Two(P)|\ge  6d\cdot\ell\right] \le 2 \cdot n^{-d^{9/10}} \le n^{-d^{85/100}}
\enspace.
\end{align}
Let us write $\cE_P$ for the conjunction of the events $\{|\One(P)|\le 3d\cdot\ell\} $ and $ \{|\Two(P)|\le  3d\cdot\ell\}$. From the law of total probability we see that
\begin{align}\label{eq:totalProbForUpss}
\Pr\left[\textstyle{\compWeight(P) > \frac{\log n}{\log^2 d}} \right]
\le
\Pr\left[\textstyle{\compWeight(P) > \frac{\log n}{\log^2 d}} \ \middle|\ \cE_P\right] + \Pr[\bar{\cE}_P] \enspace,
\end{align}
where $\bar{\cE}_P$ is the complement of the event ${\cE}_P$. Using the union bound, we have that
\begin{align}
\Pr[\bar{\cE}_P] &\leq \Pr \left[|\One(P)| \ge 3 d\cdot \ell \right]+\Pr \left[|\Two(P)| \ge 3 d\cdot \ell \right] \leq 2n^{-d^{9/10}} \enspace. 
\end{align}
We now focus on bounding $\Pr\left[\compWeight(P) > {\frac{\log n}{\log^2 d}} \ \middle|\ \cE_P\right]$. Specifically, since $\cE_P$ occurs, there are at most $6d\cdot\ell$ edges incident to $P$, and thus, the first term of $\compWeight(P)$ is dominated by a sum of $6d\cdot\ell$  half-normal distributions with parameter $\beta$. Due to Theorem \ref{thm:UpperTailBoundSumOfHalfNormRelative}, we have that
\begin{align*}
\Pr\left[\sum\nolimits_{e }\beta |J_{e}| \ge 2\cdot 6d \cdot \ell \cdot \beta \cdot \sqrt{\frac{2}{\pi}} \ \right]
\le 
\exp\left(-6d\cdot \ell \cdot \frac{1}{\pi}\right) \enspace,
\end{align*}
where the sum is over all edges incident to $P$. Since $0<\beta \le (1-\varepsilon)\beta_c(d)$, we have that for large~enough~$d$
\begin{align}\label{eq:BoundOnHalfSumm}
\Pr\left[{\beta\sum\nolimits_{e }} |J_{e}| \ge 80 \cdot\ell \ \right]
&\le 
n^{-d^{9/10}} \enspace.
\end{align}
For the second term, we first note that
\begin{align}\label{eq:SumOfDegToProd}
 \sum\nolimits_{v\in V(P)} \log\degr(v) = \log\left(\prod\nolimits_{v \in V(P)} \degr (v)\right) \le \ell\cdot\log\left(\ell^{-1}\cdot\sum\nolimits_{v \in V(P)} \degr (v)\right)
 \enspace,
\end{align}
where the last inequality follows from the arithmetic-geometric mean inequality. Since $\cE_P$ occurs, we also have that
\begin{align}
\sum_{v \in V(P)}\degr(v) = |\One(P)| + 2\cdot|\Two(P)| \le 9 d\cdot\ell \enspace.
\end{align}
Plugging the above into \eqref{eq:SumOfDegToProd} we get that
\begin{align}\label{eq:finalBoundonSumDeg}
\sum_{v\in V(P)} \log\degr(v) \le \ell \cdot \log (9d) \enspace.
\end{align}
 Noticing that for large enough $d$, $\ell \cdot \log (9d) + 80 \ell \le \frac{\log n}{(\log d)^2}$, and combining \eqref{eq:BoundOnHalfSumm} and \eqref{eq:finalBoundonSumDeg}, we get 
 \begin{align*}
 \Pr\left[\textstyle{\compWeight(P) > {\frac{\log n}{\log^2 d}}} \ \middle|\ \cE_P\right] \le 
n^{-d^{9/10}}  \enspace,
\end{align*}
which per \eqref{eq:totalProbForUpss}, and \eqref{eq:BoundOnEdgeNumm}, yields $\Pr\left[\compWeight(P) > {\frac{\log n}{(\log d)^2}}\right] \le n^{-d^{85/100}}$.
Finally, we focus on the cardinality of the set $\cK$. It is easy 
to show that
\begin{align}
\Exp[|\cK|]& \leq 2n\cdot d^{{\ell}}\cdot \Pr\left[\textstyle{\compWeight(P) > {\frac{\log n}{(\log d)^2}}}\right] \leq n^{-d^{8/10}}\enspace. \nonumber
\end{align}
The lemma follows from the above and the Markov's inequality.
\end{proof}

\subsection{Property \ref{itm:propty3}: Bound on extremal coupling values}

Let $\cH$ be the event that there is an edge $e$ in $\G(n,d/n)$ such that
the coupling $|\bJ_e|$ is larger than $10\sqrt{\log n}$, or smaller than $n^{-7/3}$. 

\begin{lemma}\label{lemma:HeavyCouplingBound}
We have that $\Pr[\cH]\leq n^{-1/3}$. 
\end{lemma}
\begin{proof}
It is standard to show that 
for the standard normal random variable $\bI$, taking large $n$, we have that 
\begin{align}\label{eq:GaussianVLDeviation}
\Pr\left[| \bI | \geq 10\sqrt{\log n}\right] \leq n^{-5}\enspace. 
\end{align}
Furthermore, we have that
\begin{align}\label{eq:GaussianVSDeviation}
\Pr\left[|\bI| \leq n^{-7/3}\right] &= \frac{1}{\sqrt{2\pi}}\int^{n^{-7/3}}_{-n^{-7/3}} \exp(-\frac{x^{2}}2)dx \leq \frac{1}{\sqrt{2\pi}}\int^{n^{-7/3}}_{-n^{-7/3}} dx \leq 
n^{-7/3}\enspace. 
\end{align}

Furthermore, note that the graph instance $\G(n,d/n)$ has at most $\binom{n}{2}$ edge couplings
which are i.i.d standard Gaussians. Let $N$ be the expected number of the edges
$e$ whose coupling $\bJ_e$ is larger than $10\sqrt{\log n}$. Then, \eqref{eq:GaussianVLDeviation}  and \eqref{eq:GaussianVSDeviation}
imply that 
\begin{align}
\Exp[N]\leq \binom{n}{2} \cdot\left(\Pr\left[|\bI|\geq  10\sqrt{\log n}\right] + 
\Pr\left[|\bI| \leq n^{-7/3}\right] \right) \leq n^{-1/3}\enspace. 
\end{align}
Using the Markov's inequality, we get that $\Pr[\cH]\leq n^{-1/3}$, concluding the proof of the lemma. 
\end{proof}

\subsection{Proof of Theorem \ref{thrm:BlockParitionGnp}}
Consider the  triplet $(\G,\bm{\vecJ}, \beta)$ as this is specified in Theorem \ref{thrm:BlockParitionGnp}.
For $j=1,2,3$, let $\cL_j$ be the event that $(\G,\bm{\vecJ}, \beta)$ does {\em not} satisfy the Property $j$. 
From Theorem~\ref{thm:algGivesPart}, we have that $\Pr[\cL_1]\leq n^{-2/3}$. From Lemma \ref{lem:SmallCompWeight}, we have that
$\Pr[\cL_2]\leq n^{-d^{8/10}}$. Finally, from Lemma \ref{lemma:HeavyCouplingBound}, we have that $\Pr[\cL_3]\leq n^{-{1/3}}$. 

The theorem follows by applying the union bound. That is, we have that
\begin{align}
\Pr\left[\bigcup\nolimits_{j=1,2,3}\cL_j \right] &\leq \sum\nolimits_{j=1,2,3}\Pr\left[\cL_j \right] {\leq n^{-1/4}}\enspace. 
\end{align}

\section{Proof of Theorem \ref{thm:algGivesPart}}\label{sec:thm:algGivesPart}

Let us start with an additional definition.
For some integer $r\geq 0$, a path $P$ of $G$ is  \emph{$r$-{\pure}} if $N_r(P)$, i.e., the set of all 
vertices reachable from a vertex in $P$ via a path of length $r$, induces a tree in $G$.
Furthermore, for $\ell\geq 0$,  we define the set of paths
\begin{align*}
\mathcal{P}_{r,\ell}(G)
    = \left\{ P \ :\ \text{$P$  is $r$-{\pure}, and has length at least }\ell
      \right\} \enspace.
\end{align*}
The following theorem states that for $r=\ell=\frac{\log n}{\sqrt{d}}$, and for large enough fixed $d$, every path in $\mathcal{P}_{r,\ell}(\G(n,d/n))$ contains at least one $(d,\varepsilon)$-block vertex.

\begin{restatable}{thrm}{AllPathsGood}
\label{thm:AllPathsGood}
For any $\varepsilon>0$, there exists $d_{0}=d_0(\varepsilon) \ge 1$ such that for any
$d\ge d_0$, and for any $0<\beta \le (1-\varepsilon)\beta_c(d)$,
 the following is true:

Let $\G=\G(n,d/n)$, and let the set of influences $\{\bInf_e\}$ 
be induced by the Edwards-Anderson model on $\G$, with inverse temperature $\beta$.
Then, for $\ell = \frac{\log n}{\sqrt{d}}$, and $r =\ell$ we have that 
\begin{align}\label{eq:thrm:AllPathsGood}
\Pr\left[ \text{ every path in } \cP_{r,\ell}(\G) \text{  contains a $(d,\varepsilon)$-block vertex } \right] 
\ge 1 - {n}^{-d^{1/8}} \enspace,
\end{align}
where the probability is over the instances of $(\G,\{\bInf_e\})$. 
\end{restatable}

The proof of Theorem \ref{thm:AllPathsGood} appears in Section \ref{sec:thm:AllPathsGood}. We also use the following standard lemma, which we prove in Appendix \ref{sec:lem:SamllUnicyclicGnp}.

\begin{lemma}\label{lem:SamllUnicyclicGnp}
Let $\G = \G(n,d/n)$. There exist $d_0 \ge 1$ such that for every $d\ge d_0$, with probability at least $1-n^{-3/4}$, every set of vertices $S$ in $\G$ with $|S| \le 2\frac{\log n}{\log^2 d}$, spans at most $|S|$ edges in $\G$ .
\end{lemma}

Let us denote with $\cE_\cP$ and $\cE_S$ the desirable events of Theorem \ref{thm:AllPathsGood} and Lemma \ref{lem:SamllUnicyclicGnp}, respectively. 
We show that $\cE_\cP \cap \cE_S$ imply all of the Conditions \ref{eq:CHeck0Alg}, \ref{eq:CHeck1Alg}, and \ref{eq:CHeck2Alg}, and thus, the algorithm successfully returns a block-partition. Per the union bound, we have that 
\begin{align}
\Pr[\cE_\cP \cap \cE_S] \ge 1-n^{-2/3} \enspace ,
\end{align}
yielding Theorem \ref{thm:algGivesPart}.

We start by noticing that $\cE_S$ yields Condition \ref{eq:CHeck0Alg}. To the contrary, assume that there exist two short cycles $C , C^\prime$ within distance $\frac{\log n}{\log^2 d}$ from each other. Let $P$ be a path of minimum length connecting $C$ to $C^\prime$. Then, for sufficently large $d$
\begin{align*}
|C \cup C^\prime \cup P| \le \textstyle 2 \cdot 4 \frac{\log n}{\log^4 d} + \frac{\log n}{\log^2 d} < 2 \frac{\log n}{\log^2 d} \enspace,
\end{align*}
but $C \cup C^\prime \cup P$ contains more at least $|C \cup C^\prime \cup P|+1$ edges, contradiction.

We now show the following

\begin{lemma}\label{lem:SmallDiamB}
If  $\cE_\cP \cap \cE_S$ occurs, then Conditions \ref{eq:CHeck1Alg} and \ref{eq:CHeck2Alg} hold.
\end{lemma}

\begin{proof}[Proof of Lemma~\ref{lem:SmallDiamB}]

Let $B_C$ be a unicyc block created in the first stage of the algorithm due to the cycle $C \in \cC$. Let $P= (v_0,\ldots, v_{\ell})$ be any path that emanates from $\hatC$, and is such that none of the $v_i$'s, for $i=1,\ldots, \ell$, is a $(d, \varepsilon)$-block vertex. Let $q = 2\frac{\log n}{\sqrt{d}}$ We claim that $|P| < q $, which clearly implies Condition \ref{eq:CHeck1Alg}.

Indeed, we first observe that, under $\cE_S$, every path in $\G$ of length $q$ must contain a subpath in $\cP_{r,\ell}(\G)$. If no such subpath exists, then we can find another short cycle $C^\prime$ intersecting $P$. Considering the set of vertices $C\cup C^\prime\cup P$, we see that $|C\cup C^\prime\cup P| < 2 \frac{\log n}{\log^2d }$, which violates our assumption that $\cE_S$ occurs. 

Our assumption that  $\cE_{\cP}$ occurs, now implies that $|P| <  q$, since otherwise, there must exist some $v_i$ in $P$ that is a $(d, \varepsilon)$-block vertex, contradiction.

Condition \ref{eq:CHeck2Alg} follows by the same line of argument as above, with the additional observation every path in $\G$ of length $2q$ either contains a subpath in $\cP_{r,\ell}(\G)$, or connects two short cycles $C, C^\prime$.
\end{proof}

All the above conclude the proof of Theorem \ref{thm:algGivesPart}.
\hfill{$\Box$}

\section{Proof of Theorem \ref{thm:AllPathsGood}}\label{sec:thm:AllPathsGood}

Let us first recall that the path $P$ of $G$ is $r$-{\pure} if $N_r(P)$, i.e., the set of all vertices reachable from a vertex in $P$ via a path of length $r$, induces a tree in $G$.
Let $\ell = \frac{\log n}{\sqrt{d}}$, $r =\ell$, and recall we defined
\begin{align*}
\mathcal{P}_{r,\ell}(G)
    = \left\{ P \ :\ \text{$P$  is $r$-{\pure}, and has length at least }\ell
      \right\} \enspace.
\end{align*}
Also recall that given  $\varepsilon, d>0$,  a vertex $u$ in $G$ is called  $(d,\varepsilon)$-{block vertex}, if for every path $P$ of length at most $\log n$ emanating  from $u$, we have that  
the $(d, \varepsilon)$-weight of $P$ is less than~1, i.e.,  $\cappedWA(P) < 1$. 

Let the event 
\begin{align*}
\cS_A &= \{\text{every path in $\cP_{r,\ell}(\G)$ contains a block vertex}\} \enspace.
\end{align*}
Theorem \ref{thm:AllPathsGood} follows by showing $\Pr[\cS_A]\geq 1-n^{-d^{1/8}}$.

For our proof we will also use the following refinement: Given $\varepsilon, d>0$, and an integer $k\le 
\log n$, we say that a vertex $u$ in $G$ is $(d,\varepsilon)$-{\em block vertex of range} $k$,  
if for every path $P$ of length at most $k$ emanating  from $u$, we have that $\cappedWA(P) < 1$, (so that 
a $(d,\varepsilon)$-{block vertex} is also a $(d,\varepsilon)$-{block vertex of range} $\log n$). 
In the following, we omit $d$ and $\varepsilon$, when they are clear from the context.

We establish Theorem \ref{thm:AllPathsGood} in two steps. 
First, we consider the event
\begin{align*}
\cS_B &= \{\text{every path in $\cP_{r,\ell}(\G)$ contains a block vertex of range $r$}\}\enspace.
\end{align*}
Note that $\cS_B$ is different from $\cS_A$ in that it considers block vertices of range $r$, rather than standard block vertices.  We show that $\Pr[\cS_B] \ge {1 - n^{-d^{10/75}}}$.

Next, we consider the event
\begin{align*}
\cS_C &= \{\text{for every path $P$ in $\G$, with $r \le |P| \le \log n$, we have $\cappedWA(P) < 1$}\}
\end{align*}
We prove that $\Pr[\cS_C] \ge 1 - n^{-d^{1/5}}$.

Note that $\cS_C\cap \cS_B$ implies $\cS_A$. This follows from the observation that, under the event $\cS_C$ every block vertex of range $r$ is also a standard block vertex. 
Theorem \ref{thm:AllPathsGood} follows from the aforementioned bounds on $\Pr[\cS_C]$ and $\Pr[\cS_B]$.

More specifically, in \cref{sec:lem:ClacOfR} we show the following result. 
 
\begin{lemma}\label{lem:ClacOfR}
 Under the hypotheses of Theorem \ref{thm:AllPathsGood}, we have that $\Pr[\cS_C] \ge 1 - n^{-d^{1/5}}$.
\end{lemma}
 
 Moreover, in the Section \ref{subsec:IndConstr}, we prove the following bound for the probability of $\cS_B$.
 
 \begin{theorem}\label{thm:AllPathsGoodPart1}
 Under the hypotheses of Theorem \ref{thm:AllPathsGood}, we have that $\Pr[\cS_B] \ge {1 - n^{-d^{10/75}}}$.
 \end{theorem}
 
Observing that $\cS_A \supseteq \cS_A\cap\cS_B$, and applying the union bound gives that:
\begin{align*}
\Pr[\cS_A] \ge \Pr[\cS_C \cap \cS_B] =1- \Pr\left[\ \overline{\cS_C} \cup \overline{\cS_B}\ \right]
\ge 1 - n^{-d^{1/5}} - n^{-d^{10/75}}
\ge 1 -  n^{-d^{1/8}} \enspace.
\end{align*}
Theorem \ref{thm:AllPathsGood} follows.  \hfill $\Box$

\subsection{Proof of Theorem \ref{thm:AllPathsGoodPart1}}\label{subsec:IndConstr}

For $\ell = \frac{\log n}{\sqrt{d}}$, $r = \ell$, and for every $(\ell+1)$-tuple $P =(v_0, v_1, \ldots, v_{\ell})$ of vertices inducing a path in $\G$, we define the events
\begin{align*}
\cC_P &= \{ P \in \mathcal{P}_{r,\ell}(G)\} \enspace, &&\text{and}&&
\cE_P &= \left\{P \text{ contains less than $({4}/{10}) \ell$ block vertices of range $r$}\right\} \enspace.
\end{align*}
We wish to show that $\cE_P \cap \cC_P$ is unlikely to occur. To this end, we introduce a new probabilistic construction, that gives rise to a tree-model $(\GWT, \beta, \{\pmb{\UpJ}_e\})$ enjoying more independence properties than $(\G, \beta, \{\bm{\bJ}_e\})$, using {\em Galton-Watson trees} (GW trees for short). 

Given a distribution $\zeta:\mathbb{Z}_{\geq 0}\to [0,1]$ over the non-negative integers, and a single vertex $u$, recall that the (random) tree $\mathbf{T}$ is a Galton-Watson tree rooted at $u$ with offspring distribution $\zeta$, if the root of $\mathbf{T}$ is vertex $u$, and
the number of children for each vertex in $\mathbf{T}$ is distributed according to $\zeta$, independently from the other vertices.

For an integer $n\ge 1$, a number $d>0$, and an $(\ell+1)$-tuple of vertices $P=(v_0, \ldots, v_{\ell})$, we define the following random process giving rise to the triplet $(\GWT, \beta, \{\pmb{\UpJ}_e\})$, where $\GWT = \GWT(P, n ,d)$ is a tree, and each edge $e \in E(\GWT)$ has coupling  $\pmb{\UpJ}_e$. Specifically, starting from the set $P$ we execute the following steps

\begin{enumerate}
\item we (deterministically) add the edges $\{v_{i-1},v_{i}\}$, for every $i \in \{1,\ldots, \ell\}$, making $P$ a path,
\item from each $v_i$ in $P$, we hang a GW tree rooted at $v_i$, with offspring distribution {${\tt Binom}(n, d/n)$},
\item we truncate each GW tree at depth $r$
\item for every edge $e$ of $\GWT$, we sample $\pmb{\UpJ}_e$ according to $N(0,1)$, independently.
\end{enumerate}
We also define the following event 
\begin{align}\label{eq:DefineET}
\cE_{\GWT}  &= \{P \text{ contains fewer than $({4}/{10}) \ell$ block vertices of range $r$ in } \GWT \} \enspace,
\end{align}
and claim that the following is true.

\begin{lemma}\label{lem:couplingDomin}
There exists a coupling between $(\G,\{\bJ_e\})$ and $(\GWT,\{\pmb{\UpJ}_e\})$, such that $\Ind_{\cE_P} \times \Ind_{\cC_P} \le \Ind_{\cE_\UpT}$.
\end{lemma}

\begin{proof}


We consider the following ``BFS tree of the path $P$" in $\G$, which we denote with $\ST$.

Specifically, for $i = 0, 1 , \ldots, \ell$, we obtain a tree $\ST_i$ of depth $r$, which is rooted at vertex $v_i$ using the following recursive exploration procedure. 

For $ 0 \le k \le r-1$, assume that the first $k$ levels of $\ST_i$ have been determined, and let $\{x_1, \ldots, x_m\}$ be the vertices of its $k$-th level. To obtain the level $(k+1)$ of $\ST_i$ we proceed as follows.

For $j = 1, \ldots, m$, we reveal $x_j$'s  neighbors in $\G$, excluding the connection with vertices that have been already explored, or belong to the path $P$. That is, the children of $x_j$ in $\ST_i$ will be precisely its neighbors in $\G$ that do not belong to: (i) any tree $\ST_{l}$ for  $l < i$,  (ii) any of the vertices that have been revealed to belong in $\ST_i$, up to this point,  (iii) the path $P$.

The tree $\ST$ is simply the tree obtained by union of the $\ST_i$'s and $P$.

Clearly, the offspring of each vertex $u$ in $\ST$, is dominated by ${\tt Binom}(n, d/n)$. Therefore, we can couple the graph structures of $\G$ and $\GWT$, so that $\ST$ is always a subtree of $\GWT$.

Since $\ST$ is a subtree of $\GWT$, let $\widehat{\GWT}$ be a subtree of $\GWT$ which is isomorphic
to $\ST$, while let $h:V(\ST)\to V(\widehat{\GWT})$ be a adjacency preserving bijection. 
For each pair of edges $e=\{u,w\}$ in  $E(\ST)$ and $\hat{e}=\{h(u), h(w)\}$  in 
$E(\widehat{\GWT})$ we couple $\pmb{\UpJ}_e$,  $\bJ_{\hat{e}}$ identically. That is, 
we have $\pmb{\UpJ}_e=\bJ_{\hat{e}}$. 
For the remaining edges $e$ of $\GWT$, we draw $\pmb{\UpJ}_e\sim N(0,1)$, independently.

All the above, complete the coupling construction.

Consider now the subgraph of $\G$ induced by $N_r(P)$, and notice that this will be different from $\ST$ exactly when  $N_r(P)$ contains at least one cycle. Clearly, in that case $\Ind_{\cC_P} = 0$, and thus, $\Ind_{\cE_P} \times \Ind_{\cC_P} \le \Ind_{\cE_{\GWT}}$.

On the other hand, if $N_r(P)$ is a tree,  then we have that $\Ind_{\cC_P} = 1$, and $N_r(P)$ coincides with $\ST$. Since the number of block vertices of $P$ decreases if we add more vertices or more edges in $N_r(P)$, and $\ST$ is a subtree of $\GWT$ we have that
$\Ind_{\cE_P} \times \Ind_{\cC_P} = \Ind_{\cE_P} \le \Ind_{\cE_{\GWT}}$.

All the above conclude the proof of Lemma \ref{lem:couplingDomin}.
\end{proof}

\begin{proposition}\label{prop:ExpectedBadPaths}
For any $\varepsilon>0$, there exists $d_0=d_0(\varepsilon)\ge 1$, such that for all $d\ge d_0$, and 
$0<\beta \le (1-\varepsilon)\beta_c(d)$, the following is true:

For $\ell = \frac{\log n}{\sqrt{d}}$, $r = \ell $, and an $(\ell+1)$-tuple $P=(v_0,\ldots, v_{\ell})$, let $(\GWT, \beta, \{\pmb{\UpJ}_e\})$ be the random tree construction defined on Section \ref{subsec:IndConstr} with respect to $P, d , r$ and $\beta$. Then,
\begin{align*}
\Pr[\cE_{\GWT}] &\le n^{-d^{1/7}}
\enspace,
\end{align*}
where the event $\cE_{\GWT}$ is defined by \eqref{eq:DefineET}.
\end{proposition}

We prove Proposition \ref{prop:ExpectedBadPaths} in Section \ref{sec:AnalysisIndCon}. 
Defining now 
\begin{align}
X = \sum_{P \in V^{(\ell+1) }}  \Ind_{\cE_P} \times \Ind_{\cC_P} \enspace,
\end{align}
we have the following corollary of Proposition \ref{prop:ExpectedBadPaths}
\begin{lemma}\label{cor:BoundOnEX}
For any $\varepsilon>0$, there exists $d_0=d_0(\varepsilon)\ge 1$, such that for every $d\ge d_0$, and 
$0<\beta \le (1-\varepsilon)\beta_c(d)$, we have:
\begin{align*}
\Exp[X] \le n^{-d^{10/75}} \enspace.
\end{align*}
\end{lemma}

\begin{proof}[Proof of \cref{cor:BoundOnEX}]
We have that
\begin{align*}
\Exp[X] &= 
\sum_{P \in V^{(\ell+1)}}  \Exp\left[
            \Ind_{\cE_P} \times \Ind_{\cC_P} 
        \right] 
        \le
n^{(\ell+1)} \cdot \left(\frac{d}{n}\right)^{\ell}  
        \cdot
        \Exp\left[
            \Ind_{\cE_P} \times \Ind_{\cC_P} \mid P \text{ is a path}
        \right]         
        \enspace,
\end{align*}
which due to Lemma \ref{lem:couplingDomin}, and Proposition \ref{prop:ExpectedBadPaths}, yields
\begin{align*}
\Exp[X] &\le
n \cdot {d}^{\ell} \cdot  \Exp\left[
            \Ind_{\cE_{\GWT}} 
        \right]
        \le
        n \cdot {n}^{\frac{\log d}{\sqrt{d}}} 
        \cdot  
        n^{-d^{1/7}}
        \enspace.
\end{align*}
Therefore, there exists a $d_0=d_0(\varepsilon)\ge 1
$, such that for every $d \ge d_0$, we have
$
\Exp[X] \le
        {n^{-d^{10/75}}}
$.
\end{proof}

Since $X$ takes non-negative, integral values, Markov's inequality, and Corollary \ref{cor:BoundOnEX} yield  
\begin{align*}
\Pr[X>0] \le \Exp[X] \le {n^{-d^{10/57}}} \enspace,
\end{align*}
concluding the proof of the theorem.

\section{Proof of Proposition \ref{prop:ExpectedBadPaths}}\label{sec:AnalysisIndCon}

Let us start with a few more definitions. Recall that $\ell = \frac{\log n}{\sqrt{d}}$, $r=\ell$.

\begin{definition}[Vertex Weight induced by Path]
Let $P=(v_0, v_1, \ldots, v_{\ell})$ be a path in~$\G$. Then, for each $i\in \{0, \ldots, \ell\}$ we define
the weight
\begin{align}\label{eq:DefineWB}
\WB_P(v_i) & = \max_{Q} \left\{ \cappedWA(Q) \right\}\enspace, 
\end{align}
where $Q$ varies over all paths of length  at most $r$, that emanate from $v_i$ and do not intersect with $P$, i.e.,  they to do not share vertices.  
\end{definition}

\begin{definition}[Left/Right-Block Vertex]
For $\varepsilon, d >0$, and a path $P=(v_0, \ldots, v_{\ell})$,   we say that a vertex $v_j$ is a
\begin{itemize}
    \setlength\itemsep{0.5em}
	\item $(d, \varepsilon)$-{\em left-block vertex} with respect to $P$, if for all $ t\le j$, we have that
	$
	\prod^{j}_{k=t} \WB_{P}(v_k) <1
	$,
	\item $(d, \varepsilon)$-{\em right-block vertex} with respect to $P$, if for all $ t \ge j$, we have that
	$
	\prod^{t}_{k=j} \WB_{P}(v_k) <1 
	$,
\end{itemize}
where the weights $\WB_P$ are specified with respect to $\varepsilon, d$.
\end{definition}

We have the following lemma.

\begin{lemma}\label{lemma:LRBlockVsBlockVertex} 
For any $\varepsilon, d, \beta>0$, for $\ell = \frac{\log n}{\sqrt{d}}$, $r =\ell$, and an $(\ell+1)$-tuple $P=(v_0,\ldots, v_{\ell})$, let $(\GWT, \beta, \{\pmb{\UpJ}_e\})$ be the random tree construction defined on Section \ref{subsec:IndConstr} with respect to $P, d , r$ and $\beta$.

 If a vertex $v_i$ of $P$, is both $(d, \varepsilon)$-left-block, and $(d, \varepsilon)$-right-block vertex with respect to $P$, then $v_i$ is an 
$(d, \varepsilon)$-block vertex for the tree {${\GWT}$}. 
\end{lemma}

\begin{proof}[Proof of Lemma \ref{lemma:LRBlockVsBlockVertex}]
Let 
$Q$ be an arbitrary path of length $r$, starting at vertex $v_i$. We will show that $\cappedWA(Q) < 1$. Let $v_j$ be the last vertex of $P$ that appears on~$Q$, (note that this is well-defined since~$Q$ starts at $v_i$, and ${\GWT}$ is a tree). 

W.l.o.g. assume that $j\ge i$, using just the fact that $v_i$ is $(d, \varepsilon)$-right-block, we will show that $\cappedWA(Q) < 1$, (if $j\le i$, we use the fact that $v$ is $(d, \varepsilon)$-left-block). Let $Q_1$ be the sub-path of $Q$ staring at $v_i$ and arriving at $v_{j-1}$, and $Q_2$ be the sub-path of $Q$ staring at $v_j$ and arriving at the last vertex of $Q$, (so that $Q$ is the concatenation of $Q_1$ and $Q_2$). We have
\begin{align*} 
     \cappedWA(Q) &=   \cappedWA(Q_1) \cdot \cappedWA(Q_2)  
    =
     \left( \prod^{j-1}_{k=i} \cappedWA(v_{k}) \right)\cdot
     \cappedWA(Q_2)
    \le 
    \left( \prod^{j-1}_{k=i} \WB_{P}(v_{k}) \right)\cdot
     \WB_{P}(v_j)
    <1 \;\enspace,
\end{align*}
where the first inequality follows from the definition of $\WB_P$, and the last inequality is due to the fact that $v_i$ is an $(d, \varepsilon)$-right-block vertex.
\end{proof}

Given an instance of $(\GWT, \beta, \{\pmb{\UpJ}_e\})$ as in Section \ref{subsec:IndConstr}, we define
\begin{align}
{\cE^{\mathrm{left}}_{\GWT}} 
    &= \left\{P \text{ contains fewer than {$({7}/{10})\ell$} vertices that are $(d, \varepsilon)$-left-block in }{\GWT}\right\} \enspace,\\
{\cE^{\mathrm{right}}_{\GWT}}  
    &= \left\{P \text{ contains fewer than {$({7}/{10})\ell$} vertices that are $(d, \varepsilon)$-right-block in }{\GWT}\right\} \enspace,
\end{align}
and 
\begin{equation*}
\cE_{\GWT}
    = \left\{P \text{ contains fewer than {$({4}/{10})\ell$} vertices that are $r$-block in } {\GWT} \right\} \enspace,
\end{equation*}
We will prove the following bounds.

\begin{theorem}\label{thrm:PrbOfManyLeftRight}
For any $\varepsilon>0$, there exists $d_0=d_0(\varepsilon)\ge 1$, such that for all $d\ge d_0$, and 
$0<\beta \le (1-\varepsilon)\beta_c(d)$, the following is true:

For $\ell = \frac{\log n}{\sqrt{d}}$, $r =\ell$, and an $(\ell+1)$-tuple $P=(v_0,\ldots, v_{\ell})$, let $(\GWT, \beta, \{\pmb{\UpJ}_e\})$ be the random tree construction defined on Section \ref{subsec:IndConstr} with respect to $P, d , r$ and $\beta$. Then,
\begin{align}\label{eq:BoundOnSmallNumOfLRB}
    \Pr\left[{\cE^{\mathrm{left}}_{\GWT}}  \right]   \le n^{-d^{1/6}}, && \text{ and }&&  \Pr\left[{\cE^{\mathrm{right}}_{\GWT}} \right]   \le  n^{-d^{1/6}}\enspace,
\end{align}
where the probability is over the instances of $(\GWT, \beta, \{\pmb{\UpJ}_e\})$.
\end{theorem}

Using Theorem~\ref{thrm:PrbOfManyLeftRight}, and the union bound, we further get that 
\begin{equation*}
\Pr\left[\cE_{\GWT} \right] 
\le 
\Pr\left[{\cE^{\mathrm{left}}_{\GWT}}\right] + \Pr\left[{\cE^{\mathrm{right}}_{\GWT}}\right] 
\le  2 \cdot n^{-d^{1/6}} 
\le n^{-d^{1/7}}
\enspace,
\end{equation*}
yielding Proposition \ref{prop:ExpectedBadPaths}.

\section{Proof of Theorem~\ref{thrm:PrbOfManyLeftRight}} \label{sec:thrm:PrbOfManyLeftRight} 

We prove the bound of Theorem~\ref{thrm:PrbOfManyLeftRight} only for $\Pr[{\cE^{\mathrm{left}}_{\GWT}}]$, as the derivations for $\Pr[{\cE^{\mathrm{right}}_{\GWT}}]$ are identical. 

Let us start with a few definitions. We define the set of \emph{heavy} vertices of $P$ by , 
\begin{align*}
H = H(P)=\left\{i\in \{0, \ldots, \ell\}:\WB_{P}(v_i)\ge1\right\} 
\enspace .
\end{align*}
For $i \in H$, let $t_i$ be the greatest index in $\{i, i+1, \ldots, \ell\}$ such that for all $ t \in \{i, i+1, \ldots, t_i\}$
\begin{align*}
\prod_{k=i}^{t}\WB_{P}(v_k) \ge 1 \enspace,
\end{align*}
and let $L_i=\{i, i+1, \ldots, t_i\}$. For $i \notin H$, define $L_i =\emptyset$, and let $ L= L_0 \cup L_1 \cup \ldots \cup L_{\ell}$.

\begin{lemma}\label{lemma:LVsLeftblock}
For all $j \in \{0,1,\dots, \ell\} \setminus L$, vertex $v_j$ is a $(d, \varepsilon)$-left-block vertex with respect to $P$.
\end{lemma}

\begin{proof}
We show the contrapositive, i.e., assuming $v_j$ is not $(d, \varepsilon)$-left-block with respect to $P$, we show that $j \in L$. If $v_j$ is not $(d, \varepsilon)$-left-block, then there must exist at least one index $i \in \{0,1, \ldots, j\}$ such that 
$
 \prod^{j}_{k=i}\; \WB_{P}(v_k) \ge 1 
$.

Let $i^*$ be the greatest such index; we claim that $j \in L_{i^*}$, i.e., for every $t \in \{i^*, i^*+1, \ldots, j\}$ we have that $\prod^{t}_{k=i^*}\; \WB_{P}(v_k) \ge 1 $. Indeed, if there was a $t \in \{i^*, i^*+1, \ldots, j\}$ such that $\prod^{t}_{k=i^*}\; \WB_{P}(v_k) < 1 $, then we would have that $\prod^{j}_{k=t+1}\; \WB_{P}(v_k) \ge 1 $, which contradicts with the fact that $i^*$ is the greatest such index. 
\end{proof}

The proposition below shows that we should expect no more than a small fraction of vertices in $P$ to belong in $L$.

\begin{proposition}\label{prop:BoundingL}For any $\varepsilon>0$, there exists $d_0=d_0(\varepsilon)\ge 1$, such that for all $d\ge d_0$, and 
$0<\beta \le (1-\varepsilon)\beta_c(d)$, the following is true:

For $\ell = \frac{\log n}{\sqrt{d}}$, $r =\ell$, and an $(\ell+1)$-tuple $P=(v_0,\ldots, v_{\ell})$, let $(\GWT, \beta, \{\pmb{\UpJ}_e\})$ be the random tree construction defined on Section \ref{subsec:IndConstr} with respect to $P, d , r$ and $\beta$. Then,
\begin{align}
 \Pr\left[|L|\geq \frac{3}{10}\ell \right]\leq n^{-d^{1/6}} \enspace.  
\end{align}
\end{proposition}
In light of Lemma~\ref{lemma:LVsLeftblock}, and Proposition \ref{prop:BoundingL}, Theorem~\ref{thrm:PrbOfManyLeftRight} follows easily.
Therefore, to finish proving Theorem~\ref{thrm:PrbOfManyLeftRight}, let us prove Proposition \ref{prop:BoundingL}.

\section{Proof of Proposition \ref{prop:BoundingL}}\label{sec:prop:BoundingL}

Using the union bound, we have that
\begin{align}\nonumber
|L|&=\left | \bigcup\nolimits_{i\in H} L_i \right| \leq \sum\nolimits_{i\in H}|L_i|\enspace. 
\end{align}


To estimate each $|L_i|$ for $i\in H$ we think as follows:
Each heavy vertex, $v_i$, introduces a weight, $\cappedWB_P(v_i)>1$. If $L_i = \{i, i+1\ldots,j\}$, then notice that  $\{i, i+1\ldots,j+1\}$ is the first interval on the right of $i$ that ``absorbs'' the weight $\WB_P(v_i)$, i.e., $j$ is the smallest index in $\{0, \ldots, i\}$ such that 
$$
{\prod\nolimits_{k=i}^{j+1} \WB_{P}(v_k) < 1} \enspace.
$$
Let $\theta = \varepsilon/2$. If  $H \cap L_i=\emptyset$,   we observe that $v_i$ requires at most a number of $\log_{(1-\theta)} \cappedWB_P(v_i)$ light vertices to get $\cappedWB_P(v_i)$ absorbed. 
Similarly, if  $ H \cap L_i\neq \emptyset$, then we observe that the number of light vertices needs to be
$$
\sum\nolimits_{j \in H\cap L_i} \log_{(1-\theta)} \WB_P(v_j)\enspace. 
$$
Therefore, with the above union-bound we obtain that 
$|L| \le |H| + \sum_{i\in H} \log_{(1-\theta)} \WB_P(v_i)$.
Hence, we have the following corollary. 

\begin{corollary}\label{lem:BreakIntoPartsToBound}
Let  $ \theta = \varepsilon/2$, then we have that
    \begin{align}\label{eq:BreakIntoPartsToBound}
        |L| 
        \le |H| + \sum\nolimits_{i\in H} \log_{(1-\theta)} \WB_P(v_i)
        \le |H| + \frac{1}{\theta}\cdot\sum\nolimits_{i\in H}\log\WB_P(v_i)\enspace.
     \end{align}
\end{corollary}

Next, we prove tail bounds for each term in the rhs of \eqref{eq:BreakIntoPartsToBound}, using the following technical lemma, which we prove in Section \ref{sec:lem:MGFForcappedWB}.

\begin{lemma}\label{lem:MGFForcappedWB}
For any $\varepsilon>0$, there exists $d_0=d_0(\varepsilon)\ge 1$, such that for all $d\ge d_0$, and 
$0<\beta \le (1-\varepsilon)\beta_c(d)$, the following is true:

For $\ell = \frac{\log n}{\sqrt{d}}$, $r =\ell$, and an $(\ell+1)$-tuple $P=(v_0,\ldots, v_{\ell})$, let $(\GWT, \beta, \{\pmb{\UpJ}_e\})$ be the random tree construction defined on Section \ref{subsec:IndConstr} with respect to $P, d, r$ and $\beta$. Then, for {$t= d^{95/100}$} and {$q= d^{93/100}$}, we have that
\begin{align}
\Exp\left[( \WB_P(v_1) )^t\right]  \le \left(1-\frac{\varepsilon}{4}\right)^{q} \enspace.
\end{align}
\end{lemma}

We start by overestimating $|H|$, as follows

\begin{proposition}\label{prop:BoundOnHevy} 
For any $\varepsilon>0$, there exists $d_0=d_0(\varepsilon)\ge 1$, such that for all $d\ge d_0$, and 
$0<\beta \le (1-\varepsilon)\beta_c(d)$, the following is true:

For $\ell = \frac{\log n}{\sqrt{d}}$, $r =\ell$, and an $(\ell+1)$-tuple $P=(v_0,\ldots, v_{\ell})$, let $(\GWT, \beta, \{\pmb{\UpJ}_e\})$ be the random tree construction defined on Section \ref{subsec:IndConstr} with respect to $P, d, r$ and $\beta$. Then, 
\begin{align}
\Pr\left[|H|\geq \frac{\ell}{d^{1/10}} \right ]\leq n^{-d^{1/4}} 
\enspace,
\end{align}
where $H$ is the set of heavy vertices defined above.
\end{proposition}

\begin{proof}[Proof of Proposition \ref{prop:BoundOnHevy}]

Let $I_1$, and $I_2$ ,be the sets of odd, and even indices of $\{0, \ldots, \ell\}$, respectively. Writing $H_1 =H \cap I_1$, and $H_2 =H \cap I_2$, and using the union bound, it is easy to see that
\begin{align}
\Pr\left[|H|\geq \frac{\ell}{d^{1/10}} \right ] 
\le 
\Pr\left[|H_1|\geq \frac{1}{2}\cdot\frac{\ell}{d^{1/10}} \right ]
+
\Pr\left[|H_2|\geq \frac{1}{2}\cdot\frac{\ell}{d^{1/10}} \right ]
\enspace,
\end{align}
therefore, it suffices to show that
\begin{align*}
\Pr\left[|H_1|\geq \frac{1}{2}\cdot\frac{\ell}{d^{1/10}} \right ] \le \frac{1}{2} \cdot n^{-d^{1/4}} 
\enspace.
\end{align*}

Notice that for any two indices $i, j \in I_1$, the corresponding random variables $\WB_{P}(v_i), \WB_{P}(v_j)$, are i.i.d., and thus, each index $i \in I_1$ belongs to $H_1$ with probability at most $\Pr[\WB_{P}(v_1) \ge 1]$. 
Specifically, per Lemma~\ref{lem:MGFForcappedWB}, and Markov's inequality, we have that for {$t= d^{95/100}$} and {$q= d^{93/100}$}
\begin{align*}
\Pr\left[\WB_{P}(v_1) \ge 1 \right]
\le
{\Exp\left[\left(\WB_{P}(v_i)\right)^t\right]} \le \left(1-\frac{\varepsilon}{4}\right)^{q}
\enspace.
\end{align*}

Therefore, $|H_1|$ is upper bounded by the number of successes of a binomial distribution with  
$\ell/2$ number of trials, and probability of success $\left(1-\frac{\varepsilon}{2}\right)^{q}$. 
 Expanding the tail-probability of aforementioned distribution, and bounding appropriately, we get that 
\begin{align*}
    \Pr\left[|H_1|\ge  \frac{1}{2}\cdot\frac{\ell}{d^{1/10}}\right]
&=
\Pr\left[|H_1|\ge \frac{1}{2}\cdot\frac{\log n}{d^{3/5}}\right]
\\
&\le
\binom{\frac{\log n}{2\cdot d^{1/2}}}{\frac{\log n}{2\cdot d^{3/5}}} \cdot
\exp\left(\log {\left(1-\frac{\varepsilon}{4}\right)}\cdot d^{93/10} \cdot\frac{\log n}{d^{3/5}}\right)\\
&\le
\left(e\cdot {{d^{1/10}}}\right)^{\frac{\log n}{d^{3/5}}} \cdot
\exp\left(\log {\left(1-\frac{\varepsilon}{4}\right)}\cdot d^{86/100} \cdot{\log n}\right)\\
&\le
\exp\left(\log n\left[
\frac{4\log d}{d^{3/5}}
+\log {\left(1-\frac{\varepsilon}{4}\right)}\cdot d^{86/100}\right]\right)
\enspace.
\end{align*}
So that for {$d \ge d_0(\varepsilon)$}, we have that
\begin{align*}
  \Pr\left[|H_1|\ge  \frac{1}{2}\cdot\frac{\ell}{d^{1/10}}\right]
  \leq n^{-d^{4/5}} \le \frac{1}{2} \cdot n^{-d^{1/4}}
\enspace,
\end{align*}
as desired, concluding the proof of Proposition \ref{prop:BoundOnHevy}.
\end{proof}

We have the following bound on the upper tail of $\sum_{i \in H} \log \WB_P(v_i)$.

\begin{theorem}\label{thm:BoundOnWBThm}
For any $\varepsilon>0$, there exists $d_0=d_0(\varepsilon) \ge 1$, such that for all $d\ge d_0$, and 
$0<\beta \le (1-\varepsilon)\beta_c(d)$, the following is true:

For $\ell = \frac{\log n}{\sqrt{d}}$, $r =\ell$, and an $(\ell+1)$-tuple $P=(v_0,\ldots, v_{\ell})$, let $(\GWT, \beta, \{\pmb{\UpJ}_e\})$ be the random tree construction defined on Section \ref{subsec:IndConstr} with respect to $d, \ell , r$ and $\beta$. Then,  
we have that
\begin{equation*} 
\Pr \left[ \sum_{i \in H} \log \WB_P(v_i) \ge \frac{\ell}{d^{1/50}} \middle| 
|H|\le\frac{\ell}{d^{{1}/{10}}}\right] \le 
 n^{-d^{1/3}}\enspace.
 \end{equation*}
\end{theorem}

Given Theorem \ref{thm:BoundOnWBThm}, we establish Proposition \ref{prop:BoundingL} as follows. 

\begin{proof}[Proof of Proposition \ref{prop:BoundingL}]

 Write $\theta = \varepsilon/2$, due to \cref{lem:BreakIntoPartsToBound}, we have that 
\begin{align}
 \Pr\left[|L|\geq \frac{3}{10}\ell \right]\leq 
  \Pr\left[|H|  + \frac{1}{\theta}\cdot\sum_{i\in H}\log\WB_P(v_i)\geq \frac{3}{10}\ell \right]
 \enspace.   
\end{align}
From the law of total probability it is easy to see that
\begin{multline}
\Pr\left[|H|+ \frac{1}{\theta}\cdot\sum_{i\in H}\log\WB_P(v_i)\geq \frac{3}{10}\ell \right]\leq
\\
     \Pr\left[|H|\ge \frac{\ell}{d^{1/10}}\right]
     +
     \Pr\left[|H| + \frac{1}{\theta}\cdot\sum_{i\in H}\log\WB_P(v_i)\geq \frac{3}{10}\ell \middle||H|\le \frac{\ell}{d^{1/10}} \right] \enspace.
    \label{eq:FirstApproxOFprobL}
\end{multline}
Moreover, we also have that for {appropriately large} $d$
\begin{align*}
\Pr\left[|H| + \frac{1}{\theta}\cdot\sum_{i\in H}\log\WB_P(v_i)\geq \frac{3}{10}\ell \middle||H|\le \frac{\ell}{d^{1/10}} \right]
\le
\Pr\left[\sum_{i\in H}\log\WB_P(v_i)\ge\frac{\ell}{d^{1/50}}\middle||H|\le \frac{\ell}{d^{1/10}}\right] \enspace.
\end{align*}
Applying Theorem \ref{thm:BoundOnWBThm}, and Proposition \ref{prop:BoundOnHevy} in the above, we get that for every $d\ge d_0$
\begin{align*}
\Pr\left[|H|  + \frac{1}{\theta}\cdot\sum_{i\in H}\log\WB_P(v_i)\geq \frac{3}{10}\ell \right]\le n^{-d^{1/3}} +  {n}^{-d^{1/4}} 
\le {n}^{-d^{1/6}}
\enspace,
\end{align*}
concluding the proof of Proposition \ref{prop:BoundingL}.
\end{proof}

Let us now prove Theorem \ref{thm:BoundOnWBThm}.

\begin{proof}[Proof of Theorem \ref{thm:BoundOnWBThm} ]
Similarly to the proof of Proposition \ref{prop:BoundOnHevy}, we let $I_1$, and $I_2$ ,be the sets of odd, and even indices of $\{0, \ldots, \ell\}$, respectively. Writing $H_1 =H \cap I_1$, and $H_2 =H \cap I_2$, and using the union bound, we see it is sufficient to prove that
\begin{align*} 
2 \cdot \Pr \left[ \sum\nolimits_{i \in H_1} \log \WB_P(v_i) \ge \frac{1}{2}\cdot\frac{\ell}{d^{1/50}} \middle| |H_1|\leq \frac{\ell}{d^{{1}/{10}}}
	\right] \le
 n^{-d^{1/3}}\enspace.
 \end{align*}
From Markov's inequality we get that for every $t\ge0$
\begin{align}\label{eq:MarkovOnWB}
\Pr \left[ \sum_{i \in H_1} \log \WB_P(v_i) \ge \frac{1}{2}\cdot\frac{\ell}{d^{1/50}} \middle| 
|H_1|\leq \frac{\ell}{d^{{1}/{10}}}
	\right] 
 &\le 
\frac{\Exp\left[\exp( t\cdot \sum_{i \in H_1}\cdot\log\WB_P(v_i)) 
\middle| |H_1|\leq \frac{\ell}{d^{{1}/{10}}} \right]}
	{\exp\left(\frac{t}{2\cdot d^{1/50}} \cdot \ell \right)}
 \nonumber
 \\
 &= 
\frac{
\left(\Exp\left[\exp(t\cdot\log\WB_P(v_1))\ | v_1 \in H_1\right]\right)
        ^{{\ell} \cdot d^{-1/10}}
    }
	{\exp\left(\frac{t}{2\cdot d^{1/50}} \cdot \ell \right)} \enspace,
\end{align}
where the equality follows from the fact that the random variables $\WB_{P}(v_j)$'s where $j\in H_1$ are i.i.d.

Let us write $v$ instead of $v_1$, and notice that since $\WB_P(v)$ is non-negative, we also have
\begin{align}
\label{eq:HoldForNoNegativeRatio}
\Exp\left[\exp(t \cdot \log \WB_P(v) )\ |\ v \in H_1\right] 
&\le \frac{\Exp\left[\left( \WB_P(v) \right)^t\right]}
        {\Pr[\WB_P(v)\ge 1]}\enspace. 
\end{align}
Notice that an upper bound for the enumerator of \eqref{eq:HoldForNoNegativeRatio} is provided by Lemma \ref{lem:MGFForcappedWB}.
Let us now bound the denominator of \eqref{eq:HoldForNoNegativeRatio}, i.e., $\Pr[\WB_P(v) \ge 1] $. First, recalling the definitions of $\WA, \cappedWA$, and $\WB$, it is easy to see that
\begin{align*}
   \Pr[\WB_P(v) \ge 1] 
   \ge {\Pr\left[\,\cappedWA(v) \ge 1\right]} 
   \ge {\Pr[\WA(v) \ge 1]} 
   \ge \Pr[\WA(v) \cdot \mathds{1}\{d< \degr(v) \le 2d\}\ge 1] \enspace,
\end{align*}
so that using Baye's rule, we get
\begin{align}\label{eq:lowerBoundWeightB}
   \Pr[\WB_P(v) \ge 1] 
   \ge
   \Pr[{\WA(v) \ge 1} \mid d \le \degr(v) \le 2d] \cdot \Pr[ d \le \degr(v) \le 2d]
    \enspace.
\end{align}
Let us now focus on the first factor in the rhs of \eqref{eq:lowerBoundWeightB}. As we show in the proof of Theorem \ref{theorem:TailBound4WA}, 
\begin{equation*}
\WA(v) = \sum_{i=1}^{\degr(v)} \left|\tanh\left(\frac{\beta}{2} \bJ_i \right)\right|\enspace,
\end{equation*}
where each $\bJ_i$ follows the standard normal distribution. With that in mind, we underestimate the probability $\Pr[\WA(v) \ge 1 \mid d \le \degr(v) \le 2d]$ as follows
\begin{align*}
\Pr[\WA(v) \ge 1 \mid d \le \degr(v) \le 2d] 
&\ge 
\left(\Pr\left[\;\left|\tanh\left(\frac{\beta}{2} \bJ_i \right)\right| \ge 
                                            \frac{1}{d}\;\right]\right)^{2d} \\
&\ge 
\left(2\cdot\Pr\left[\;\bJ_i\ge \frac{2}{\beta}\cdot\mathrm{arctanh}\left(\frac{1}{d}\right)\;\right]\right)^{2d}
\\
&\ge 
\left(2\cdot\Pr\left[\;\bJ_i\ge \frac{2}{\beta}\cdot\frac{1}{2}\cdot\left(\frac{1+({1}/d)}
{1-({1}/d)}-1\right)\;\right]\right)^{2d}
\\
&\ge 
\left(2\cdot\Pr\left[\;\bJ_i\ge \frac{2}{\beta\cdot (d-1)}\;\right]\right)^{2d}
\enspace.
\end{align*}
Since $0<\beta \le (1-\varepsilon)\beta_c(d)$, there must exist a $\lambda\in (0,1)$ such that $\beta (d-1) = \sqrt{2\pi}\cdot(1-\lambda)$, and thus, we can rewrite the last inequality as
\begin{align*}
\Pr[\WA(v) \ge 1 \mid d \le \degr(v) \le 2d] \ge \left(2\cdot\Pr\left[\;\bJ_i\ge 
\sqrt{\frac{2}{\pi}} \cdot (1-\theta)^{-1}\;\right]\right)^{2d} \enspace.
\end{align*}
Therefore, for $d$ large enough, there exists and a constant $C_1>0$, such that 
\begin{align}\label{eq:finalUnderEstOfBuProb}
\Pr[\WA(v) \ge 1  \mid d \le \degr(v) \le 2d] \ge \exp(-C_1\cdot d) \enspace.
\end{align}
Regarding the second factor in the rhs of \eqref{eq:lowerBoundWeightB}, we have
\begin{align*}
\Pr[ d < \degr(v) \le 2d] 
&=
\sum_{k=d}^{2d}\binom{n}{k}\left(\frac{d}{n}\right)^{k} \left(1-\frac{d}{n}\right)^{n-k} 
{\ge} 
\sum_{k=d}^{2d}\left(\frac{d}{k}\right)^{k}\cdot e^{-d} \enspace.
\end{align*}
Therefore, for $d$ large enough, there exists a constant $C_2>0$, such that 
\begin{align}\label{eq:finalUnderEstOfD2Dprob}
\Pr[ d < \degr(v) \le 2d]  \ge \exp(-C_2\cdot d) \enspace.
\end{align}
Substituting \eqref{eq:finalUnderEstOfBuProb} and \eqref{eq:finalUnderEstOfD2Dprob} in \eqref{eq:lowerBoundWeightB}, we get that for $C=\max\{C_1, C_2\}$, and for every $d \ge d_0$
\begin{align}\label{eq:FinalLowerBoundinDenom}
   \Pr[\WB_P(v) \ge 1] 
   \ge \exp(-C\cdot d) \enspace.
\end{align}
Choosing $t=d^{95/100}$, invoking Lemma \ref{lem:MGFForcappedWB}, and putting everything together, we get that for every {$d\ge d_0 $}  we have
\begin{align*}
\nonumber 
\Pr \left[ \sum_{i \in H_1} \log \WB_P(v_i) \ge \frac{1}{2}\cdot\frac{\ell}{d^{1/50}} \middle| 
|H_1|\leq \frac{\ell}{d^{{1}/{10}}}
	\right] 
 &\le 
		{\left(e^{d \cdot C} \cdot 
    \left(1-\frac{\varepsilon}{4}\right)^{d^{\frac{93}{100}}}\right)^{\ell \cdot d^{-\frac{1}{10}}}}
  \cdot
		{\exp\left(-\frac{t \cdot \ell}{2\cdot d^{1/50}} \right)} \\
&\le 
\left(
{\exp\left({d^{9/10} \cdot {C}}-\frac{t}{d^{1/50}} \right)}
\right)^{\ell}
\le 
\frac{1}{2}\cdot 
n^{-d^{1/3}}
     \enspace,
\end{align*}
as desired. This concludes the proof of Theorem \ref{thm:BoundOnWBThm}.
\end{proof}



\section{Proof of Lemma \ref{lem:MGFForcappedWB}}\label{sec:lem:MGFForcappedWB}

Let us write $v$ instead of $v_1$. For $s\ge 0$, and a path $Q =(w_0, \ldots, w_s)$, let us define 
 \begin{align*}
 \cappedWA^{\rm even}(Q) = \prod_{i=0}^{\lfloor s/2\rfloor} \cappedWA(w_{2i}) \enspace,
 \end{align*}
 that is, the weight of $Q$ contributed only by vertices at even distance from $w_0$. Similarly, define $\cappedWA^{\rm odd}(Q)$ to be the weight of $Q$ contributed only by vertices at odd distance from $w_0$. Let us also define
 also define $\WB_{P}^{\rm even}(v) = \max_Q\{ \cappedWA^{\rm even}(Q)\}$, and $\WB_{P}^{\rm odd}(v) = \max_Q\{ \cappedWA^{\rm odd}(Q)\}$ where the maximisation is over all paths $Q$ of length  at most $r$, that emanate from $v$ and do not intersect with $P$, i.e.,  they to do not share vertices. Since $\WB_P(v) \le \WB_{P}^{\rm odd}(v) \cdot \WB_{P}^{\rm even}(v)$, the union bound yields
 \begin{equation*}
 \Pr \left[\WB_P(v) \ge 1\right] \le \Pr[\WB_{P}^{\rm odd}(v) \ge 1] + \Pr[\WB_{P}^{\rm even}(v) \ge 1] \enspace,
 \end{equation*}
 and thus, it suffices to show that
 \begin{equation*}
 \Pr[\WB_{P}^{\rm even}(v) \ge 1]\le \frac{1}{2} \cdot \left(1-\frac{\varepsilon}{4}\right)^{q} \enspace,
 \end{equation*}

 Writing $\mathrm{path}_v(r)$ to denote all paths of length $r$ in ${\GWT}$, that emanate from $v$, and do not intersect with $P$, and $\mathrm{path}_v(\le r) = \cup_{k\le r} \mathrm{path}_v(k)$, we see that for any $t>0$ we have that
\begin{align}
\label{eq:OnlySumTrickStepToExP}
\Exp\left[\left( \WB_P^{\rm even}(v) \right)^t\right] 
&= \Exp\left[\left( \max_{Q \in \mathrm{path}_v(\le r)}\{\cappedWA^{\rm even}(Q)\}\right)^t\right] 
\le \Exp\left[ \sum_{Q \in \mathrm{path}_v\left(\le r\right)} \left(\cappedWA^{\rm even}(Q)\right)^t\right] \enspace.
\end{align}
Note that we cannot pull the sum out of the expectation in the rhs of the inequality above, as the set $\mathrm{path}_v\left(\le r\right)$ is a random variable. Therefore, we think in the following way. For $k = 0 \ldots r$, there are at most $n^k$ potential paths of length $k$ emanating from $v$, and each potential path of has probability $(d/n)^k$ to be present in ${\GWT}$. Denoting with $ (w_0, \ldots , w_k)$ an arbitrary such potential path emanating from $v$, i.e., $w_0=v$, we see that
\begin{align}
\Exp\left[\left( \WB_P^{\rm even}(v) \right)^t\right] 
&\le \sum_{k=0}^r n^k \cdot \left(\frac{d}{n}\right)^{k} \cdot 
\Exp\left[
\left(
\prod_{i=0}^{\lfloor k/2 \rfloor} \cappedWA\left(w_{2i}\right) 
\right)^t\right]
\le
\sum_{k=0}^r d^k \cdot \left(\Exp\left[{\cappedWA^t\left(v\right)}\right] \right)^{\lfloor k/2 \rfloor+1}
\enspace,
\label{eq:ExpNeedsCalc}
 \end{align}
where the last inequality follows from the independence of the weights $\cappedWA$ corresponding to same parity vertices along $(w_0, \ldots, w_\ell)$.
To upper bound $\Exp\left[{\cappedWA^t\left(v\right)} \right]$ we consider two regimes in terms of the degree of vertex $v$:
 \begin{align}
\Exp\left[{\cappedWA^t\left(v\right)} \right] 
&= 
\Exp\left[{\cappedWA^t\left(v\right)} \cdot\Ind\{\degr(v) \le 3d\}  \right]
+
\Exp\left[{\cappedWA^t\left(v\right)} \cdot\Ind\{\degr(v) > 3d\}  \right]
\label{eq:ExpectationBigSmallDeg}
 \end{align}
Let us now focus on the first term of \eqref{eq:ExpectationBigSmallDeg}.  Writing $g(x)$ for the pdf of $\cappedWA(v)\cdot\Ind\{\degr(v) \le 3d\}$, and noticing that $\cappedWA(v)$ is at most $3d^2$,
we see that
\begin{align}
 \Exp\left[{\cappedWA^t\left(v\right)} \cdot\Ind\{\degr(v) \le 3d\}  \right]
&=
\int_{0}^{1-\varepsilon/4} x^t \cdot g(x) \; dx
+
\int_{1-\varepsilon/4}^{3d^2} x^t \cdot g(x) \; dx
\nonumber\\
&\le
 \left(1-\frac{\varepsilon}{4}\right)^t \cdot \Pr\left[\WA(v) \le 1-\frac{\varepsilon}{2}\right]
+
(3d^2)^{t} \cdot \int_{1-\varepsilon/2}^{3d^2} g(x) \; dx
\nonumber\\
&\le
\left(1-\frac{\varepsilon}{4}\right)^t
 +
(3d^2)^{t} \cdot \Pr\left[\WA(v) \ge {1-\frac{\varepsilon}{2}}\right]
\nonumber\\ \label{eq:FromOurBoundWA}
&\le
\left(1-\frac{\varepsilon}{4}\right)^t
+
(3d^2)^{t} \cdot {\exp\left(-\frac{\varepsilon^4 }{8\pi}\cdot d\right)}
\enspace, 
\end{align}
where \eqref{eq:FromOurBoundWA} follows from Theorem \ref{theorem:TailBound4WA}.
We now focus the large-degree case. Accounting only for the randomness on the degree of $v$, and overestimating each term of $\cappedWA(v)$ to be equal to $d$, we get that
\begin{align}
\nonumber
\Exp\left[{\cappedWA^t\left(v\right)} \cdot\Ind\{\degr(v) > 3d\}  \right]
 &\le
\sum_{k=3d}^{n} (d\cdot k)^t \cdot\Pr[\degr(v) = k] 
 \le
 d^t \cdot\sum_{k=3d}^{n} k^t \cdot \binom{n}{k} \cdot \left(\frac{d}{n}\right)^k
 \\
  &\le
  d^t \cdot
 \sum_{k=3d}^{n} k^t \cdot \left(\frac{ne}{k}\right)^k  \cdot \left(\frac{d}{n}\right)^k
   \le d^t \cdot
 \sum_{k=3d}^{n} k^t \cdot \left(\frac{de}{k}\right)^k  \enspace.
 \label{eq:finalToBound}
\end{align}
Substituting \eqref{eq:FromOurBoundWA} and \eqref{eq:finalToBound} to  \eqref{eq:ExpectationBigSmallDeg}, gives
\begin{equation*}
 \Exp\left[{\cappedWA^t\left(v\right)} \right] \le
 \left(1-\frac{\varepsilon}{4}\right)^t
+
(3d^2)^{t} \cdot \exp\left(-\frac{\varepsilon^4 }{8\pi}\cdot d\right)
+
 d^t \cdot \sum_{k=3d}^{n} k^t \cdot \left(\frac{de}{k}\right)^k \enspace.
\end{equation*}
Since $ed/k < 1$, for every $k >3d$, choosing $t = d^{{95}/{100}}$, we see that there exist $d_0(\varepsilon)\ge 1$, such that for every $d \ge d_0$, we have that
$\Exp\left[{\cappedWA^t\left(v\right)} \right] \le (1-\frac{\varepsilon}{4})^{t^\prime}$, where $t^\prime = d^{{94}/{100}}$.  We now bound \eqref{eq:ExpNeedsCalc} as
\begin{align*}
    \Exp\left[\left( \WB^{\rm even}_P(v) \right)^{{t}}\right] \le \frac{(1-\frac{\varepsilon}{4})^{{t^\prime}}}{1-d^2\cdot(1-\frac{\varepsilon}{4})^{{t^\prime}}}
    {\le
    {\left(1-\frac{\varepsilon}{4}\right)^{{q}}}
    }
    \enspace,
\end{align*}
where we can take $q = d^{93/100}$.

\section{Proof of Theorem~\ref{theorem:TailBound4WA}} \label{sec:theorem:TailBound4WA}

Let $ d_0, \theta, \eta > 0$, to be determined later. First, note that due to the total probability law we have
\begin{align}\label{eq:breakPrCondDeg}
\Pr\left[\WA(v) \geq (1-\theta) \right] \le
\Pr\left[\WA(v) \geq (1-\theta) \mid \degr(v) \le (1+\eta) d \right]
+ \Pr\left[\degr(v) > (1+\eta) d\right]  \enspace.
\end{align}
Let us now focus on the first term of \eqref{eq:breakPrCondDeg}. Recall that $ \WA(v)= \sum_{{z \sim v }} \Inf_{\{v,z\}}$, with
\begin{equation*}
\Inf_{e} =  \frac{\left|1-\exp\left( \beta \bJ_e  \right) \right|}{1+\exp\left( \beta \bJ_{e}\right) }=  \left|\tanh\left(\frac{\beta}{2} \bJ_e \right) \right| \enspace. 
\end{equation*}
where $\bJ_e$ follows the standard Gaussian distribution. 
Since $|\tanh(x)| \le |x|$, for every real $x$, we get further that
\begin{equation*}
\Inf_{e} \le \frac{\beta}{2} |\bJ_e| \enspace,
\end{equation*}
 Hence, bounding from above the upper-tail of $\sum_e \frac{\beta}{2} |\bJ_e|$, provides an upper bound to the corresponding tail for $\sum_e \Inf_{e}$, i.e., for every $x\in \mathbb{R}$ we have
\begin{equation}\label{eq:BboundedBySumOfHalfNorm}
\textstyle
    \Pr\left[\WA(v) \ge x\right] \le 
    \Pr \left[\sum_{e} \frac{\beta}{2} |\bJ_e| \ge x\right] \enspace.
\end{equation}

Applying the bound \eqref{eq:MMRwrite} of Theorem~\ref{thm:UpperTailBoundSumOfHalfNormRelative} in the rhs of \eqref{eq:BboundedBySumOfHalfNorm} and since $0<\beta \le (1-\varepsilon)\beta_c(d)$, we have that for every {$\eta \le (\varepsilon - \theta)/(1-\varepsilon)$}
\begin{align}\label{eq:PrBSmallIfDegSmall}
\Pr\left[\WA(v) \geq (1-\theta) \mid \degr(v) \le (1+\eta) d \right] 
&\le
\exp\left(-\left[\frac{ \left( (1-\theta)\sqrt{\pi} -\frac{\beta}{2} (1+\eta)d \sqrt{2}  \right)^{2}}{ 2\pi(1+\eta)^2d^2\frac{\beta^2}{4} }\right]  (1+\eta)d \right) \nonumber \\
&=
\exp\left(-\left[\frac{ \left( (1-\theta)- (1+\eta) \frac{\beta d}{\sqrt{2\pi}}  \right)^{2}}{ (1+\eta)^2d^2{\beta^2}}\right] 2(1+\eta)d \right) \nonumber\\
&\le
\exp\left(-\left[\frac{ \left( (1-\theta)- (1+\eta) (1-\varepsilon)  \right)^{2}}{ 2\pi(1+\eta)^2(1-\varepsilon)^2}\right] 2(1+\eta)d \right) \nonumber\\ 
&\le
\exp\left(-\left[\frac{ \left( (1-\theta)- (1+\eta) (1-\varepsilon)  \right)^{2}}{ \pi(1+\eta)(1-\varepsilon)^2}\right] d \right) 
\enspace. 
\end{align}
Choosing $\theta = \eta = \varepsilon/2$ in \eqref{eq:PrBSmallIfDegSmall} gives that
\begin{align}\label{eq:PrBSmallIfDegSmallSpecificEtaTheta}
\Pr\left[\WA(v) \geq (1-\theta) \mid \degr(v) \le (1+\eta) d \right] 
\le
\exp\left(-\left[\frac{ \varepsilon^4}{ 2\pi(2+\varepsilon)(1-\varepsilon)^2}\right] d \right)
\le
\exp\left(-\frac{ \varepsilon^4}{ 6\pi}\cdot d \right)
\enspace. 
\end{align}

Let us now turn to the second term of \eqref{eq:breakPrCondDeg}. Since $\degr(v)$ is a sum of independent Bernoulli random variables, applying the Chernoff tail bound gives that for every $\eta \ge 0$
\begin{align}\label{eq:PrDegBeLarge}
\Pr\left[\degr(v) > (1+\eta) d\right] \le \exp\left(-\frac{\eta^2 d}{2 + \eta}\right)
\enspace.
\end{align}
Taking $\eta = \varepsilon$, and
substituting \eqref{eq:PrBSmallIfDegSmallSpecificEtaTheta} and \eqref{eq:PrDegBeLarge}, in \eqref{eq:breakPrCondDeg} gives that
\begin{align*}
\Pr\left[\WA(v) \geq (1-\varepsilon/2) \right]
\le \exp\left(-\frac{\varepsilon^4 }{6\pi}\cdot d + \log 2\right)
\enspace,
\end{align*}
which for $d \ge d_0 := 53\cdot \varepsilon^{-4}$, yields
\begin{align*}
\Pr\left[\WA(v) \geq (1-\varepsilon/2) \right]
\le \exp\left(-\frac{\varepsilon^4 }{8\pi}\cdot d\right)
\enspace.
\end{align*}

\section{Proof of Theorem \ref{thrm:RelaxBounds}}\label{sec:thrm:RelaxBounds}

Since, $1\leq N\leq n$,  it is standard to show that Theorem \ref{thrm:RapidMixingBlockDyn} implies that
\begin{align}\label{eq:FinalBound4BlockDynG}
\uptau_{\rm block}&=O(n\log n)\enspace.
\end{align}
For the rest of the proof we focus on bounding $\uptau_{B}$ for every $B\in \cB$. 

For each block $B\in \cB$ that is unicyclic  with cycle $C=(w_1, \ldots, w_{\ell})$, let   $T_i(B)$ 
be  the connected component in  $B\cup \partial_{\rm out} B$  that includes $w_i$, once we delete {\em all} the edges in  $C$.

Let $\cT=\cT(\cB)$ be the collection of  the following trees:
 \begin{enumerate}[label=(\alph*)]
 \item for every  multi-vertex block $B\in \cB$ that is tree, $\cT$ includes the tree $B\cup \partial_{\rm out} B$,
 \item for every unicyclic block $B\in\cB$,  $\cT$ includes every tree $T_i(B)$.
 \end{enumerate}
 
For each  tree of the type  (a) above, the root is assumed to be a (any) heavy  vertex $w$ such that $\WA(w)>1-\varepsilon/2$.  
For each  tree of type (b), i.e., for each $T_i(B)$,  the root  is assumed to be  the vertex $w_i$, i.e., the vertex that $T_i(B)$
intersects with  the cycle of the block. 

%
%

In order to avoid repetitions,  we use the following convention for the rest of the proof: the Glauber dynamics, or the block dynamics on a subgraph $H$ of $G$, 
is always assumed to be  with respect to the marginal of the Gibbs distribution $\mu_{G,\vecJ,\beta}$ at the subgraph $H$. When necessary,
we impose boundary conditions at $\partial_{\rm out} H$. Recall that $\partial_{\rm out} H$ is the set of 
vertices outside $H$ that have neighbouring vertices inside $H$. 

The following result provides a bound on the relaxation time of tree $T\in \cT$ by utilising a recursive argument. This  argument relies on ideas used to derive  a similar bound in \cite{EfthymiouHSV18, mossel2010gibbs}. 

\begin{theorem}\label{thrm:TreeRelaxationBound}%
For a  tree $T\in \cT$, having vertex $v$ as its root, and  a fixed configuration $\sigma$ at $\partial_{\rm out} T$, 
consider the  Glauber dynamics on $T$. Let 
$\uptau_{\rm rel}=\uptau_{\rm rel}(T,\sigma)$ be the relaxation time of this dynamics. 
We have that 
\begin{align*}
\uptau_{\rm rel}&\leq \exp\left(  m(T,v)  \right)\enspace,
\end{align*}
where $m(T,v)= \max_{P}\{ {\compWeight}(P)\}$ and  the maximum is  over all root-to-leaf paths $P$ in $T$.
\end{theorem}

\noindent
The proof of Theorem \ref{thrm:TreeRelaxationBound} appears in Section \ref{sec:thrm:TreeRelaxationBound}.
For every block $B\in \cB$ which is tree, Theorem~\ref{thrm:TreeRelaxationBound} immediately  implies that 
\begin{align}\label{eq:RelaxTreeBFirstComparison}
\uptau_{B}\leq \exp\left(  m(B,v)  \right)\enspace,
\end{align}
for any vertex $v\in B$ being the root of $B$. 

Recall that  $m(B,v)= \max_{P}\{ {\compWeight}(P)\}$ and  the maximum is  over all the paths $P$ in $T$ 
from the root $v$ to the leaves of the tree. By construction of the block partition $\cB$, we have that all paths considered for $m(B,v)$ are of length at most 2{$\frac{\log n}{\sqrt{d}}$}.  Moreover, our assumption that
$(G,\vecJ,\beta)\in \IntrstGraphFam(d,\varepsilon)$ implies that $m(B,v)\leq {\frac{\log n}{\log^2 d}}$. 
Plugging this bound for  $m(B,v)$ into \eqref{eq:RelaxTreeBFirstComparison} we get that 
\begin{align}\label{eq:RelaxTreeBFinalComparison}
\uptau_{B} &\leq n^{\frac{1}{\log^2d}}\enspace,
\end{align}
establishing the desired bound for $\uptau_{B}$ for the case where $B\in \cB$ is a tree. We now focus on the case where $B$ is a unicyclic factor graph.  Specifically, we prove the following 
lemma.
\begin{lemma}\label{lemma:RelaxUnicyclicBFinalComparison}
For the case where $B\in \cB$ is unicyclic we have that $\uptau_{B} \leq n^{\frac{3}{\log^2 d}}$.
\end{lemma}
For the proof of Lemma \ref{lemma:RelaxUnicyclicBFinalComparison} see Section \ref{sec:lemma:RelaxUnicyclicBFinalComparison}.
Given the bound for $\uptau_{\rm block}$ in \eqref{eq:FinalBound4BlockDynG}, and the bounds of $\uptau_{B}$ in \eqref{eq:RelaxTreeBFinalComparison} and  Lemma \ref{lemma:RelaxUnicyclicBFinalComparison},
we conclude  the proof of the theorem. 
\hfill $\Box$

\subsection{Proof of Theorem \ref{thrm:TreeRelaxationBound}}\label{sec:thrm:TreeRelaxationBound}

For a vertex $u\in T$, let $T_u$ denote the subtree of $T$ containing $u$
and all its descendants. Unless otherwise specified, we assume that the root of $T_u$ is $u$. Note also that $\partial_{\rm out} T_u$ is the subset of $\partial_{\rm out} T$ comprised by all $w\in \partial_{\rm out} T$
having a neighbour in $T_u$.

\begin{proposition}\label{prop:MixingStarUnified}
Let $T\in \cT$ and let $u \in T$. Consider $T_u$  and let $w_1, \ldots, w_{R}$  
be the children of the root $u$. Consider the block dynamics  $(X_t)_{t\geq 0}$ in $T_u$ with set of blocks $\cM=\{\{u \}, T_{w_1}, \ldots, T_{w_R}\}$.

Under any boundary condition at $\partial_{\rm out} T_u$,  the  block dynamics $(X_t)_{t\geq 0}$  exhibits relaxation time
\begin{align*}
\uptau_{\rm block}(T_u) & \leq  { \exp\left(10 \log(R)+ 2\beta\sum\nolimits^R_{i=1} \left|J_{\{u, w_i\}}\right|\right)}\enspace.
\end{align*}
\end{proposition}
The proof of Proposition \ref{prop:MixingStarUnified} appears in Section \ref{sec:prop:MixingStarUnified}.
In light of the above proposition, the theorem follows by induction on  the height of the tree $T$. 

The base case corresponds to a single vertex tree, where trivially,  we have that  $\uptau_{\rm rel}(T)=1$.
Assume,  now, that the root $u$ of $T$ has children $w_1, \ldots, w_{\ell}$, for some $\ell\ge 1$. 
Then, per the induction hypothesis we have that
\begin{align}\label{eq:IndHypoRelax}
\uptau_{\rm rel}(T_{w_i}) & \leq \exp\left( m(T_{w_i}, w_i)\right) \qquad \textrm{ for $i=1, \ldots, \ell$}\enspace,
\end{align}
where recall that  $\uptau_{\rm rel}(T_{w_i})$ corresponds to the relaxation time for the  Glauber dynamics on $T_{w_i}$. 

In order to derive $\uptau_{\rm rel}(T)$, i.e., the relaxation time for the   Glauber dynamics on $T$,
consider first the block dynamics on $T$ where the set of blocks is $\{\{u\}, T_{w_1}, \ldots, T_{w_{\ell}}\}$.
The relaxation time for this process is given by Proposition \ref{prop:MixingStarUnified}.  That is, 
\begin{align}\label{eq:IndStepRelax}
\uptau_{\rm block}(T) &={ \exp\left(10 \log(\ell)+ 2\beta\sum\nolimits^{\ell}_{i=1} \left|J_{\{u, w_i\}}\right|\right)}\enspace. 
\end{align}
Using the bounds from \eqref{eq:IndStepRelax}, \eqref{eq:IndHypoRelax}, and Proposition \ref{prop:Comparison} we deduce that
$\uptau_{\rm rel}(T) \leq \exp\left( m(T, u)\right)$.

All the above conclude the proof of the theorem. 
\hfill $\Box$

\subsection{Proof of Proposition \ref{prop:MixingStarUnified}}\label{sec:prop:MixingStarUnified}

We prove the proposition using coupling. Let $(X_t)_{t\ge0}$, $(Y_t)_{t\ge0}$ be two copies of the Markov chain, 
with configuration  $\sigma$ at $ \partial_{\rm out}T$. 

We couple the two chains  such that at each transition we update the same block in the 
two copies.  Also, note that the coupling has the following basic property:  if for some $t_0$ we have $X_{t_0}=Y_{t_0}$, then for every $t>t_0$ we  also have $X_{t}=Y_{t}$.

We divide the evolution of the chains into ``epochs". Each epoch consists of $N=5 R \log (R)$ transitions.  
We  argue that at the end of each epoch the pair of chains in the coupling  are at the same configuration with probability
at least $\frac{1}{2}\exp\left(-2\beta\sum^R_{j=1}|J_{\{u,w_j\}}|\right)$, i.e., we have that 
\begin{align}\label{eq:Target4prop:MixingStarUnified}
\min_{X_0, Y_0} \Pr[X_N=Y_N\ |\ X_0, Y_0] & \geq 
\frac{1}{2}\exp\left(-2\beta\sum\nolimits^R_{j=1}\left|J_{\{u,w_j\}}\right|\right)\enspace,
\end{align}
where $J_{\{u,w_j\}}$ is the coupling parameter at the edge $\{u,w_j\}$.   Then, it is elementary  to show that 
after $100 \exp\left(2\beta\sum^R_{j=1}|J_{\{u,w_j\}}|\right)$  epochs, the probability that the two chains 
agree  is larger than~$0.8$. 

Clearly, the above implies that  $\uptau_{\rm block}(T_u) \leq 500 R\log(R)\exp\left(2\beta\sum^R_{j=1}|J_{\{u,w_j\}}|\right) $.
Hence, the proposition follows by showing that \eqref{eq:Target4prop:MixingStarUnified} holds.

Suppose that at time $t>0$, we update block~$T_{w_j}$. Recall that we update the same block in 
the two copies. We further design our coupling such that the following are satisfied at each update: 

(a) If we have agreement at the root, i.e., we have that $X_t(u)=Y_t(u)$, then deterministically, i.e., with probability $1$, we have 
$X_t(T_{w_j})=Y_t(T_{w_j})$.  We can achieve this because the marginal distributions at the block $T_{w_j}$ is the 
same for both copies, and thus, we can use identical coupling. 

(b) If we have a disagreement at the root of $T_u$, i.e., we have that $X_t(u) \neq Y_t(u)$, then we have 
$X_t(T_{w_j})=Y_t(T_{w_j})$ with probability $1-\Inf_{\{u,w_j\}}$.  To see this, note that we
first couple the configuration at  $w_j$. Using maximal coupling, and due to Lemma~\ref{lemma:Coupling4TreeMeasures}, we have $X_t(w_j)=Y_t(w_j)$ with probability at least $1-\Inf_{\{u,w_j\}}$. Subsequently, i.e., once we obtain the configuration at $w_j$, we couple maximally the configuration for
the remaining vertices in $T_{w_j}$. In that respect, it is easy to see that if $X_t(w_j)=Y_t(w_j)$,
then we can couple identically the configuration at the remaining vertices of the block. On the the other 
hand, having $X_t(w_j) \neq Y_t(w_j)$ precludes having $X_t(T_{w_j})=Y_t(T_{w_j})$.

From the above we conclude that there is coupling such that, if at time $t>0$ we update
any of the blocks $\{ T_{w_1}, \ldots, T_{w_R} \}$, then we have that
\begin{align}\label{eq:DisagreeProbComp}
\Pr[X_{t}(T_{w_j})= Y_{t}(T_{w_j})\ |\ \cF] & \geq 1-\Inf_{\{u,w_j\}} & \forall j=1,2\ldots, R\enspace, 
\end{align}
where $\cF$ is the $\sigma$-algebra generated by the configurations of $X_t$, $Y_t$ at the  blocks in $\cM \setminus \{T_{w_j}\}$.

W.l.o.g. let us focus on the first epoch. Let $\mathcal{U}_{r}$ be the event the root is updated at least once
between  the transitions $3R\log(R)$ and $5R\log(R)$. Also, let $\mathcal{U}_{\rm all}$ be the event that
prior to $3R\log(R)$ all the blocks in $\{ T_{w_1}, \ldots, T_{w_R} \}$ are updated at least once.  
A standard coupon collector type  argument implies that for any $X_0, Y_0$ we have that  
\begin{align}\label{eq:CCollectComparison}
\Pr[ \mathcal{U}_{r}, \ \mathcal{U}_{\rm all} \ |\ X_0, Y_0] \geq 1-10^{-2} \enspace. 
\end{align}

Furthermore, for each $T_{w_j}$, let $\cA_j$ be the event that the last time block $T_{w_j}$ was updated, 
prior to the update of the root, the two chains agree on the configuration of this block.  Similarly, let $\cA_{r}$
be the probability that, after its update, the root has the same configuration in both  copies. 

By the design of our coupling, e.g. see \eqref{eq:DisagreeProbComp}, we have that
for any $X_0, Y_0$ it holds that 
\begin{align}\label{eq:AllAgreeComparison}
\Pr\left[ \bigcap\nolimits^R_{j=1} \cA_j  \mid   \mathcal{U}_{r}, \ \mathcal{U}_{\rm all}, X_0,\ Y_0 \right ] & \geq { \prod\nolimits^R_{j=1}}\left(1-\Inf_{\{u,w_j\}}\right)\enspace. 
\end{align}
Moreover, we have that
\begin{align}\label{eq:RootAgreesComparison}
\Pr\left[\cA_r \mid\cap^R_{i=1}\cA_i,\ \mathcal{U}_{r}, \  \mathcal{U}_{\rm all}, X_0,\ Y_0 \right ] &=1\enspace.
\end{align}
The above holds by noticing that under the event $\cap^R_{i=1}\cA_i$, when the root is updated the
distributions of its configuration in the two copies are identical, and thus, we can use identical coupling.

From \eqref{eq:CCollectComparison}, \eqref{eq:AllAgreeComparison} and \eqref{eq:RootAgreesComparison}, we conclude that
for any $X_0, Y_0$, we have that 
\begin{align}
\Pr[X_N=Y_N\ |\ X_0, \ Y_0] &\geq  \Pr\left[\cA_r, \  {\textstyle \bigcap}^R_{i=1}\cA_i,\ \mathcal{U}_{r}, \  \mathcal{U}_{\rm all} \ |\ X_0, \ Y_0\right ] \nonumber\\
&= \Pr\left[\cA_r \ |\ \cap^R_{i=1}\cA_i,\ \mathcal{U}_{r}, \  \mathcal{U}_{\rm all},\ X_0, \ Y_0\right ] \times \nonumber\\
&\times \Pr\left[ \cap^R_{i=1}\cA_i\ |\ \mathcal{U}_{r}, \  \mathcal{U}_{\rm all},\ X_0, \ Y_0 \right ] 
\cdot \Pr\left[ \mathcal{U}_{r}, \  \mathcal{U}_{\rm all}\ |\ X_0, \ Y_0 \right ] \nonumber\\
&\geq \frac{1}{2} {\textstyle \prod^R_{j=1}}(1-\Inf_{\{u,w_j\}}) \enspace. \label{eq:BasicLowerBound4TreeComp}
\end{align}
The proposition follows by showing that for every $w_j$ we have that
\begin{align}\label{eq:LB31MInfluenceComparison}
1-\Inf_{\{u,w_j\}}\geq \exp\left(-2\beta |J_{\{u,w_j\}}| \right)\enspace. 
\end{align}
We distinguish two cases.  In the first one, we assume that $|J_{\{u,w_j\}}|$ is such that $\exp(\beta |J_{\{u,w_j\}}|)\geq 3$. The second case corresponds to
having $|J_{\{u,w_j\}}|$  such that $\exp(\beta |J_{\{u,w_j\}}|)<3$. 

\subsubsection*{Case 1:}  Recall that we assume that $\exp(\beta |J_{\{u,w_j\}}|)\geq 3$.  It is elementary to show that, for any $J_{\{u,w_j\}}$, i.e., not necessarily large, we have that
\begin{align}\label{eq:MoveAbs2ExpInfl}
\Inf_{\{u,w_j\}}&={\textstyle \frac{|1-\exp(\beta J_{\{u,w_j\}})|}{1+\exp(\beta J_{\{u,w_j\}})}} = {\textstyle \frac{1-\exp(-\beta |J_{\{u,w_j\}}|)}{1+\exp(-\beta |J_{\{u,w_j\}}|)}} \enspace. 
\end{align}
From the above, we have that 
\begin{align}
1-\Inf_{\{u,w_j\}}=2\textstyle \frac{\exp(-\beta|J_{\{u,w_j\}}|)}{1+\exp(-\beta|J_{\{u,w_j\}}|)} \geq \frac{3}{2}\exp\left(-\beta|J_{\{u,w_j\}}|\right) \geq \exp\left(-2\beta|J_{\{u,w_j\}}|\right)\enspace.
\end{align}
The above proves \eqref{eq:LB31MInfluenceComparison} for the case where $\exp(\beta |J_{\{u,w_j\}}|)\geq 3$.

\subsubsection*{Case 2:}  Recall that now we assume that $\exp(\beta |J_{\{u,w_j\}}|)< 3$.  
Then, using  \eqref{eq:MoveAbs2ExpInfl} we get that
\begin{align}\label{eq:InflUBCaseBComparison}
\Inf_{\{u,w_j\}} <1/2\enspace. 
\end{align}
Furthermore, using the standard inequality: $1-x\geq \exp(-\frac{x}{1-x})$ for $0<x\leq 1/2$, we get that
\begin{align}\nonumber
1-\Inf_{\{u,w_j\}} \geq {\textstyle \exp\left(-\frac{\Inf_{\{u,w_j\}}}{1-\Inf_{\{u,w_j\}}}\right)}\geq \exp\left(-2\Inf_{\{u,w_j\}}  \right)\enspace,
\end{align}
where, for the last inequality we use that $\Inf_{\{u,w_j\}} <1/2$, i.e., we use  \eqref{eq:InflUBCaseBComparison}.  The above proves \eqref{eq:LB31MInfluenceComparison} for the case 
where $\exp(\beta |J_{\{u,w_j\}}|)< 3$.
All the above conclude the proof of the proposition. 
\hfill $\Box$

\section{Remaining Proofs}

\subsection{Proof of Lemma \ref{lemma:RelaxUnicyclicBFinalComparison}}\label{sec:lemma:RelaxUnicyclicBFinalComparison}

Consider the unicyclic block $B$.
As per standard notation, we let $C=(w_1, \ldots, w_{\ell})$ be the cycle inside $B$, for some {$\ell\leq 4\frac{\log n}{\log^4 d}$}.

Consider  the,  block dynamics, on $B$ with  a fixed  boundary $\sigma$  at $\partial_{\rm out} B$. 
There are two blocks in this dynamics. The first block, $B_1$, corresponds to the  tree $T\in \cT$ which intersects with the cycle 
$C$ at vertex $w_1$.
The second block $B_2$ corresponds to the vertices in $B\setminus B_1$.  Note that both $B_1$, $B_2$ are trees. 
Let $\uptau_{\rm block}$ be the relaxation of this block dynamics.  

Similarly to Proposition \ref{prop:MixingStarUnified}, we divide the evolution of the dynamics into epochs
to  prove that 
\begin{align}\label{eq:Trelax4UnicyclicB}
\uptau_{\rm block} &\leq 500 n^{\frac{1}{\log^3 d}} \enspace.
\end{align}

Let $(X_t)$ and $(Y_t)$ be two copies of the block dynamics
with the same boundary condition at $\sigma$ at  $\partial_{\rm out} T$. 
We couple the two chains maximally such that at each transition we update the same block in both of them.  Also, note that the coupling is such that, if for some $t_0$ we have $X_{t_0}=Y_{t_0}$,
then for any $t>t_0$, we  also have $X_{t}=Y_{t}$.

We consider  epochs each of length  $N=10$. 
We  argue that at the end of each epoch the chains in the pair are at the same configuration with probability
at least $(1/2) n^{-1/\log^3 d}$, i.e., we have that 
\begin{align}\label{eq:Target4lemma:RelaxUnicyclicBFinalComparison}
\min_{X_0, Y_0} \Pr[X_N=Y_N\ |\ X_0, Y_0] & \geq \textstyle \frac{1}{2} \cdot n^{-\frac{1}{\log^3 d}}\enspace. 
\end{align}

 Then, it is elementary  to show that 
after $50 n^{1/\log^3 d}$  many epochs, the probability the two chains 
agree  is larger than  $0.8$. 
Clearly, the above yields $\uptau_{\rm block}(T_u) \leq 500 n^{1/\log^3 d}$, establishing \eqref{eq:Trelax4UnicyclicB}.

Suppose that at time $t>0$, we update the block $B_1$. Recall that we update the same block in 
the two copies. Then we use the following result. 

\begin{claim}\label{claim:Influence4Two}
There is a coupling such that 
\begin{align}
\Pr[X_{t}(B_1)= Y_{t}(B_1)\ |\ \cF_1] & \geq  n^{-\frac{1}{\log^3d}}   \enspace, 
\end{align}
where $\cF_1$ is the $\sigma$-algebra generated by the configurations of $X_t$, $Y_t$ at $B_2$.
\end{claim}

Similarly, suppose  that at time $t>0$, we update $B_2$ in both copies.  
We use the following result.
 
\begin{claim}\label{claim:TwoDisagreements}
There is a coupling such that 
\begin{align}
\Pr[X_{t}(B_2)= Y_{t}(B_2)\ |\ \cF_2] & \geq n^{-\frac{1}{\log^3d}} \enspace, 
\end{align}
where $\cF_2$ is the $\sigma$-algebra generated by the configurations of $X_t$, $Y_t$ at $B_1$.
\end{claim}

Then, using the above claims and arguing as in Proposition \ref{sec:prop:MixingStarUnified}, we get \eqref{eq:Target4lemma:RelaxUnicyclicBFinalComparison} directly.

Let $\uptau_{1}$ be the relaxation time of the  Glauber dynamics at $B_1$ with, arbitrary boundary fixed condition 
$\sigma$ at $\partial_{\rm out}B_1$. Similarly, we define  $\uptau_2$ to be the relaxation time for block~$B_2$. 
Then,  working as in Theorem \ref{thrm:TreeRelaxationBound} and recalling that $(G,\vecJ,\beta)\in \IntrstGraphFam(d,\varepsilon)$ we obtain 
\begin{align}\label{eq:Trelax4UnicyclicBSubtrees}
\uptau_{1}&\leq n^{\frac{2}{\log^2d}} & \textrm{and} && \uptau_{2}&\leq n^{\frac{2}{\log^2 d}}\enspace. 
\end{align}
From  \eqref{eq:Trelax4UnicyclicB}, \eqref{eq:Trelax4UnicyclicBSubtrees} and Proposition \ref{prop:Comparison} we get that
$\uptau_{B} \leq n^{\frac{3}{\log^2 d}}$.

All the above conclude the proof of Lemma~\ref{lemma:RelaxUnicyclicBFinalComparison}.  
\hfill{} $\Box$

\begin{proof}[Proof of Claim \ref{claim:Influence4Two}]
Let  $J_{a}$ and $J_b$ be  the coupling parameters for the edges connecting  $w_1$ to $w_2$, and $w_1$ to $w_{\ell}$, respectively.  
The probability $\Pr[X_{t}(B_1)= Y_{t}(B_1)\ |\ \cF_1]$ is minimised by choosing $\eta^+$ and $\eta^-$ two  configurations at $w_2$ and $w_{\ell}$,
such that $\eta^+$ maximises the  probability of having $+1$ in the Gibbs marginal at $w_{1}$, while $\eta^-$ maximises $-1$ for the same
marginal. 

Specifically, note that $\eta^+$ is such that $\eta^+(w_{\ell})={\rm sign}(J_b)$ and $\eta^+(w_{2})={\rm sign}(J_a)$. Similarly,  
$\eta^-$ is such that $\eta^-(w_{\ell})=-{\rm sign}(J_b)$ and $\eta^-(w_{2})=-{\rm sign}(J_a)$. 
Then, letting
\begin{align}\nonumber 
\Inf^* &=\textstyle \frac{1-\exp(-\beta(|J_a|+|J_b|))}{1+\exp(-\beta(|J_a|+|J_b|))}\enspace,
\end{align}
it is standard to show that the maximal coupling of the configuration at $w_1$ satisfies
\begin{align}
\Pr[X_{t}(B_1)= Y_{t}(B_1)\ |\ \cF_1] &\geq 1-\Inf^*  \nonumber \\
&= 2\frac{\exp(-\beta(|J_a|+|J_b|))}{1+\exp(-\beta(|J_a|+|J_b|))} \geq \exp(-\beta(|J_a|+|J_b|))\enspace.  \label{eq:JustBeforeclaim:Influence4Two} 
\end{align}
Recall that we assume  $(G,\vecJ,\beta)\in \IntrstGraphFam(d,\varepsilon)$, and in particular, that $|J_a|, |J_b| \leq 10\sqrt{\log n}$. 
Clearly, this implies that  $\exp(-\beta(|J_a|+|J_b|))\geq n^{-1/\log^3 d }$, for large $d$ and $n$.  
The claim follows by plugging this bound into \eqref{eq:JustBeforeclaim:Influence4Two}.
\end{proof}

\begin{proof}[Proof of Claim \ref{claim:TwoDisagreements}]
%

The  probability $\Pr[X_{t}(B_2)= Y_{t}(B_2)\ |\ \cF_2]$ is minimised by 
having a disagreement at $w_1$, i.e., otherwise this probability is $1$. 
Note that $w_1$ has two neighbours in $B_2$. 

The coupling we use is as follows: First we couple maximally vertex $w_2$ and then, given the outcome for $w_2$, 
we couple maximally vertex $w_{\ell}$.  We have that 
\begin{eqnarray}
\lefteqn{
\Pr[X_{t}(B_2)= Y_{t}(B_2) \mid \cF_2] 
} \vspace{2cm} \nonumber  \\ 
&\geq &
\Pr[X_{t}(w_{\ell })= Y_{t}(w_{\ell}) \mid  X_{t}(w_2)= Y_{t}(w_2),\ \cF_2]  
\cdot
\Pr[X_{t}(w_2)= Y_{t}(w_2) \mid  \cF_2] \enspace.   \label{eq:claim:Base4V}
\end{eqnarray}

We start by focusing on  the coupling for $w_{2}$. We have that
\begin{eqnarray}
\lefteqn{
\Pr[X_{t}(w_2) =  Y_{t}(w_2)\mid X_t(w_1)\neq Y_t(w_1), \cF_2] 
}\vspace{4cm} \label{eq:DisagreementAtW2OnlyW1} \\
&\leq &
\max_{\eta^+,\eta^-} 
\Pr[X_{t}(w_2) =  Y_{t}(w_2)\ |\ X_t(\{w_1, w_3\})=\eta^+, \ Y_t(\{w_1,w_3\})=\eta^-, \cF_2]\enspace. \label{eq:DisagreementAtW2OnlyW1W3}
\end{eqnarray}

To see the reason why the above is true, first note that the disagreement at vertex $w_1$ affects the marginal at $w_2$ from two directions. 
The first one is over the edge $\{w_1,w_2\}$, the second one is over the path $(w_1, w_{\ell}, \ldots, w_2)$. 

The  probability of having a disagreement at $w_2$ in \eqref {eq:DisagreementAtW2OnlyW1} is under maximal coupling, conditional 
on the configurations at $w_1$.
The  probability of having a disagreement at $w_2$ in \eqref{eq:DisagreementAtW2OnlyW1W3} is under maximal coupling, conditional on 
the configurations at $w_1$ and $w_3$. The inequality follows by a standard convexity argument. 

Then, arguing as in Claim \ref{claim:Influence4Two}, we get that
\begin{align}\nonumber 
\Pr[X_{t}(w_2) =  Y_{t}(w_2)\ |\ X_t(w_1)\neq Y_t(w_1),\  \cF_2] &\geq \exp(-\beta(|J_a|+|J_c|)) \enspace,  
\end{align}
where recall that $J_{a}$ and $J_c$ are the coupling parameters for the edge that connects  $w_2$, with $w_1$ and $w_{3}$, respectively.
Recall that we assume  that $(G,\vecJ,\beta)\in \IntrstGraphFam(d,\varepsilon)$. This implies that $ |J_a|, |J_c| \leq 10\sqrt{\log n}$. 
Plugging this bound into the inequality above, for large $d$ and $n$,  we get that
\begin{align}\label{eq:Bound4W2Agreement}
\Pr[X_{t}(w_2) =  Y_{t}(w_2)\ |\ X_t(w_1)\neq Y_t(w_1),\  \cF_2] &\geq n^{-1/(\log d)^4} \enspace.
\end{align}

Then, conditional on $X_{t}(w_2) =  Y_{t}(w_2)$,  we couple maximally $X_{t}(w_{\ell}),  Y_{t}(w_{\ell})$.  Note that 
the disagreement at $w_{1}$ only affects the distribution of the configuration at $w_{\ell}$ through the edge $\{w_1, w_{\ell}\}$. 
Hence, using standard arguments (e.g. the same as those in the proof of Proposition \ref{prop:MixingStarUnified}), we get that 
\begin{align}\nonumber
\Pr[X_{t}(w_{\ell })= Y_{t}(w_{\ell})\ |\ X_{t}(w_2)= Y_{t}(w_2),\ \cF_2]&\geq 2\frac{\exp(-\beta |J_b|)}{1+\exp(-\beta |J_b|)}\enspace. 
\end{align}
Recall that we assume  that $(G,\vecJ,\beta)\in \IntrstGraphFam(d,\varepsilon)$. This implies that $ |J_b| \leq 10\sqrt{\log n}$. 
Plugging this bound into the inequality above, for large $d$ and $n$,  we get that
\begin{align}\label{eq:Bound4WLAgreement}
\Pr[X_{t}(w_{\ell })= Y_{t}(w_{\ell})\ |\ X_{t}(w_2)= Y_{t}(w_2),\ \cF_2] &\geq n^{-1/(\log d)^4} \enspace. 
\end{align}

The claim follows by plugging \eqref{eq:Bound4W2Agreement} and \eqref{eq:Bound4WLAgreement} into \eqref{eq:claim:Base4V}. 
\end{proof}

\subsection{Proof of Lemma \ref{lem:ClacOfR}}\label{sec:lem:ClacOfR}
We show that for an arbitrary path $P$ in $\G$, with $r\le|P|\le \log n$, we have that 
\begin{align}\label{eq:SmallTuples}
\Pr \left[\cappedWA(P) \ge 1\right] \le {n^{-d^{1/4}}} \enspace.
\end{align}
In light of \eqref{eq:SmallTuples}, applying the union bound over all such paths gives us the result
\begin{align*}
\Pr \left[\cS_1\right] 
\ge 1 - \sum_{k = r}^{\log n} \binom{n}{k +1}\left(\frac{d}{n}\right)^k \cdot n^{-d^{1/4}} 
 \ge 1 - \sum_{k = r}^{\log n} {n}\cdot{d}^k \cdot n^{-d^{1/4}}
 \ge 1 - n^{-d^{1/5}},
\end{align*}
where the last inequality holds for large $d$ and $n$. Therefore, we now focus on proving ~\eqref{eq:SmallTuples}. 

Let $P= (v_0, \ldots, v_k)$, with $r \le k \le \log n$. We split the vertices of $P$ into two sets, $\In(P)$, and $\Out(P)$. The set $\In(P)$ is comprised by all vertices $v_i$ of $P$ that are adjacent to a vertex in $V(P) \setminus \{v_{i-1}, v_{i+1}\}$, and $\Out(P) = V(P) \setminus \In(P)$, so that 
\begin{align} \label{eq:BreakLamInTwo}
\cappedWA(P) = 
\cappedWA\left(\In(P)\right)\cdot 
\cappedWA(\Out(P))= \prod_{v\in \In(P)} \cappedWA(v) \cdot \prod_{w\in \Out(P)} \cappedWA(w) \enspace.
\end{align}
 We next bound separately $\cappedWA\left(\In(P)\right)$, and $\cappedWA\left(\Out(P)\right)$, so that their product is less than $1$, w.h.p..

Let us start by bounding $\cappedWA\left(\In(P)\right)$. We first notice the following tail bound for $|\In(P)|$:
\begin{align}\label{eq:InSetbound}
\Pr\left[|\In(P)| \ge 2\sqrt{d}\right]
\le 
\sum_{s = \sqrt{d}}^{k}
\binom{k^2}{s}
\cdot
\left( \frac{d}{n}\right)^s
\le 
\sum_{s = \sqrt{d}}^{k}
\left( \frac{d\cdot k^2}{n}\right)^s
\le
2\cdot\left( \frac{d\cdot k^2}{n}\right)^{\sqrt{d}}
\le
 n^{-{\sqrt{d}}/{2}} \enspace,
\end{align}
where the last two inequalities hold for large $d$ and $n$. 

For a vertex $v_i$ in $P$, let $\degr_{\rm in}(v_i)$ be the number of neighbors
that $v_i$ has in $P$, while let $\degr_{\rm out}(v_i)$ be the number of 
neighbors in $V\setminus P$.

For a  vertex $w\in \In(P)$ we have that $\degr_{\rm out}(w)$ is dominated by 
${\tt Binom}(n, d/n)$. Then, we obtain that 
\begin{align}\label{eq:InDegBound}
\Pr\left[\degr_{\rm out}(w)\ge d\cdot\log n, \;\text{ for at least one } w \in \In(P)\ \mid\ |\In(P)| <2\sqrt{d}\right] 
\le n^{-d^{{9}/{10}}}
\enspace .
\end{align}
Indeed,  the Chernoff bound gives
\begin{align*}
\Pr\left[\ \degr_{\rm out}(w)\ge d\cdot\log n\ |\  w\in \In(P)\right] 
\le \exp\left(-\frac{(\log n - 2)^2}{\log n} \cdot d\right)
\le \exp\left(-\frac{\log n}{2} \cdot d\right)
\le n^{-d^{95/100}}
\enspace , 
\end{align*}
where the last two inequalities hold for large enough $d$ and $n$. 
The above and a simple union-bound imply \eqref{eq:InDegBound}.

Let $B_1$ be the event that $|\In(P)| < 2\sqrt{d}$, while let 
$B_2$ be the event that for all $w\in \In(P)$ we have $\degr_{\rm out}(w)<d\log n$.

On the events $B_1$ and $B_2$, we have that for all $w\in \In(P)$
\begin{enumerate} 
\item $\degr(w)<d\log n+2+2\sqrt{d}$,
\item $\cappedWA(w)\leq 2d^2\log n$.
\end{enumerate}

The first item follows by noticing that $\degr(w)=\degr_{\rm in}(w)+\degr_{\rm out}(w)$, 
and that $\degr_{\rm in}(w)\leq 2+|\In(P)|$, for all $w\in \In(P)$.

The second item follows from the first item, and by noticing that 
for any vertex $v$ in the graph we have $\cappedWA(v) \leq d\cdot \degr(v)$. 
Hence, we have that 
\begin{align}
\Pr\left[\cappedWA(\In(P)) < (2d^2 \log n)^{2\sqrt{d}}\right] &\ge \Pr[B_1\cap B_2] \nonumber\\
&\geq  1-\Pr[\overline{B}_1]-\Pr[\overline{B}_2] & \mbox{[union bound]}\nonumber \\
&\geq 1-  n^{-d^{1/3}}\enspace .& \mbox{[from \cref{eq:InDegBound} \& \cref{eq:InSetbound}]} \label{eq:FinalBoundForInP}
\end{align}

To bound $\cappedWA(\Out(P))$, we follow the proof strategy of Lemma \ref{lem:MGFForcappedWB}. In particular, we partition $\Out(P)$ into two sets: the set of vertices with even index in $P$, and the set of vertices with odd index in $P$. Moreover, we let $\cappedWA_{\rm even}(\Out(P))$ be the product of weights over the set of even vertices in $\Out(P)$, and $\cappedWA_{\rm odd}(\Out(P))$ similarly. We claim that
\begin{align}\label{eq:CappedEvenBoundOut}
\Pr \left[\cappedWA_{\rm even}(\Out(P)) \ge (\log n)^{-d}\right] \le n ^{-d^{2/5}} \enspace.
\end{align}
Indeed, Markov's inequality yields that for every $t>0$ 
\begin{align}
\Pr \left[\cappedWA_{\rm even}(\Out(P)) \ge (\log n)^{-d}\right]
\le 
\frac
{\Exp\left[\cappedWA^t_{\rm even}(\Out(P))\right]}
{(\log n)^{-t\cdot d}} \enspace.
\end{align}
In the proof of Lemma \ref{lem:MGFForcappedWB}, we have that for $t= d^{95/100}$, and $s=d^{94/100}$ and arbitrary vertex $v$, we have that $\Exp\left[{\cappedWA^t\left(v\right)} \right] \le (1-\frac{\varepsilon}{4})^{s}$, (notice that although we prove this for the random tree construction, the arguments work precisely for $\G(n, d/n)$ as well). Acknowledging the fact that $\cappedWA_{\rm even}(\Out(P))$ is a product of i.i.d. random variables, we get that
\begin{align}
\Pr \left[\cappedWA_{\rm even}(\Out(P)) \ge (\log n)^{-d}\right] 
\le 
\frac
{\left(1-\varepsilon/4\right)^{ks/2}}
{(\log n)^{-t\cdot d}}
\le 
{(\log n)^{d^2}}
\cdot
{n^{-d^{42/100}}}
\le 
n^{-d^{2/5}}\enspace,
\end{align}
Notice that, by symmetry, 
\eqref{eq:CappedEvenBoundOut} holds for $\cappedWA_{\rm odd}(\Out(P))$ as well. Moreover, since $\cappedWA(\Out(P)) = \cappedWA_{\rm odd}(\Out(P))\cdot\cappedWA_{\rm even}(\Out(P))$, using the union bound further yields
\begin{align}\label{eq:FinalBoundForOutP}
\Pr \left[\cappedWA(\Out(P)) < (\log n)^{-2d}\right] \ge 1- 2\cdot n ^{-d^{2/5}} \ge 1-n^{-d^{1/3}} \enspace.
\end{align}
From the union bound, equations \eqref{eq:FinalBoundForInP} and \eqref{eq:FinalBoundForOutP} , and the fact that
$\cappedWA(P) = \cappedWA(\In(P))\cdot\cappedWA(\Out(P))$,
we have that
\begin{align*}
\Pr \left[\cappedWA(P) < 1\right] \ge 1- 2\cdot n ^{-d^{1/3}} \ge 1-n^{-d^{1/4}} \enspace,
\end{align*}
concluding the proof of Lemma \ref{lem:ClacOfR}.


\bibliographystyle{plainurl}
\bibliography{PapersSamplingGlassses}

\begin{thebibliography}{10}

\bibitem{OptasOghlan08}
Dimitris Achlioptas and Amin Coja{-}Oghlan.
\newblock Algorithmic barriers from phase transitions.
\newblock In {\em 49th Annual {IEEE} Symposium on Foundations of Computer
  Science, {FOCS} 2008}, pages 793--802. {IEEE} Computer Society, 2008.
\newblock URL: \url{https://doi.org/10.1109/FOCS.2008.11}.

\bibitem{alaoui2020algorithmic}
Ahmed~El Alaoui and Andrea Montanari.
\newblock Algorithmic thresholds in mean field spin glasses.
\newblock {\em arXiv preprint arXiv:2009.11481}, 2020.

\bibitem{OptElAlaoui}
Ahmed~El Alaoui, Andrea Montanari, and Mark Sellke.
\newblock {Optimization of mean-field spin glasses}.
\newblock {\em The Annals of Probability}, 49(6):2922 -- 2960, 2021.
\newblock URL: \url{https://doi.org/10.1214/21-AOP1519}.

\bibitem{AlaouiMSFOCS22}
Ahmed~El Alaoui, Andrea Montanari, and Mark Sellke.
\newblock Sampling from the sherrington-kirkpatrick gibbs measure via
  algorithmic stochastic localization.
\newblock In {\em 63rd {IEEE} Annual Symposium on Foundations of Computer
  Science, {FOCS} 2022}, pages 323--334. {IEEE}, 2022.

\bibitem{OptMCMCIS}
Nima Anari, Kuikui Liu, and Shayan~Oveis Gharan.
\newblock Spectral independence in high-dimensional expanders and applications
  to the hardcore model.
\newblock In {\em 61st {IEEE} Annual Symposium on Foundations of Computer
  Science, {FOCS} 2020}, pages 1319--1330. {IEEE}, 2020.

\bibitem{bauerschmidt2019very}
Roland Bauerschmidt and Thierry Bodineau.
\newblock A very simple proof of the lsi for high temperature spin systems.
\newblock {\em Journal of Functional Analysis}, 276(8):2582--2588, 2019.

\bibitem{BezakovaGGS22}
Ivona Bez{\'{a}}kov{\'{a}}, Andreas Galanis, Leslie~Ann Goldberg, and Daniel
  Stefankovic.
\newblock Fast sampling via spectral independence beyond bounded-degree graphs.
\newblock In {\em 49th International Colloquium on Automata, Languages, and
  Programming, {ICALP} 2022}, volume 229, pages 21:1--21:16, 2022.

\bibitem{bubley1997path}
Russ Bubley and Martin Dyer.
\newblock Path coupling: A technique for proving rapid mixing in markov chains.
\newblock In {\em Proceedings 38th Annual Symposium on Foundations of Computer
  Science}, pages 223--231. IEEE, 1997.

\bibitem{Chen0YZFOCS22}
Xiaoyu Chen, Weiming Feng, Yitong Yin, and Xinyuan Zhang.
\newblock Optimal mixing for two-state anti-ferromagnetic spin systems.
\newblock In {\em 63rd {IEEE} Annual Symposium on Foundations of Computer
  Science, {FOCS} 2022}, pages 588--599. {IEEE}, 2022.

\bibitem{ChenLVStoc21}
Zongchen Chen, Kuikui Liu, and Eric Vigoda.
\newblock Optimal mixing of glauber dynamics: entropy factorization via
  high-dimensional expansion.
\newblock In {\em {STOC} '21: 53rd Annual {ACM} {SIGACT} Symposium on Theory of
  Computing, Virtual Event, Italy, June 21-25, 2021}, pages 1537--1550. {ACM},
  2021.

\bibitem{ChManMo23}
Zongchen Chen, Nitya Mani, and Ankur Moitra.
\newblock From algorithms to connectivity and back: Finding a giant component
  in random \emph{k}-sat.
\newblock In {\em Proceedings of the 2023 {ACM-SIAM} Symposium on Discrete
  Algorithms, {SODA} 2023}, pages 3437--3470. {SIAM}, 2023.
\newblock URL: \url{https://doi.org/10.1137/1.9781611977554.ch132}.

\bibitem{COghlanEfth11}
Amin Coja{-}Oghlan and Charilaos Efthymiou.
\newblock On independent sets in random graphs.
\newblock In {\em Proceedings of the Twenty-Second Annual {ACM-SIAM} Symposium
  on Discrete Algorithms, {SODA} 2011}, pages 136--144. {SIAM}, 2011.
\newblock URL: \url{https://doi.org/10.1137/1.9781611973082.12}.

\bibitem{CoEfJKKCMI}
Amin Coja-Oghlan, Charilaos Efthymiou, Nor Jaafari, Mihyun Kang, and Tobias
  Kapetanopoulos.
\newblock {Charting the replica symmetric phase}.
\newblock {\em Communications in Mathematical Physics}, 359:603--698, 2018.

\bibitem{dyer2006randomly}
Martin Dyer, Abraham~D Flaxman, Alan~M Frieze, and Eric Vigoda.
\newblock Randomly coloring sparse random graphs with fewer colors than the
  maximum degree.
\newblock {\em Random Structures \& Algorithms}, 29(4):450--465, 2006.

\bibitem{DyerFriez10}
Martin~E. Dyer and Alan~M. Frieze.
\newblock Randomly coloring random graphs.
\newblock {\em Random Struct. Algorithms}, 36(3):251--272, 2010.
\newblock URL: \url{https://doi.org/10.1002/rsa.20286}.

\bibitem{EAIsingIntroWork}
Samuel~Frederick Edwards and Phil~W Anderson.
\newblock {Theory of spin glasses}.
\newblock {\em Journal of Physics F: Metal Physics}, 5(5):965, 1975.

\bibitem{EfthymiouDA14}
Charilaos Efthymiou.
\newblock {MCMC} sampling colourings and independent sets of \emph{G}(\emph{n,
  d}/\emph{n}) near uniqueness threshold.
\newblock In {\em Proceedings of the Twenty-Fifth Annual {ACM-SIAM} Symposium
  on Discrete Algorithms, {SODA} 2014}, pages 305--316. {SIAM}, 2014.

\bibitem{EfthICALP22}
Charilaos Efthymiou.
\newblock {On Sampling Symmetric Gibbs Distributions on Sparse Random Graphs
  and Hypergraphs}.
\newblock In {\em 49th International Colloquium on Automata, Languages, and
  Programming, {ICALP} 2022, July 4-8, 2022}, volume 229, pages 57:1--57:16.
  Schloss Dagstuhl - Leibniz-Zentrum f{\"{u}}r Informatik, 2022.

\bibitem{EftFeng23}
Charilaos Efthymiou and Weiming Feng.
\newblock On the mixing time of glauber dynamics for the hard-core and related
  models on g(n, d/n).
\newblock {\em CoRR}, abs/2302.06172, 2023.
\newblock URL: \url{https://doi.org/10.48550/arXiv.2302.06172}, \href
  {http://arxiv.org/abs/2302.06172} {\path{arXiv:2302.06172}}.

\bibitem{EfthymiouHSV18}
Charilaos Efthymiou, Thomas~P. Hayes, Daniel Stefankovic, and Eric Vigoda.
\newblock Sampling random colorings of sparse random graphs.
\newblock In {\em Proceedings of the Twenty-Ninth Annual {ACM-SIAM} Symposium
  on Discrete Algorithms, {SODA} 2018}, pages 1759--1771. {SIAM}, 2018.

\bibitem{efthymiou2018sampling}
Charilaos Efthymiou, Thomas~P Hayes, Daniel {\v{S}}tefankovi{\v{c}}, and Eric
  Vigoda.
\newblock Sampling random colorings of sparse random graphs.
\newblock In {\em Proceedings of the Twenty-Ninth Annual ACM-SIAM Symposium on
  Discrete Algorithms}, pages 1759--1771. SIAM, 2018.

\bibitem{EfthZamp2302}
Charilaos Efthymiou and Kostas Zampetakis.
\newblock Broadcasting with random matrices.
\newblock {\em CoRR}, abs/2302.11657, 2023.
\newblock \href {http://arxiv.org/abs/2302.11657} {\path{arXiv:2302.11657}},
  \href {https://doi.org/10.48550/arXiv.2302.11657}
  {\path{doi:10.48550/arXiv.2302.11657}}.

\bibitem{eldan2022spectral}
Ronen Eldan, Frederic Koehler, and Ofer Zeitouni.
\newblock A spectral condition for spectral gap: fast mixing in
  high-temperature ising models.
\newblock {\em Probability theory and related fields}, 182(3-4):1035--1051,
  2022.

\bibitem{franz2001exact}
Silvio Franz, Michele Leone, Federico Ricci-Tersenghi, and Riccardo Zecchina.
\newblock {Exact solutions for diluted spin glasses and optimization problems}.
\newblock {\em Physical review letters}, 87(12):127209, 2001.

\bibitem{GalStefVigJACM15}
Andreas Galanis, Daniel Stefankovic, and Eric Vigoda.
\newblock Inapproximability for antiferromagnetic spin systems in the tree
  nonuniqueness region.
\newblock {\em J. {ACM}}, 62(6):50:1--50:60, 2015.

\bibitem{GamarnikJWFOCS20}
David Gamarnik, Aukosh Jagannath, and Alexander~S. Wein.
\newblock Low-degree hardness of random optimization problems.
\newblock In {\em 61st {IEEE} Annual Symposium on Foundations of Computer
  Science, {FOCS} 2020}, pages 131--140. {IEEE}, 2020.
\newblock URL: \url{https://doi.org/10.1109/FOCS46700.2020.00021}.

\bibitem{GamSudan17}
David Gamarnik and Madhu Sudan.
\newblock Performance of sequential local algorithms for the random {NAE-K-SAT}
  problem.
\newblock {\em {SIAM} J. Comput.}, 46(2):590--619, 2017.

\bibitem{guerra2004high}
Francesco Guerra and Fabio~Lucio Toninelli.
\newblock {The high temperature region of the Viana-Bray diluted spin glass
  model}.
\newblock {\em Journal of statistical physics}, 115:531--555, 2004.

\bibitem{KoehLRColt22}
Frederic Koehler, Holden Lee, and Andrej Risteski.
\newblock Sampling approximately low-rank ising models: {MCMC} meets
  variational methods.
\newblock In {\em Conference on Learning Theory, 2022}, volume 178 of {\em
  Proceedings of Machine Learning Research}, pages 4945--4988. {PMLR}, 2022.

\bibitem{martinelli1999lectures}
Fabio Martinelli.
\newblock Lectures on glauber dynamics for discrete spin models.
\newblock {\em Lectures on probability theory and statistics (Saint-Flour,
  1997)}, 1717:93--191, 1999.

\bibitem{mezard1990spin}
M.~M\'ezard, G.~Parisi, and M.~Virasoro.
\newblock {\em Spin glass theory and beyond}.
\newblock World Scientific, 1987.

\bibitem{mossel2010gibbs}
Elchanan Mossel and Allan Sly.
\newblock Gibbs rapidly samples colorings of g (n, d/n).
\newblock {\em Probability theory and related fields}, 148(1-2):37--69, 2010.

\bibitem{panchenko2013parisi}
Dmitry Panchenko.
\newblock The parisi ultrametricity conjecture.
\newblock {\em Annals of Mathematics}, pages 383--393, 2013.

\bibitem{PanchenkoTalagrand}
Dmitry Panchenko and Michel Talagrand.
\newblock Bounds for diluted mean-fields spin glass models.
\newblock {\em Probability Theory and Related Fields}, 130(3):319--336, 2004.
\newblock \href {https://doi.org/10.1007/s00440-004-0342-2}
  {\path{doi:10.1007/s00440-004-0342-2}}.

\bibitem{RSBParisi}
Giorgio Parisi.
\newblock {Infinite number of order parameters for spin-glasses}.
\newblock {\em Physical Review Letters}, 43(23):1754, 1979.

\bibitem{SKModel}
David Sherrington and Scott Kirkpatrick.
\newblock {Solvable model of a spin-glass}.
\newblock {\em Physical review letters}, 35(26):1792, 1975.

\bibitem{SlySun12}
Allan Sly and Nike Sun.
\newblock The computational hardness of counting in two-spin models on
  d-regular graphs.
\newblock In {\em 53rd Annual {IEEE} Symposium on Foundations of Computer
  Science, {FOCS} 2012}, pages 361--369. {IEEE} Computer Society, 2012.
\newblock URL: \url{https://doi.org/10.1109/FOCS.2012.56}.

\bibitem{SteinNewmanSpinGlassBook}
Daniel~L Stein and Charles~M Newman.
\newblock {\em {Spin glasses and complexity}}, volume~4.
\newblock Princeton University Press, 2013.

\bibitem{TalagrandAnnals}
Michel Talagrand.
\newblock The {Parisi} formula.
\newblock {\em Ann. Math. (2)}, 163(1):221--263, 2006.
\newblock \href {https://doi.org/10.4007/annals.2006.163.221}
  {\path{doi:10.4007/annals.2006.163.221}}.

\end{thebibliography}

\appendix

\section{Tail Bound for sums of half-normal}\label{sec:HalfNormalTails}

We say that $Y$ follows the \emph{half-normal} distribution with parameter $\sigma$, if $Y=|X|$, where $X$ follows the Gaussian distribution $\mathcal{N}(0,\sigma^2)$. For $N>0$ integer, and $\sigma\ge 0$, let $X_1, \ldots, X_N \sim \mathcal{N}(0,\sigma^2)$ be i.i.d standard Gaussians, and write 
\begin{equation*}
X= \sum_{i=1}^N |X_i| \enspace.
\end{equation*}
We show the following concentration bound for $X$
\begin{theorem}\label{thm:UpperTailBoundSumOfHalfNormRelative}
For every $\delta \ge 0$, we have that
\begin{equation*}
\Pr\left[\,X > (1+\delta)\cdot\Exp[X]\,\right]\leq\exp\left(-N\cdot\frac{\delta^{2}}{\pi}\right) \enspace.
\end{equation*}
\end{theorem}

\begin{proof}
For every real $t \ge 0$, consider the moment generating function $\Exp[\exp(t\cdot X)]$, as a non-negative random variable. Applying Markov's inequality on $\Exp[\exp(t\cdot X)]$, we get that for every $a \in \mathbb{R}$
\begin{equation}\label{eq:markov}
\Pr\left[X > a\right] = \Pr\left[\exp(t\cdot X) > \exp(t\cdot a)\right] \le \frac{\Exp[\exp(t\cdot X)]}{\exp(t\cdot a)} \enspace.
\end{equation}
To calculate $\Exp[\exp(t\cdot X)]$, we first observe that since $X_i$'s are i.i.d., so that
\begin{align}\label{eq:ExpOfProdMGF}
\Exp[\exp(t\cdot X)] = \Exp\left[\exp\left(t\cdot \sum_{i=1}^N |X_i|\right)\right] = \left(\Exp\left[\exp\left(t\cdot|X_1|\right)\right] \right)^N \enspace.
\end{align}
We now focus on $\Exp\left[\exp\left(t\cdot|X_1|\right)\right] $. We have 
\begin{align}
\Exp\left[\exp\left(t\cdot|X_1|\right)\right]
	&= 
\int_{-\infty}^{+\infty}\exp\left(t\cdot|x|\right) \frac{1}{\sigma\sqrt{2\pi}} \exp\left(-\frac{x^2}{2\sigma^2}\right) dx\nonumber \\
	&=2
\int_{0}^{+\infty}\frac{1}{\sigma\sqrt{2\pi}}\exp\left(tx-\frac{x^2}{2\sigma^2}\right) dx\nonumber\\
	&=2\cdot\exp\left(\frac{t^2\sigma^2}{2}\right) \cdot
\int_{0}^{+\infty}\frac{1}{\sigma\sqrt{2\pi}}\exp\left(-\frac{t^2\sigma^2}{2}+tx-\frac{x^2}{2\sigma^2}\right) dx\nonumber\\
	&=2\cdot\exp\left(\frac{t^2\sigma^2}{2}\right) \cdot
\int_{0}^{+\infty}\frac{1}{\sigma\sqrt{2\pi}}\exp\left(-\frac{1}{2}\left(\frac{x}{\sigma}-\sigma t\right)^2\right) dx\nonumber\\
	&=2\cdot\exp\left(\frac{t^2\sigma^2}{2}\right) \cdot 
\mathrm{\Phi}\left({\sigma t}\right) \label{eq:thatsMGF}
\enspace.
\end{align}
Hence, from \eqref{eq:ExpOfProdMGF} and \eqref{eq:markov}, we further have that
\begin{align}\label{eq:LastEquality}
\Pr\left[X > a\right] \le \frac{\Exp[\exp(t\cdot X)]}{\exp(t\cdot a)} =  \exp\left(N\cdot\frac{t^2\sigma^2}{2} - t\cdot a+N\cdot\log
\left[2\cdot\mathrm{\Phi}\left(\sigma t\right)\right] \right)\enspace.
\end{align}
As we prove in Lemma~\ref{lem:BoundPhi}, for every $x\geq 0$, we have
$
\mathrm{\Phi}(x)\leq \frac{1}{2} +\frac{x}{\sqrt{2 \pi}} $.
Therefore, 
\begin{align}
\Pr\left[X > a\right]
	&\le 
 \exp\left(N\cdot\frac{t^2\sigma^2}{2} - t\cdot a+N\cdot\log
\left[1+\sigma t\cdot\sqrt{\frac{2}{\pi}}\right] \right)\nonumber\\
	&\le 
	\exp\left(N\cdot\frac{t^2\sigma^2}{2} - t\cdot a+N\cdot\sigma t\cdot\sqrt{\frac{2}{\pi}}\right)\nonumber\\
	&=
	\exp\left(N\left[t^2\cdot\frac{\sigma^2}{2} + t\cdot\left(\sigma \cdot\sqrt{\frac{2}{\pi}}-\frac{a}{N}\right)\right]\right) \label{eq:noMinimize}
\enspace. 
\end{align}
Minimizing the exponent in the rhs of \eqref{eq:noMinimize}  with respect to $t \geq 0$, yields that for every 
\begin{equation}\label{eq:Alphacond}
\textstyle
a \geq 
N \sigma \cdot \sqrt{\frac{2}{\pi}} \enspace,
\end{equation}
we have
\begin{equation}\label{eq:MMRwrite}
\Pr\left[X > a\right]\leq\exp\left(N\left[-\frac{ \left(a \sqrt{\pi}-\sigma N\sqrt{2} \right)^{2}}{2 \pi N^{2} \sigma^{2}}\right]\right) \enspace.
\end{equation}
Considering the derivative at $t=0$ of the moment generating function of $|X_1|$ in \eqref{eq:thatsMGF}, it is easy to check that 
\begin{equation*}
\Exp[X] 
	= \Exp\left[\sum\nolimits_{i=1}^N|X_i| \right]
	=  N\cdot\Exp\left[|X_1| \right]
	= N \sigma \cdot \sqrt{\frac{2}{\pi}}\enspace.
\end{equation*}
So that condition \eqref{eq:Alphacond}, can be written as $a\ge \Exp[X]$, or equivalently, $a=(1+\delta)\cdot\Exp[X]$ for some $\delta \ge 0$. Substituting $a=(1+\delta)\cdot\Exp[X]$ in \eqref{eq:MMRwrite}, gives
\begin{align*}
\Pr\left[\,X > (1+\delta)\cdot\Exp[X]\,\right]
\le
	\exp\left(-N\cdot\frac{\delta^{2}}{\pi}\right) \label{eq:RREwrite}
\enspace,
\end{align*}
as desired.
\end{proof}

\section{Linear Approximation Of the Gaussian CDF}

\begin{lemma}\label{lem:BoundPhi}
Let $\mathrm{\Phi}: \mathbb{R} \to [0,1]$ be the CDF of the Standard Gaussian distribution. Then, for every $x \ge 0$
\begin{equation*}
\mathrm{\Phi}(x)\leq \frac{1}{2} +\frac{x}{\sqrt{2 \pi}} \enspace.
\end{equation*}
\end{lemma}

\begin{proof}
Define $H: [0,+\infty) \to \mathbb{R}$ with 
\begin{equation*}
H(x) = \frac{1}{2} +\frac{x}{\sqrt{2 \pi}} -\mathrm{\Phi}(x)\enspace.
\end{equation*}
Differentiating gives
\begin{align*}
H^\prime(x) = \frac{1}{\sqrt{2 \pi}} -\frac{\exp(-\frac{x^2}{2})}{\sqrt{2\pi}} \ge 0\enspace,
\end{align*}
so that $H$ is increasing. Noticing also that $H(0) = 0$, yields $H \ge 0$, as desired.
\end{proof}

\section{Cycles in \texorpdfstring{$\G(n, d/n)$}{}
- Proof of Lemma \ref{lem:SamllUnicyclicGnp}}\label{sec:lem:SamllUnicyclicGnp}

Write $\cE_S$ for the event that every set of vertices $S \subseteq V(\G)$, with cardinality at most $2\frac{\log n}{\log^2  d}$, spans at most $|S|$ edges in $\G$. We have

\begin{eqnarray}
\Pr\Big[ \; \overline{\cE_S} \;\Big]
&\leq 
&\sum_{k=1}^{2\frac{\log n}{ (\log  d)^2}} \binom{n}{k} \binom{\binom{k}{2}}{k+1}\left(\frac{d}{n}\right)^{k+1}
\leq \sum_{k=1}^{2\frac{\log n}{ (\log  d)^2}} \left(\frac{ne}{k}\right)^k\left(\frac{k^2e}{2(k+1)}\right)^{k+1}\left(\frac{d}{n}\right)^{k+1}
 \nonumber\\
&\leq& \frac{1}{n}\sum_{k=1}^{2\frac{\log n}{ (\log  d)^2}}\left(\frac{ekd}{2}\right)\left(\frac{e^2d}{2}\right)^{k}
\leq \frac{ed}{ (\log  d)^2}\ \frac{\log n}{n}\sum_{k=1}^{2\frac{\log n}{ (\log  d)^2}}\left(\frac{e^2d}{2}\right)^{k}
\nonumber \\
&\leq& n^{-9/10}\left({e^2d}/{2}\right)^{2\frac{\log n}{ (\log  d)^2}}
\leq n^{-3/4} \enspace, \nonumber
\end{eqnarray}
where the last two inequalities hold for sufficiently large $d$.

\end{document}